\pdfoutput=1
\documentclass[10pt,twocolumn]{article}

\usepackage[paperwidth=210mm,paperheight=297mm,hmargin=18.1mm,vmargin={20mm,22mm},columnsep=18pt]{geometry}
\usepackage[T1]{fontenc}
\usepackage{stix}
\usepackage{amsmath,amsthm}
\usepackage{graphicx}
\usepackage{array}
\usepackage{booktabs}
\usepackage[table]{xcolor}
\usepackage[font=small,labelfont=bf,labelsep=period]{caption}
\usepackage{algorithm}
\usepackage{algorithmic}
\usepackage{placeins}
\usepackage[numbers,sort&compress]{natbib}
\usepackage[colorlinks=true,linkcolor=blue!60!black,citecolor=blue!60!black,urlcolor=blue!60!black]{hyperref}

\providecommand{\ket}[1]{\lvert #1\rangle}
\providecommand{\bra}[1]{\langle #1\rvert}
\providecommand{\braket}[1]{\langle #1\rangle}

\newcommand{\tblwidth}{\linewidth}
\newcolumntype{L}{@{\extracolsep{\fill}}l}
\newcolumntype{R}{@{\extracolsep{\fill}}r}
\newcolumntype{C}{@{\extracolsep{\fill}}c}

\newtheorem{theorem}{Theorem}
\newtheorem{proposition}{Proposition}
\newtheorem{lemma}{Lemma}[section]

\hypersetup{pdftitle={Fully Quantum Latent Diffusion for Synthetic Malware Records: Effective Denoising, Limited Class Information, and Costly Post-Selection}, pdfauthor={Ban Q. Tran, Quang Duy Tran, Anh H. Nguyen, Ziwen Pan, Tommy Dang, Susan Mengel}}

\begin{document}

\twocolumn[{%
\begin{center}
{\LARGE\bfseries Fully Quantum Latent Diffusion for Synthetic Malware Records: Effective Denoising, Limited Class Information, and Costly Post-Selection\par}
\vspace{12pt}
{\large Ban Q. Tran\textsuperscript{1,2}, Quang Duy Tran\textsuperscript{1}, Anh H. Nguyen\textsuperscript{1}, Ziwen Pan\textsuperscript{1}, Tommy Dang\textsuperscript{1}, and Susan Mengel\textsuperscript{1}\par}
\vspace{8pt}
{\small \textsuperscript{1}Department of Computer Science, Texas Tech University, Lubbock, TX 79416, USA\\
\textsuperscript{2}Department of Computing Fundamentals, FPT University, Hanoi, Vietnam\par}
\end{center}
\vspace{6pt}
\begin{list}{}{\setlength{\leftmargin}{12mm}\setlength{\rightmargin}{12mm}}\item\relax
\small\textbf{Abstract.} Malware benchmarks are difficult to share, because the samples are harmful and a corpus can reveal the incidents of its holder, while synthetic substitutes can leak training records. Quantum diffusion models, however, have been evaluated only on images or quantum states, leaving the utility, privacy, and sampling cost of their samples unknown. In this research, we introduce a fully quantum latent diffusion model for malware records, in which a parameterized circuit with one mid-circuit measurement denoises amplitude-encoded latent vectors of opcode sequences without classical neural layers. On 17 malware families with ten seeds, we compared it with random states, a perfect generator, a parameter-matched classical control, and a classical latent diffusion model. The circuit lowers its denoising infidelity by 63\%, yet classifiers trained on its records reach a macro-F1 score of 0.11 on real records, against 0.08 for random states, 0.61 for a perfect generator, and 0.44 for a classical diffusion model of comparable size, and adding its records to real training data lowers the score from 0.635 to 0.593. Post-selected sampling accepts $3.0\times10^{-7}$ of runs, and a training penalty that raises the acceptance to 0.14 removes most of the denoising. No synthetic record copies a training record, and smoothing the decoded lengths lowers the area under the curve of a length-based membership test from 0.80 to 0.51, which a differentially private length histogram matches with a formal guarantee. These results indicate that quantum generators of shareable security data require a joint evaluation of utility, privacy, and sampling cost.

\end{list}
\vspace{14pt}
}]

\section{Introduction}\label{sec:intro}

Benchmark datasets of malware underpin the development and comparison of learning-based detectors, yet they are difficult to share: the samples are harmful programs, and their collection often depends on sources whose terms restrict redistribution. Synthetic records produced by generative models offer a substitute that research groups could exchange instead of the programs themselves. Diffusion models \citep{ho2020denoising} have become the leading generators of such data, from images in latent spaces \citep{rombach2022high} to tabular records \citep{kotelnikov2023tabddpm} and embeddings of malware programs \citep{bao2025generating}. However, a synthetic dataset protects its sources only if the generator does not reproduce them. Diffusion models can memorize and regenerate individual training samples \citep{carlini2023extracting}, and synthetic data in general either fail to prevent inference attacks or lose much of their utility \citep{stadler2022synthetic}. A release of synthetic security data therefore requires evidence on both utility and privacy.

Quantum machine learning (QML) offers a different class of generators, in which parameterized quantum circuits play the role of neural networks \citep{benedetti2019parameterized}. In our systematic review of quantum deep learning in the noisy intermediate-scale quantum (NISQ) era \citep{tran2026systematic}, we examined how such models have been designed and evaluated under the limited qubit counts and circuit depths of current devices \citep{preskill2018quantum}. These limits also shape the performance of quantum neural networks (QNNs): in our comparative analysis of three hybrid QNN architectures \citep{tran2026comprehensive}, models that performed well on datasets with few features degraded markedly on high-dimensional data, and their robustness to quantum noise differed considerably between architectures. A quantum generator for malware records should therefore operate on compact representations, which latent diffusion provides by moving the generative process into the low-dimensional space of a pre-trained autoencoder \citep{rombach2022high}.

Quantum diffusion models can be broadly categorized into two families. In hybrid models, variational circuits replace components of a classical denoiser, such as the layers of a U-Net \citep{kolle2024quantum}, or a classical network processes the measured state between the steps \citep{cacioppo2025hybrid}. Since the classical network of such a model can carry most of the generative task, the contribution of its circuits is difficult to isolate. In fully quantum models, the circuit itself maps noisy quantum states onto denoised ones, with ancilla measurements that provide the required non-unitary step \citep{zhang2024generative,cacioppo2023quantum,cacioppo2025hybrid}. These models rest on two expectations: since $n$ qubits hold $2^{n}$ amplitudes, a circuit should describe the data distribution with far fewer trainable parameters than a classical denoiser, and the generative process should run on a quantum device from the noisy state to the denoised one, without classical layers between the steps. However, these models share two limitations with respect to the release of security data. First, they have been evaluated on images or on quantum states with measures of sample quality, so that the utility of their samples for downstream classifiers and the privacy risks of their release remain unknown. Second, fully quantum denoisers depend on mid-circuit measurements and post-selection, whose cost has not been quantified for the circuits that training produces. Consequently, it remains unclear whether a fully quantum diffusion model can generate useful and private security records at a sampling cost that quantum hardware can afford.

In this research, we introduce a fully quantum latent diffusion model for malware records and evaluate it against these requirements, which tests the expectations above directly: the circuit performs every denoising step alone, so that its contribution can be measured in isolation, at a known parameter count and a measurable cost of sampling on quantum hardware. A pre-trained convolutional autoencoder maps every opcode sequence onto a latent vector of $32$ features, which we store in the amplitudes of five qubits, and a reverse-bottleneck circuit with one ancilla measurement performs every denoising step without any classical neural layer. Our contributions are summarized as follows:
\begin{itemize}
\item We formulate fully quantum latent diffusion for malware records, with two amplitude codecs, a family-conditioned circuit trained toward the posterior mean of each step, and a decoder that restores the length removed by the encoding, together with a threat model for the release of synthetic security data (Section~\ref{sec:method}). To the best of our knowledge, this is the first application of a quantum diffusion model to security data.
\item We establish a theorem and two propositions on the training target, on the initialization and the acceptance probability that it implies, and on the lengths of decoded records, complemented by two lemmas on the smoothness of training and the shot complexity of the read-out, and we verify every statement numerically (Section~\ref{sec:theory} and Appendices~\ref{app:proofs} and~\ref{app:numerics}).
\item We identify a privacy leak in the decoder, which reuses the exact lengths of training records and supports a membership test with an area under the receiver operating characteristic curve (AUC) of $0.80$, and we close it with a smoothed bootstrap of the lengths that requires no retraining and with a differentially private estimate of the length distribution that gives a formal guarantee (Sections~\ref{sec:encoding} and~\ref{sec:res_privacy}).
\item We evaluate the quantum denoiser on all $17$ families with ten seeds and paired tests, in terms of utility, distribution fidelity, memorization, membership, sampling cost, and hardware noise, against references that share its encoding, its decoder, and its sampling procedure, namely a parameter-matched classical control, the random states from which sampling starts, and the round-trip ceiling of a perfect generator, and we place the complete framework against a classical latent DDPM in the design of \citet{bao2025generating}. We test the explanations of the limited utility with stochastic samplers, a weighting of the records by their acceptance probability, a training objective that penalizes rejection, and diagnostics on the initial states and the gradients (Sections~\ref{sec:results} and~\ref{sec:discussion}).
\end{itemize}

Our evaluation clarifies what the current design achieves and where it requires further development. The post-selected circuit learns its objective, lowering the infidelity of its denoising steps by $63\%$ relative to the identity map, and it moves the generated states toward the data, which lowers the squared MMD of its records by $20\%$ relative to the random states from which sampling starts. Its records never copy training data, and a smoothed bootstrap of the decoded lengths closes a channel through which the decoder supports a membership test with an AUC of $0.80$. However, the class information of its records remains close to that of the random states: classifiers trained on its records and tested on real records reach a macro-F1 score of $0.114$, against $0.076$ for the random states and a ceiling of $0.605$ for a perfect generator, whereas a parameter-matched classical control reaches $0.367$ with the sign-split codec and a classical latent DDPM with $11{,}404$ parameters reaches $0.442$, and adding the records of the circuit to the real training data lowers the macro-F1 score from $0.635$ to $0.593$. After training, the probability of accepting a post-selected sample falls to $3.0\times10^{-7}$. A penalty on rejected runs, trained with three of the seeds, raises this probability from $1.6\times10^{-7}$ to $0.14$ but removes most of the denoising. These findings identify the encoding of the families and the competition between denoising and acceptance under post-selection, more than the parameter budget or the sampler, as the main directions for the development of fully quantum diffusion for security data.

The remainder of this paper is organized as follows. Section~\ref{sec:related} reviews related work on synthetic security data, diffusion models, and quantum generative models. Section~\ref{sec:method} presents the quantum latent diffusion model and its theoretical analysis, and Section~\ref{sec:setup} describes the experimental setup. Section~\ref{sec:results} reports the results, Section~\ref{sec:discussion} discusses their implications and the limitations of our study, and Section~\ref{sec:conclusion} concludes the paper. Appendix~\ref{app:proofs} contains the proofs, and Appendix~\ref{app:numerics} contains their numerical verification.

\section{Related work}\label{sec:related}

In this section, we review three lines of work on which our study builds: the generation of synthetic data for security research together with the assessment of its privacy, diffusion models, and quantum generative models, including the quantum diffusion models that Table~\ref{tab:related} compares with ours.

\subsection{Synthetic data for security research}\label{sec:rel_synthetic}

Machine learning for malware analysis depends on large labeled corpora of programs, which are difficult to share, since the programs themselves are harmful and their collection often involves restricted sources. Public benchmarks have therefore released derived representations, such as the static features of $1.1$ million Windows executables in EMBER \citep{anderson2018ember}, and opcode sequences obtained by disassembling executables have long served as a representation for the detection of unknown malware \citep{santos2013opcode}.

Generative models offer a second way to share such data, and work on them has followed two directions. The first direction builds the generators. Generative adversarial networks (GANs) \citep{goodfellow2014generative} have served attackers to craft malware features that evade black-box detectors \citep{hu2017generating}, and conditional GANs \citep{xu2019modeling} and diffusion models \citep{kotelnikov2023tabddpm,jiang2024netdiffusion} have become standard generators of synthetic tabular records and network traffic. Generators of malware data have also operated on opcode sequences, for instance through word embeddings of the opcodes and a Wasserstein GAN \citep{tran2022word}, and diffusion models have recently followed: \citet{bao2025generating} trained a denoising diffusion probabilistic model (DDPM) \citep{ho2020denoising} with a one-dimensional U-Net on per-program embeddings of the opcode sequences of $25$ malware families and reported that its synthetic embeddings improved the classification of small families more than those of GANs, \citet{camerota2024addressing} augmented images of IoT malware traffic with a DDPM, and \citet{paim2025maldatagen} included a latent diffusion model among the generators of tabular Android malware records. These studies evaluate the generated records by their utility for classifiers and do not measure what the records reveal about their training data. The second direction assesses what the generated data reveal. Empirical audits test whether an adversary can decide that a record belonged to the training set \citep{shokri2017membership}, measure the distance of synthetic records to their closest training records \citep{park2018data}, or extract training samples from the generator \citep{carlini2023extracting}, whereas differentially private training \citep{abadi2016deep} bounds the influence of every record at a cost in utility. Generators of synthetic data have adopted such guarantees, for instance GANs trained with private aggregation of teacher ensembles \citep{jordon2019pate} or with differentially private gradients for tabular records \citep{fang2022dpctgan}, and diffusion models trained with differentially private gradients \citep{dockhorn2023differentially}, whose sample quality decreases as the privacy budget shrinks. Both directions share a limitation that \citet{stadler2022synthetic} made explicit: synthetic data either failed to prevent inference attacks or lost much of their utility, so that privacy and utility cannot be assessed separately. These assessments concern the trained generator, while the deterministic processing that maps generated vectors back to records has received less attention. In contrast, we assess utility and privacy jointly (Section~\ref{sec:measures}), and we analyze the decoder of our model as a privacy channel of its own, which reveals the lengths of training records unless they are smoothed (Proposition~\ref{prop:length}).

\subsection{Diffusion models}\label{sec:rel_diffusion}

Diffusion models learn to reverse a gradual noising process \citep{sohl2015deep}, and DDPMs \citep{ho2020denoising} have made this approach practical by training a network to predict the noise of each step with a regression objective. Subsequent work has mainly reduced the cost of sampling and of the denoiser. Implicit samplers \citep{song2021denoising} replace the stochastic chain of about $1{,}000$ steps with a deterministic chain of far fewer steps, and latent diffusion \citep{rombach2022high} moves the process into the low-dimensional space of a pre-trained autoencoder. Nevertheless, classical denoisers typically remain networks with millions of parameters that are evaluated at every step, and they can memorize their training samples \citep{carlini2023extracting}. In contrast, our denoiser is a circuit of $11{,}520$ angles that operates in a latent space of $32$ features and samples with ten deterministic steps, a combination whose consequences for the diversity of the samples we analyze in Section~\ref{sec:disc_learning}.

\subsection{Quantum generative models}\label{sec:rel_quantum}

Parameterized quantum circuits \citep{benedetti2019parameterized} have served as generative models either by sampling from their measurement distribution \citep{benedetti2019generative} or by playing the generator of an adversarial game \citep{lloyd2018quantum,dallairedemers2018quantum}, which can also load probability distributions into quantum states \citep{zoufal2019quantum}. These models share a limitation of training: their gradients can vanish exponentially with the number of qubits \citep{mcclean2018barren}, so that most quantum generators have remained small. Diffusion offers a way around this limitation, since it splits generation into a sequence of short denoising steps, each of which a shallow circuit can learn \citep{zhang2024generative}.

\begin{table*}[t]
\small
\caption{Quantum diffusion models compared with this work. Quantum part: the component that the circuit implements. Circuit output: the state that the circuit passes to the next step after its measurement, or the expectation values that a classical network receives. Latent: whether the diffusion runs in the latent space of an autoencoder. FID: Fr\'echet inception distance; WaM: Wasserstein distance between Gaussian mixtures; SSIM and PSNR: structural similarity and peak signal-to-noise ratio; AUC: area under the receiver operating characteristic curve.}\label{tab:related}
\footnotesize
\begin{tabular*}{\tblwidth}{@{}>{\raggedright\arraybackslash}p{2.6cm}>{\raggedright\arraybackslash}p{2.0cm}>{\raggedright\arraybackslash}p{2.4cm}>{\raggedright\arraybackslash}p{2.4cm}>{\raggedright\arraybackslash}p{1.3cm}>{\raggedright\arraybackslash}p{3.5cm}@{}}
\toprule
Model & Quantum part & Circuit output & Data & Latent & Evaluation\\
\midrule
QuDDPM \citep{zhang2024generative} & Denoiser & State after ancilla measurement & Quantum state ensembles & No & Distances between state ensembles\\
QDM \citep{cacioppo2023quantum} & Denoiser & State after ancilla measurement & MNIST & Pixels and latent & FID, WaM, classifier AUC; sampling on an IBM processor\\
Hardware-oriented and hybrid QDM \citep{cacioppo2025hybrid} & Denoiser, partly driven by a classical network & State after ancilla measurement; measured amplitudes (hybrid) & MNIST digits 0 and 1 & Pixels and latent & FID; sampling on an IBM processor\\
Quantum U-Net \citep{kolle2024quantum} & U-Net layers & Measurement outcomes & MNIST, Fashion-MNIST, CIFAR-10 & No & FID, SSIM, PSNR\\
QLDM \citep{li2026quantum} & Autoencoder & Expectation values & Quantum states & Yes & Fidelity, entanglement entropy\\
This work & Denoiser & Post-selected state & Malware latent vectors & Yes & Classifier utility, MMD, memorization, membership, acceptance probability, hardware noise\\
\bottomrule
\end{tabular*}
\end{table*}

Quantum diffusion models differ mainly in where the circuit acts and in how it realizes the non-unitary map that denoising requires, and Table~\ref{tab:related} compares them along these axes. In fully quantum models, the circuit acts on quantum states and obtains the non-unitary step from mid-circuit measurements of ancilla qubits, after which the measured state enters the next step. QuDDPM \citep{zhang2024generative} and its mixed-state extension \citep{kwun2025mixed} have generated ensembles of quantum states in this way, \citet{parigi2024quantum} proposed to exploit the noise of quantum processors as the diffusion mechanism itself, and the quantum diffusion models of \citet{cacioppo2023quantum,cacioppo2025hybrid} have generated images from amplitude-encoded pixels and from the latent vectors of a classical autoencoder, with sampling executed on an IBM processor. In hybrid models, a classical network processes the outputs of the circuits, for instance when quantum layers form part of a U-Net \citep{defalco2024quantum,kolle2024quantum} or denoise the latent features of a classical autoencoder \citep{defalco2024latent}; concurrently with our work, \citet{li2026quantum} placed the circuits in a quantum autoencoder and kept a classical denoiser for the generation of high-dimensional quantum states.

These lines share three limitations with respect to the release of security data. First, their evaluations measure the quality of images or the distance between ensembles of quantum states; \citet{cacioppo2023quantum} also checked conditional samples with classifiers trained on real images, but no study has trained classifiers on the generated samples or assessed what the samples reveal about their training data. Second, fully quantum models have avoided branch selection on hardware, because the probability of accepting a run decreases exponentially with the number of selected measurements \citep{cacioppo2023quantum,cacioppo2025hybrid}, but the acceptance probability of a trained post-selected circuit along its sampling trajectories has not been reported. Third, in hybrid models the classical network carries the state between steps, so that the contribution of the circuits is difficult to separate from that of the network.

Our model builds on the latent, fully quantum design of \citet{cacioppo2023quantum} and differs from it in three respects. First, we train the circuit toward the posterior mean of each step instead of a single draw of the previous state, which removes the variance of the target without changing the optimal denoiser (Theorem~\ref{thm:target}), and we keep the post-selected branch during training and sampling, whose acceptance probability we report (Section~\ref{sec:res_training}). Second, family-specific first blocks condition the circuit on the malware family instead of additional label qubits, and a sign-split codec separates the families in the Hilbert space (Section~\ref{sec:encoding}). Third, we evaluate the generated records with the utility and privacy measures of synthetic security data, against a parameter-matched classical control and the bounds that the encoding sets on any generator, and under hardware noise. Quantum machine learning has also entered security applications, for instance with hybrid and fully quantum networks that classify images of Android applications \citep{mercaldo2022towards} and with a hybrid quantum-classical Wasserstein GAN that detects fraudulent credit card transactions as anomalies \citep{herr2021anomaly}; to the best of our knowledge, however, no previous study has applied a quantum diffusion model to security records.

\section{Methodology}\label{sec:method}

In this section, we introduce the quantum latent diffusion model that we employ to generate synthetic security records. Figure~\ref{fig:pipeline} illustrates the complete pipeline. A pre-trained encoder maps every record onto a latent vector, which we store in the amplitudes of a quantum state. We corrupt the encoded states with a classical forward diffusion process and train a parameterized quantum circuit to reverse this process one step at a time with an infidelity objective. To generate a record, the trained circuit denoises a random quantum state, and a classical decoder maps the final state back to the latent space. The denoiser is the only quantum component of the pipeline, and it contains no classical neural layer between its noisy input state and its denoised output state. We formulate the problem and describe the latent representation in Section~\ref{sec:latent}, state the threat model in Section~\ref{sec:threat}, and present the amplitude encoding and the forward process in Sections~\ref{sec:encoding} and \ref{sec:forward}. Sections~\ref{sec:denoiser} to \ref{sec:sampling} describe the quantum denoiser, its training, and the sampling procedure. Sections~\ref{sec:noise} to \ref{sec:impl} cover the simulation of hardware noise, the parameter-matched classical control, the evaluation measures, and the implementation. Finally, Section~\ref{sec:theory} establishes the theoretical properties of the model, whose proofs are given in Appendix~\ref{app:proofs}.

\begin{figure*}[t]
\centering
\includegraphics{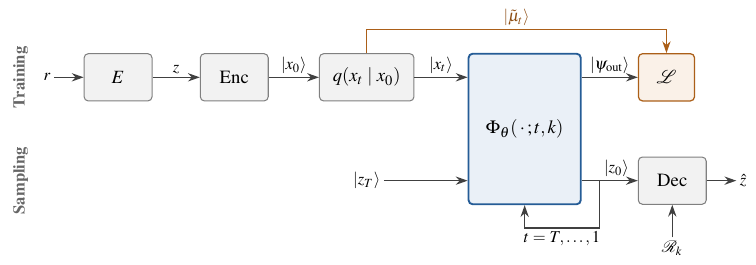}
\caption{Training and sampling pipeline of the quantum latent diffusion model. Training (top): the encoder $E$ of the convolutional autoencoder maps a record $r$ onto a latent vector $z\in\mathbb{R}^{32}$, the amplitude encoding $\mathrm{Enc}$ maps $z$ onto the state $\ket{x_0}$ of $n=5$ qubits (Eq.~\eqref{eq:amp}), and the classical forward process $q(x_t\mid x_0)$ produces the noisy state $\ket{x_t}$ for every step $t=1,\dots,T$ (Eq.~\eqref{eq:forward}). The quantum denoiser $\Phi_\theta(\,\cdot\,;t,k)$ of Fig.~\ref{fig:circuit} maps $\ket{x_t}$ onto $\ket{\psi_{\mathrm{out}}}$, and the infidelity loss $\mathcal{L}$ compares this output with the normalized posterior mean $\ket{\tilde\mu_t}$ (Eqs.~\eqref{eq:posterior} and \eqref{eq:objective}). Sampling (bottom): the trained denoiser is applied for $t=T,\dots,1$ to a Haar-random initial state $\ket{z_T}$, and the decoder $\mathrm{Dec}$ reads out the magnitudes of the final state $\ket{z_0}$ and restores a length drawn from a smoothed distribution of the training lengths $\mathcal{R}_k$ of class $k$, which gives the synthetic latent vector $\hat z$ (Eqs.~\eqref{eq:decode} and \eqref{eq:readout}). Gray boxes are classical computations, the blue box is the quantum circuit, and the orange box is the training objective.}
\label{fig:pipeline}
\end{figure*}

\subsection{Problem formulation and latent representation}\label{sec:latent}

We consider a labeled security dataset $\mathcal{D}=\{(r_j,k_j)\}_{j=1}^{N}$, in which $r_j$ is a record, such as the opcode sequence of a malware program, and $k_j\in\{1,\dots,K\}$ is its class, such as the malware family. We split $\mathcal{D}$ at the record level into a training set $\mathcal{D}_{\mathrm{tr}}$, a validation set $\mathcal{D}_{\mathrm{va}}$, and a test set $\mathcal{D}_{\mathrm{te}}$, and the generative model has access to $\mathcal{D}_{\mathrm{tr}}$ only. Our objective is to learn the class-conditional distribution of the records from $\mathcal{D}_{\mathrm{tr}}$ and to generate synthetic records that remain useful for training downstream classifiers, while staying away from the individual training records.

Instead of modeling the raw records, which are long discrete sequences, we model their latent representations. A pre-trained encoder $E$ maps each record onto a latent vector
\begin{equation}
z_j=E(r_j)\in\mathbb{R}^{d}.\label{eq:latent}
\end{equation}
In this representation, $d$ denotes the latent dimension, and $E$ represents the encoder of a convolutional autoencoder (CAE) trained to reconstruct the records. For the malware programs of this study, the CAE truncates or pads each opcode sequence to $7{,}000$ tokens, embeds every token into a $64$-dimensional vector, applies three one-dimensional convolutional layers with $128$, $256$, and $512$ channels, and projects the result linearly onto $d=32$ latent features. We train the CAE once, on $\mathcal{D}_{\mathrm{tr}}$ alone, with a token-level cross-entropy reconstruction loss and keep it fixed for every generative model in this study, so that all models operate on the same latent space and can be compared directly, and so that no learned component of the pipeline observes $\mathcal{D}_{\mathrm{va}}$ or $\mathcal{D}_{\mathrm{te}}$ during training. The generative models produce latent vectors, which the decoder of the CAE can map back to opcode sequences when required.

\subsection{Threat model and release scenarios}\label{sec:threat}

We consider a data holder, such as an antivirus vendor, a computer emergency response team, or the security operations team of an organization, that holds a labeled corpus of malware programs and intends to share a benchmark derived from it. Such corpora are often confidential, although the malware itself circulates widely. Samples collected during incident response can embed configuration data, credentials, or host names of the affected organization, samples obtained from commercial feeds or under non-disclosure agreements cannot be redistributed, and the presence of a sample in the corpus of an organization can reveal that the organization was attacked. In our experiments, the programs come from the public repository VirusShare, which allows other researchers to reproduce the study, and we treat $\mathcal{D}_{\mathrm{tr}}$ as if it were confidential.

The assets to protect are the training records and their membership in $\mathcal{D}_{\mathrm{tr}}$. We distinguish two release scenarios. In the \emph{data-only release}, the holder publishes synthetic latent vectors with their class labels, and recipients train or augment classifiers in the latent space; to map their own programs into this space, the recipients need the encoder $E$ of Eq.~\eqref{eq:latent}, which must then be a public component, for instance an encoder trained on public programs. In the \emph{release with decoder}, the holder also publishes the codec of Section~\ref{sec:encoding}, that is, the scaling parameters and the length statistics of the decoder, together with the CAE. These artifacts are themselves derived from $\mathcal{D}_{\mathrm{tr}}$: every minimum and maximum of the min--max scaling equals a feature value of a training record, and the decoder stores the training lengths of every class (Eq.~\eqref{eq:lengths}).

The adversary knows the generative procedure, including the codec and the decoding rule with its parameters, receives the released records, and holds candidate programs whose latent vectors and lengths it can compute. Its goal is either to recover training records from the release or to decide whether a candidate program belongs to $\mathcal{D}_{\mathrm{tr}}$. We evaluate the first goal with distance measures and the second with membership tests on the released lengths and, for an adversary that also obtains the trained generator, on the training loss of the generator (Section~\ref{sec:measures}). We do not consider adversaries that poison the training data or query the generator adaptively, and we do not evaluate the privacy of the CAE and of the codec parameters that the release with decoder exposes, which Section~\ref{sec:disc_privacy} discusses.

\subsection{Amplitude encoding of latent vectors}\label{sec:encoding}

Amplitude encoding stores a normalized classical vector in the amplitudes of a quantum state, so that a register of $n$ qubits holds $2^{n}$ features \citep{schuld2021machine}. Since the latent dimension $d=32$ is a power of two, $n=\log_2 d=5$ qubits hold one latent vector. The encoded vector must have unit norm, and the read-out convention of Section~\ref{sec:sampling} recovers the magnitudes of the amplitudes only. We therefore first map every latent vector onto a non-negative vector with a min--max scaling fitted on the training set:
\begin{equation}
u_i(z)=\max\!\left(0,\;\frac{z_i-l_i}{h_i-l_i}\right),\qquad i=1,\dots,d.\label{eq:minmax}
\end{equation}
In this scaling, $z_i$ represents the $i$-th latent feature, $l_i$ and $h_i$ denote the minimum and the maximum of feature $i$ over the training records, and the maximum with zero removes the negative values that validation or test records may produce. We then normalize $u(z)$ and define the encoded state as
\begin{equation}
\ket{\psi(z)}=\sum_{i=0}^{D-1}a_i(z)\ket{i},\qquad a(z)=\frac{u(z)}{\lVert u(z)\rVert_2},\qquad D=2^{n}.\label{eq:amp}
\end{equation}
In this state, $\ket{i}$ represents the $i$-th computational basis state of the $n$-qubit register, $a_i(z)\geq 0$ denotes the corresponding amplitude, and $D$ is the dimension of the Hilbert space. Preparing an arbitrary amplitude-encoded state on hardware requires a number of CNOT gates that grows exponentially with $n$ \citep{mottonen2005transformation,plesch2011quantum}, which remains affordable for the five-qubit register of this study.

The normalization in Eq.~\eqref{eq:amp} removes the length $\rho(z)=\lVert u(z)\rVert_2$ of the vector, which a quantum state cannot carry. We therefore store, for every class $k$, the set of training lengths
\begin{equation}
\mathcal{R}_k=\big\{\rho(z_j)\;:\;(z_j,k_j)\in\mathcal{D}_{\mathrm{tr}},\;k_j=k\big\}.\label{eq:lengths}
\end{equation}
The decoder uses these lengths of the scaled training vectors of class $k$ to restore the length of generated vectors. Given the magnitudes $\lvert\hat\alpha\rvert\in\mathbb{R}^{D}_{\geq 0}$ of a generated state of class $k$, the decoder computes
\begin{equation}
\begin{aligned}
\hat z&=l+(h-l)\odot\big(\rho\,\lvert\hat\alpha\rvert\big),\qquad \rho=\rho_j\,e^{\tau_k\zeta},\\
\rho_j&\sim\mathcal{U}(\mathcal{R}_k),\qquad \zeta\sim\mathcal{N}(0,1).
\end{aligned}\label{eq:decode}
\end{equation}
In this decoder, $\odot$ represents the element-wise product, $\rho_j$ denotes a training length of class $k$ drawn uniformly at random, and $\zeta$ is a standard normal variable drawn independently of $\rho_j$ and of the generated state. The logarithm $\log\rho=\log\rho_j+\tau_k\zeta$ is therefore a draw of the smoothed bootstrap \citep{silverman1987bootstrap} from the logarithms of the training lengths of class $k$, which is a draw from their Gaussian kernel density estimate with bandwidth $\tau_k$. We set this bandwidth with the rule of thumb of \citet{silverman1986density},
\begin{equation}
\tau_k=0.9\,\min\!\Big(s_k,\;\frac{\mathrm{IQR}_k}{1.34}\Big)\,\lvert\mathcal{R}_k\rvert^{-1/5}.\label{eq:bandwidth}
\end{equation}
Here, $s_k$ and $\mathrm{IQR}_k$ are the sample standard deviation and the interquartile range of $\{\log r:r\in\mathcal{R}_k\}$, and $\lvert\mathcal{R}_k\rvert$ is the number of training records of class $k$. With the min--max codec, $\tau_k$ lies between $0.0006$ and $0.014$ for the classes of this study, so that the factor $e^{\tau_k\zeta}$ changes a drawn length by at most $1.4\%$ when $\lvert\zeta\rvert\leq1$.

We draw a random length instead of a fixed one, such as the median length of the class, because identical lengths would allow a classifier trained on synthetic data to use the length as a class label, and this shortcut does not transfer to real data. We perturb the drawn training length as well, because with $\tau_k=0$ every synthetic record would carry the exact length of a training record. Proposition~\ref{prop:length} in Section~\ref{sec:theory} formalizes both properties.

The smoothed bootstrap protects the lengths empirically. For a formal guarantee, we also evaluate a decoder whose length distribution satisfies $\varepsilon$-differential privacy \citep{dwork2006calibrating}. For every class $k$, we count the training lengths in $B$ bins of equal width on the range $[0,\sqrt{D}]$, which contains every training length of the min--max codec and does not depend on the data, and we perturb the counts with the Laplace mechanism,
\begin{equation}
\begin{aligned}
\tilde c_{k,b}&=\max\!\big(0,\;c_{k,b}+\eta_{k,b}\big),\qquad b=1,\dots,B,\\
\eta_{k,b}&\sim\mathrm{Laplace}(1/\varepsilon).
\end{aligned}\label{eq:dplengths}
\end{equation}
Here, $c_{k,b}$ counts the training lengths of class $k$ in bin $b$, and $\varepsilon$ is the privacy budget. The decoder draws a bin with probability $\tilde c_{k,b}/\sum_{b'}\tilde c_{k,b'}$ and a length uniformly within the bin. Adding or removing one training record changes one count by one, so that the noisy histogram of each class is $\varepsilon$-differentially private; since every record belongs to one class, the histograms of all classes together are $\varepsilon$-differentially private, and the lengths drawn from them inherit this guarantee by post-processing \citep{dwork2006calibrating}.

With the min--max codec, the encoded states of different classes are nearly parallel, since all their amplitudes are non-negative and most features lie in the middle of their range. We therefore also evaluate a sign-split codec, which standardizes every feature and separates its positive and negative parts:
\begin{equation}
s=\frac{z-\mu}{\sigma},\qquad u(z)=\big[\max(0,s)\,;\;\max(0,-s)\big]\in\mathbb{R}^{2d}_{\geq 0}.\label{eq:signsplit}
\end{equation}
The feature-wise mean $\mu$ and standard deviation $\sigma$ are computed over the training records, the division is element-wise, and $[\,\cdot\,;\,\cdot\,]$ concatenates two vectors. The sign-split codec doubles the number of amplitudes to $2d=64$, which requires $n=6$ qubits, and spreads the states of different classes further apart in the Hilbert space. Its decoder recombines the two halves as $\hat z=\mu+\sigma\odot\rho\,\big(\lvert\hat\alpha\rvert_{1:d}-\lvert\hat\alpha\rvert_{d+1:2d}\big)$, where $\rho$ is drawn as in Eq.~\eqref{eq:decode}, with the training lengths of Eq.~\eqref{eq:lengths} and the bandwidth of Eq.~\eqref{eq:bandwidth} computed from the sign-split vectors.

Every codec limits the quality that a generator operating on encoded states can reach. We define the round trip $\mathrm{rt}(z)=\mathrm{Dec}\big(\mathrm{Enc}(z)\big)$, which encodes a real record with Eq.~\eqref{eq:amp} and decodes it with Eq.~\eqref{eq:decode}, and we evaluate the classification measures of Section~\ref{sec:measures} (Aug, TSTR, and TRTS) on round trips of real records, averaged over eight draws of the decoded lengths. In other words, the round trip yields the score of a perfect generator under each codec, and we report this ceiling next to every result.

\subsection{Forward diffusion process}\label{sec:forward}

Diffusion models learn to reverse a Markov chain that gradually corrupts the data with Gaussian noise \citep{sohl2015deep,ho2020denoising}. Following \citet{cacioppo2023quantum}, we run the forward chain on classical vectors and employ complex Gaussian noise, since the states produced by the quantum denoiser have complex amplitudes. Let $x_0=a(z)\in\mathbb{R}^{D}$ be the encoded amplitude vector of a training record, and let $\{\beta_t\}_{t=1}^{T}$ be a variance schedule. The noisy vector at step $t$ admits the closed form
\begin{equation}
\begin{aligned}
x_t&=\sqrt{\bar\alpha_t}\,x_0+\sqrt{1-\bar\alpha_t}\,\epsilon,\\
\epsilon&=\epsilon_r+i\,\epsilon_i,\qquad \epsilon_r,\epsilon_i\sim\mathcal{N}(0,I_D).
\end{aligned}\label{eq:forward}
\end{equation}
In this process, $\alpha_t=1-\beta_t$, $\bar\alpha_t=\prod_{s=1}^{t}\alpha_s$ represents the fraction of the signal that remains after $t$ steps, with $\bar\alpha_0=1$, and $\epsilon\sim\mathcal{CN}(0,2I_D)$ denotes complex Gaussian noise with independent real and imaginary parts. The denoiser receives the normalized state $\ket{x_t}=x_t/\lVert x_t\rVert_2$. We employ a linear schedule,
\begin{equation}
\beta_t=\beta_{\min}+\frac{t-1}{T}\,\big(\beta_{\max}-\beta_{\min}\big),\qquad t=1,\dots,T.\label{eq:schedule}
\end{equation}
Here, $\beta_{\min}$ and $\beta_{\max}$ bound the variance range, and $T$ is the number of diffusion steps.

The appropriate variances depend on the dimension $D$. Since $x_0$ has unit norm, whereas $\mathbb{E}\lVert\epsilon\rVert_2^2=2D$, the ratio between the norms of the noise and of the signal in Eq.~\eqref{eq:forward} is approximately
\begin{equation}
\kappa_t=\sqrt{\frac{2D\,(1-\bar\alpha_t)}{\bar\alpha_t}}.\label{eq:nsr}
\end{equation}
Because this noise-to-signal ratio grows with $\sqrt{D}$, the same schedule corrupts the $D=32$ amplitudes of our latent vectors about three times less than the $D=256$ amplitudes of the $16\times16$ images considered by \citet{cacioppo2023quantum}, and we therefore employ larger variances. With $\beta_{\min}=0.02$, $\beta_{\max}=0.10$, and $T=10$, the variances range from $\beta_1=0.020$ to $\beta_{10}=0.092$, the signal fraction decreases to $\bar\alpha_{10}=0.56$, and the ratio grows from $\kappa_1=1.1$ to $\kappa_{10}=7.1$.

For the training target, we employ the mean of the posterior of the forward process \citep{ho2020denoising},
\begin{equation}
\begin{aligned}
\tilde\mu_t(x_t,x_0)&=\frac{\sqrt{\bar\alpha_{t-1}}\,\beta_t}{1-\bar\alpha_t}\,x_0+\frac{\sqrt{\alpha_t}\,(1-\bar\alpha_{t-1})}{1-\bar\alpha_t}\,x_t,\\
\ket{\tilde\mu_t}&=\frac{\tilde\mu_t}{\lVert\tilde\mu_t\rVert_2}.
\end{aligned}\label{eq:posterior}
\end{equation}
In this target, $\tilde\mu_t$ represents the expected value of $x_{t-1}$ given $x_t$ and $x_0$, computed from the unnormalized $x_t$ of Eq.~\eqref{eq:forward}. At $t=1$, the target reduces to the clean state $\ket{x_0}$, because $\bar\alpha_0=1$ and $1-\bar\alpha_1=\beta_1$. The training algorithm described by \citet{cacioppo2025hybrid} instead targets a single draw of $x_{t-1}$ from the forward chain; the posterior mean removes the variance of this draw from the target, in line with the public implementation of their model. Theorem~\ref{thm:target} in Section~\ref{sec:theory} establishes that both targets lead to the same optimal output.

\subsection{Quantum denoiser}\label{sec:denoiser}

A parameterized quantum circuit (PQC) applies a sequence of gates, some of which depend on trainable parameters, to a quantum state, and it is commonly regarded as the quantum counterpart of a neural network \citep{benedetti2019parameterized}. Following the reverse-bottleneck architecture of \citet{cacioppo2023quantum}, we implement the denoiser of step $t$ for class $k$ as a PQC $\Phi_{\theta}(\,\cdot\,;t,k)$ that acts on $n$ data qubits and one ancillary qubit (Fig.~\ref{fig:circuit}). The denoiser maps the state $\ket{x_t}$ onto an estimate of $\ket{\tilde\mu_t}$ without measuring the data qubits, so that the complete reverse chain remains a sequence of quantum operations.

\begin{figure*}[t]
\centering
\includegraphics{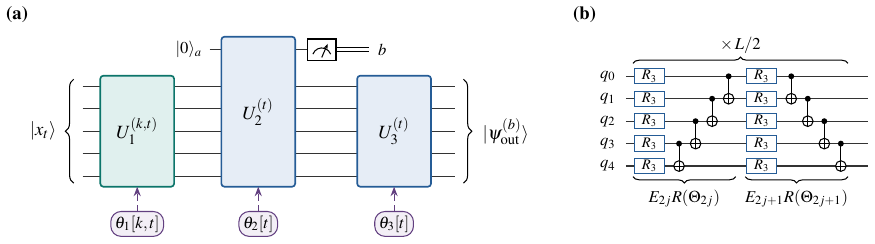}
\caption{Quantum denoiser of one reverse step, after the reverse-bottleneck architecture of \citet{cacioppo2023quantum}. (a) For step $t$ and class $k$, the block $U_1^{(k,t)}$ acts on the $n=5$ data qubits that hold $\ket{x_t}$, the block $U_2^{(t)}$ acts on the data qubits and an ancilla prepared in $\ket{0}_a$, the ancilla is measured with outcome $b\in\{0,1\}$, and the block $U_3^{(t)}$ produces the output $\ket{\psi^{(b)}_{\mathrm{out}}}$ (Eqs.~\eqref{eq:joint} and \eqref{eq:branches}). With post-selection, the denoiser keeps $b=0$ (Eq.~\eqref{eq:postselect}); without branch selection, it outputs the mixture of Eq.~\eqref{eq:trace}. The angles $\theta_1[k,t]$ of the teal block are specific to the class, whereas the angles $\theta_2[t]$ and $\theta_3[t]$ of the blue blocks are shared by all classes; with $L_1=L_2=L_3=4$ layers, the three blocks hold $60$, $72$, and $60$ angles. (b) Strongly entangling circuit (SEC) of every block (Eq.~\eqref{eq:sec}), drawn for $m=5$ qubits $q_0,\dots,q_4$ and two consecutive layers. Each layer applies the rotations $R_3=R_Z(\omega)R_Y(\vartheta)R_Z(\varphi)$ to every qubit, followed by the CNOT chain $E_{2j}=C$ on even layers and its inverse $E_{2j+1}=C^{\dagger}$ on odd layers. The pair of layers is repeated $L/2$ times, so the block equals the identity when all angles vanish. For $U_2^{(t)}$, $m=6$ and the ancilla is $q_0$.}
\label{fig:circuit}
\end{figure*}

\paragraph{Strongly entangling block}
Each unitary block of the denoiser is a strongly entangling circuit (SEC) \citep{schuld2020circuit} on $m$ qubits with $L$ layers. A layer first applies a general single-qubit rotation to every qubit,
\begin{equation}
\begin{aligned}
R_3(\omega,\vartheta,\varphi)&=R_Z(\omega)\,R_Y(\vartheta)\,R_Z(\varphi)\\
&=\begin{pmatrix}
e^{-i(\varphi+\omega)/2}\cos\frac{\vartheta}{2} & -e^{i(\varphi-\omega)/2}\sin\frac{\vartheta}{2}\\[2pt]
e^{-i(\varphi-\omega)/2}\sin\frac{\vartheta}{2} & e^{i(\varphi+\omega)/2}\cos\frac{\vartheta}{2}
\end{pmatrix}.
\end{aligned}\label{eq:r3}
\end{equation}
The three angles $\omega$, $\vartheta$, and $\varphi$ of one qubit in one layer are trainable, and $R_Y$ and $R_Z$ are the rotations around the $y$ and $z$ axes of the Bloch sphere. After the rotations, a chain of controlled-NOT (CNOT) gates entangles neighboring qubits:
\begin{equation}
C=\mathrm{CNOT}_{0,1}\,\mathrm{CNOT}_{1,2}\cdots\mathrm{CNOT}_{m-2,m-1}.\label{eq:chain}
\end{equation}
Each gate $\mathrm{CNOT}_{q,q+1}$ has the control qubit $q$ and the target qubit $q+1$, and the rightmost gate acts first. With $\Theta_l$ denoting the angles of layer $l$, the block is
\begin{equation}
\begin{aligned}
U(\Theta)&=E_{L-1}R(\Theta_{L-1})\cdots E_{1}R(\Theta_{1})\,E_{0}R(\Theta_{0}),\\
R(\Theta_l)&=\bigotimes_{q=0}^{m-1}R_3(\omega_{l,q},\vartheta_{l,q},\varphi_{l,q}),\\
E_l&=\begin{cases}C, & l\ \text{even},\\ C^{\dagger}, & l\ \text{odd}.\end{cases}
\end{aligned}\label{eq:sec}
\end{equation}
The block thus alternates the rotations of every layer $l$ with its entangling chain $E_l$. Like the public implementation of \citet{cacioppo2025hybrid}, we alternate the chain $C$ with its inverse $C^{\dagger}$ and employ an even number of layers. Since $R_3(0,0,0)=I$ and $C^{\dagger}C=I$, the block reduces to the identity when all angles vanish. We initialize the angles from $\mathcal{N}(0,0.1^2)$, a narrower distribution than the standard deviation of $0.7$ in the public implementation, so that every block, and therefore the whole denoiser, starts close to the identity map (Proposition~\ref{prop:init}). This initialization suits the task, since each reverse step maps a state onto a nearby state.

\paragraph{Reverse-bottleneck denoiser}
The denoiser composes three SEC blocks. The first block $U_1^{(k,t)}$ acts on the data register, the second block $U_2^{(t)}$ acts on the data register together with an ancillary qubit prepared in $\ket{0}_a$, and the third block $U_3^{(t)}$ acts on the data register after the ancilla has been measured. Before the measurement, the joint state is
\begin{equation}
\begin{aligned}
\ket{\Psi_t}&=U_2^{(t)}\Big(\ket{0}_a\otimes U_1^{(k,t)}\ket{x_t}\Big)\\
&=\sqrt{p_0}\,\ket{0}_a\otimes\ket{\phi_0}+\sqrt{p_1}\,\ket{1}_a\otimes\ket{\phi_1}.
\end{aligned}\label{eq:joint}
\end{equation}
In this state, $p_b=\big\lVert(\bra{b}_a\otimes I)\ket{\Psi_t}\big\rVert_2^{2}$ represents the probability of measuring the ancilla in $\ket{b}$, $b\in\{0,1\}$, and $\ket{\phi_b}=(\bra{b}_a\otimes I)\ket{\Psi_t}/\sqrt{p_b}$ denotes the post-measurement state of the data register. The third block then produces one output state per outcome,
\begin{equation}
\ket{\psi_{\mathrm{out}}^{(b)}}=U_3^{(t)}\ket{\phi_b},\qquad b\in\{0,1\}.\label{eq:branches}
\end{equation}

We consider two ways of handling the ancilla. With post-selection, also called branch selection, the denoiser keeps the outcome $b=0$ only, as in the public implementation of \citet{cacioppo2025hybrid}:
\begin{equation}
\begin{aligned}
\Phi_\theta(\ket{x_t};t,k)&=\ket{\psi_{\mathrm{out}}^{(0)}}\\
&=\frac{U_3^{(t)}\big(\bra{0}_a\otimes I\big)\,U_2^{(t)}\big(\ket{0}_a\otimes U_1^{(k,t)}\ket{x_t}\big)}{\sqrt{p_0}}.
\end{aligned}\label{eq:postselect}
\end{equation}
In this map, the division by $\sqrt{p_0}$, which depends on the input state, makes the denoiser a nonlinear function of the amplitudes of $\ket{x_t}$. The mid-circuit measurement is therefore the only source of nonlinearity of a circuit that is otherwise composed of unitary blocks \citep{cacioppo2023quantum}. Without branch selection, the denoiser outputs the mixture
\begin{equation}
\varrho_{\mathrm{out}}=\sum_{b\in\{0,1\}}p_b\,\ket{\psi_{\mathrm{out}}^{(b)}}\bra{\psi_{\mathrm{out}}^{(b)}}.\label{eq:trace}
\end{equation}
This mixture is the density matrix of the data register after the ancilla is traced out. The variant without branch selection is therefore a quantum channel, which acts linearly on the input density matrix and avoids the cost of post-selection on hardware discussed in Section~\ref{sec:sampling}. By default, the public implementation additionally applies a complex leaky rectified linear unit to the amplitudes between the measurement and $U_3^{(t)}$; since this operation has no quantum counterpart, we exclude it from our model.

\paragraph{Conditioning on the step and the class}
The denoiser contains no classical layer, so we condition it through the choice of its angles. Following the public implementation of \citet{cacioppo2025hybrid}, every step $t$ has its own angles for every block, although \citet{cacioppo2023quantum} report that a single set of angles shared by all steps reduces the performance only slightly. For the class, we give every class $k$ its own first block, while all classes share the second and the third blocks:
\begin{equation}
\theta=\Big\{\theta_1[k,t],\;\theta_2[t],\;\theta_3[t]\;:\;k=1,\dots,K,\;t=1,\dots,T\Big\}.\label{eq:params}
\end{equation}
In this parameterization, $\theta_1[k,t]$, $\theta_2[t]$, and $\theta_3[t]$ denote the angles of $U_1^{(k,t)}$, $U_2^{(t)}$, and $U_3^{(t)}$, respectively. \citet{cacioppo2023quantum} condition their model on the class with additional label qubits, using the input $\ket{k}\otimes\ket{x_t}$; our class-specific first block instead keeps the register at $n+1$ qubits for any number of classes. The number of trainable angles is
\begin{equation}
P=3T\big(K\,n\,L_1+(n+1)\,L_2+n\,L_3\big).\label{eq:count}
\end{equation}
In this count, $L_1$, $L_2$, and $L_3$ represent the numbers of layers of the three blocks, and the factor $3$ stems from the three angles of each rotation. With $n=5$, $K=17$, $T=10$, and $L_1=L_2=L_3=4$, the denoiser has $P=11{,}520$ trainable angles, of which the class-specific first blocks hold $10{,}200$ (Table~\ref{tab:params}).

\begin{table}[t]
\small
\caption{Trainable angles of the quantum denoiser with $n=5$ data qubits, $K=17$ classes, $T=10$ steps, and $L_1=L_2=L_3=4$ layers (min--max codec). With the sign-split codec ($n=6$), the total becomes $13{,}800$.}\label{tab:params}
\begin{tabular*}{\tblwidth}{@{}LLLLR@{}}
\toprule
Block & Qubits & Angles per set & Parameter sets & Angles\\
\midrule
$U_1^{(k,t)}$ & $n=5$ & $3nL_1=60$ & $KT=170$ & $10{,}200$\\
$U_2^{(t)}$ & $n+1=6$ & $3(n+1)L_2=72$ & $T=10$ & $720$\\
$U_3^{(t)}$ & $n=5$ & $3nL_3=60$ & $T=10$ & $600$\\
\midrule
Total & & & & $11{,}520$\\
\bottomrule
\end{tabular*}
\end{table}

\subsection{Training objective}\label{sec:training}

We train the denoiser to map every noisy state onto the corresponding target state. Since both are quantum states, we employ the infidelity loss of \citet{cacioppo2023quantum}, which can be estimated on quantum hardware with a swap test \citep{barenco1997stabilization,buhrman2001quantum}. For the post-selected denoiser, the loss of a single sample is
\begin{equation}
\ell_{\mathrm{ps}}(\theta;x_0,k,t,\epsilon)=1-\big\lvert\braket{\tilde\mu_t|\psi_{\mathrm{out}}^{(0)}}\big\rvert^{2},\label{eq:loss_ps}
\end{equation}
and for the denoiser without branch selection,
\begin{equation}
\begin{aligned}
\ell_{\mathrm{tr}}(\theta;x_0,k,t,\epsilon)&=1-\bra{\tilde\mu_t}\varrho_{\mathrm{out}}\ket{\tilde\mu_t}\\
&=1-\sum_{b\in\{0,1\}}p_b\,\big\lvert\braket{\tilde\mu_t|\psi_{\mathrm{out}}^{(b)}}\big\rvert^{2}.
\end{aligned}\label{eq:loss_tr}
\end{equation}
In these losses, $\lvert\braket{\tilde\mu_t|\psi}\rvert^{2}$ represents the fidelity between the target state and an output state, and the second form of Eq.~\eqref{eq:loss_tr} stems from the linearity of the fidelity with a pure state in the density matrix. The training objective averages the loss over the training records, the diffusion steps, and the noise:
\begin{equation}
\begin{aligned}
\mathcal{L}(\theta)={}&\mathbb{E}_{(z,k)\sim\mathcal{D}_{\mathrm{tr}}}\Bigg[\frac{1}{T}\sum_{t=1}^{T}\\
&\mathbb{E}_{\epsilon\sim\mathcal{CN}(0,2I_D)}\Big[\ell\big(\theta;a(z),k,t,\epsilon\big)\Big]\Bigg].
\end{aligned}\label{eq:objective}
\end{equation}
In this objective, $\ell$ represents either $\ell_{\mathrm{ps}}$ or $\ell_{\mathrm{tr}}$, and the sum over all steps reflects our training procedure: like the public implementation, we diffuse every record of a mini-batch to every step $t=1,\dots,T$, instead of drawing a single random step per record as in classical DDPMs \citep{ho2020denoising}.

We minimize Eq.~\eqref{eq:objective} with the Adam optimizer \citep{kingma2015adam} in exact state-vector simulation, since the estimation of gradients on current noisy hardware remains impractical \citep{cacioppo2025hybrid}. Lemma~\ref{lem:smooth} in Appendix~\ref{app:proofs} bounds the smoothness of the objective and the convergence rate of this procedure. The parameter-shift rule \citep{mitarai2018quantum,schuld2019evaluating} would allow the same gradients to be estimated on a quantum device. We select the model with an early-stopping rule on the validation set. Every five epochs, we evaluate the loss on all validation records at all steps with a fixed noise realization, and we stop the training when the validation loss has not improved by at least $0.1\%$ during eight consecutive evaluations. As a reference, we also compute the loss of the identity map, which returns its input unchanged:
\begin{equation}
\mathcal{L}_{\mathrm{id}}=\mathbb{E}\Big[1-\big\lvert\braket{\tilde\mu_t|x_t}\big\rvert^{2}\Big].\label{eq:identity}
\end{equation}
The identity map approximates our denoiser at initialization. A trained denoiser therefore contributes to the generation only when its validation loss falls below $\mathcal{L}_{\mathrm{id}}$. Algorithm~\ref{alg:train} summarizes the training procedure.

Post-selected sampling on hardware pays for every rejected run (Section~\ref{sec:sampling}). To examine whether training can exchange denoising accuracy for a higher acceptance probability, we also train the post-selected denoiser with the penalized objective
\begin{equation}
\begin{aligned}
\mathcal{L}_{\lambda}(\theta)={}&\mathcal{L}(\theta)+\lambda\,\mathbb{E}_{(z,k)\sim\mathcal{D}_{\mathrm{tr}}}\Bigg[\frac{1}{T}\sum_{t=1}^{T}\\
&\mathbb{E}_{\epsilon}\Big[-\log p_0\big(\theta;a(z),k,t,\epsilon\big)\Big]\Bigg].
\end{aligned}\label{eq:pacc_loss}
\end{equation}
The weight $\lambda\geq0$ scales the penalty on the probability $p_0(\theta;a(z),k,t,\epsilon)$ of the ancilla outcome $0$ for the noisy input of each pair; along a trajectory whose inputs follow the forward process, the sum of $-\log p_0$ over the $T$ steps equals $-\log P_{\mathrm{acc}}$ of Eq.~\eqref{eq:acceptance}. With $\lambda>0$, we also select the model with $\mathcal{L}_{\lambda}$ on the validation pairs.

\begin{algorithm}[t]
\caption{Training of the quantum denoiser}\label{alg:train}
\begin{algorithmic}[1]
\REQUIRE encoded training set $\{(x_0^{(j)},k_j)\}$, schedule $\{\beta_t\}_{t=1}^{T}$, initial angles $\theta\sim\mathcal{N}(0,0.1^2)$
\REPEAT
\STATE Draw a mini-batch $\mathcal{B}$ of training records
\FORALL{$(x_0,k)\in\mathcal{B}$ and $t=1,\dots,T$}
\STATE $\epsilon\gets\epsilon_r+i\,\epsilon_i$ with $\epsilon_r,\epsilon_i\sim\mathcal{N}(0,I_D)$
\STATE $x_t\gets\sqrt{\bar\alpha_t}\,x_0+\sqrt{1-\bar\alpha_t}\,\epsilon$ and $\ket{x_t}\gets x_t/\lVert x_t\rVert_2$
\STATE $\ket{\tilde\mu_t}\gets$ normalized posterior mean of Eq.~\eqref{eq:posterior}
\STATE Evaluate $\Phi_\theta(\ket{x_t};t,k)$ and the loss of Eq.~\eqref{eq:loss_ps} or Eq.~\eqref{eq:loss_tr}
\ENDFOR
\STATE Update $\theta$ with Adam on the mean loss over $\mathcal{B}$ and all steps
\STATE Every five epochs, evaluate the validation loss and keep the best angles
\UNTIL{the validation loss stops improving}
\ENSURE trained angles $\theta$
\end{algorithmic}
\end{algorithm}

\subsection{Sampling and decoding}\label{sec:sampling}

To generate a record of class $k$, we start from a random state and apply the trained denoisers in reverse order (Algorithm~\ref{alg:sample}):
\begin{equation}
\begin{aligned}
\ket{z_T}&=\frac{\xi}{\lVert\xi\rVert_2},\qquad \xi\sim\mathcal{CN}(0,2I_D),\\
\ket{z_{t-1}}&=\Phi_\theta(\ket{z_t};t,k),\qquad t=T,\dots,1.
\end{aligned}\label{eq:reverse}
\end{equation}
In this chain, $\ket{z_T}$ represents the initial state, which is distributed uniformly over the pure states of the register (Haar-random), and $\ket{z_0}$ denotes the generated state. Unlike ancestral sampling in classical DDPMs \citep{ho2020denoising}, the chain adds no noise between the steps. With post-selection, the generated state is therefore a deterministic function of $\ket{z_T}$. Without branch selection, the outcome $b_t$ of the ancilla at step $t$ is drawn with probability $p_{b_t}$, and the chain continues with $\ket{\psi^{(b_t)}_{\mathrm{out}}}$, so that the measurements become a second source of randomness.

Since every step maps its input onto an estimate of a posterior mean, the deterministic chain may contract the samples toward the mean of their class. To test this effect, we also evaluate two samplers of classical diffusion models, adapted to the normalized states of our model. The denoiser returns a direction only, so we restore the expected norms of the unnormalized chain, $n_t=\big(\bar\alpha_t+2D(1-\bar\alpha_t)\big)^{1/2}$ for $x_t$ and $r_t=\big((c_{1,t}+c_{2,t}\sqrt{\bar\alpha_t})^{2}+2Dc_{2,t}^{2}(1-\bar\alpha_t)\big)^{1/2}$ for $\tilde\mu_t$, which follow from Eqs.~\eqref{eq:forward} and \eqref{eq:posterior}, with $c_{1,t}$ and $c_{2,t}$ the coefficients of $x_0$ and $x_t$ in Eq.~\eqref{eq:posterior}. With $\ket{\phi_t}=\Phi_\theta(\ket{z_t};t,k)$, the two samplers set, for $t\geq2$,
\begin{equation}
\begin{aligned}
&\text{Ancestral:}\quad \ket{z_{t-1}}\propto r_t\ket{\phi_t}+\sqrt{\tilde\beta_t}\,\xi_t,\\
&\text{DDIM:}\quad \ket{z_{t-1}}\propto\sqrt{\bar\alpha_{t-1}}\,\hat x_0+\sqrt{1-\bar\alpha_{t-1}}\,\frac{n_t\ket{z_t}-\sqrt{\bar\alpha_t}\,\hat x_0}{\sqrt{1-\bar\alpha_t}}.
\end{aligned}\label{eq:samplers}
\end{equation}
Here, $\tilde\beta_t=(1-\bar\alpha_{t-1})\beta_t/(1-\bar\alpha_t)$ is the variance of the posterior of the forward process, $\xi_t\sim\mathcal{CN}(0,2I_D)$ is fresh noise, $\hat x_0=(r_t\ket{\phi'_t}-c_{2,t}n_t\ket{z_t})/c_{1,t}$ is the estimate of the clean vector implied by the output, and $\ket{\phi'_t}$ is the output with its global phase aligned to $\ket{z_t}$. The first update follows ancestral sampling \citep{ho2020denoising}, and the second follows the deterministic sampler of \citet{song2021denoising}; both return the output of the last step unchanged. Adding Gaussian noise to the amplitudes and renormalizing is not a quantum operation, so that the ancestral sampler would require, on hardware, a quantum channel with a comparable effect or a measurement and a new preparation of the state after every step; we therefore employ both samplers to test the role of the sampler in simulation.

On hardware, post-selection discards every run in which at least one ancilla returns $1$. The probability that a run is accepted and the expected number of runs per accepted sample are
\begin{equation}
P_{\mathrm{acc}}=\prod_{t=1}^{T}p_0^{(t)},\qquad \mathbb{E}\big[N_{\mathrm{runs}}\big]=\frac{1}{P_{\mathrm{acc}}}.\label{eq:acceptance}
\end{equation}
In these expressions, $p_0^{(t)}$ represents the probability of the ancilla outcome $0$ at step $t$ along the trajectory of the sample. Since $P_{\mathrm{acc}}$ decreases exponentially with $T$ when $p_0^{(t)}$ stays below one, we report it for every trained model as a measure of the cost of post-selected sampling on a quantum device; when every attempt starts from a fresh initial state, a record costs $1/\mathbb{E}[P_{\mathrm{acc}}]$ runs on average, where the expectation runs over the initial states.

Post-selection also changes the distribution of the released records. On hardware, only accepted runs release a record, so that the initial states of the released records follow the Haar distribution weighted by $P_{\mathrm{acc}}(\ket{z_T})$, whereas the simulation returns every initial state once. We quantify this effect with the weights $w_i=P_{\mathrm{acc}}(\ket{z_T^{(i)}})$ of the $n_s$ samples of a class, their effective sample size $\mathrm{ESS}=(\sum_iw_i)^{2}/\sum_iw_i^{2}$, and the scores of $n_s$ records drawn with replacement with probabilities $w_i/\sum_jw_j$, which approximate the release of a quantum device.

Finally, we read out the generated state in the computational basis. On hardware, measuring $S$ copies of $\ket{z_0}$ gives counts $c_i$ for the basis states, from which we estimate the magnitudes of the amplitudes as
\begin{equation}
\lvert\hat\alpha_i\rvert=\sqrt{c_i/S},\qquad i=0,\dots,D-1.\label{eq:readout}
\end{equation}
This estimate converges to the exact magnitude $\lvert\braket{i|z_0}\rvert$ as the number of shots $S$ grows. Like \citet{cacioppo2023quantum}, we discard the phases of the amplitudes, since the encoded data are non-negative. The decoder of Eq.~\eqref{eq:decode} then maps $\lvert\hat\alpha\rvert$ onto a latent vector $\hat z\in\mathbb{R}^{d}$. Lemma~\ref{lem:shots} in Appendix~\ref{app:proofs} bounds the number of shots that this read-out requires.

\begin{algorithm}[t]
\caption{Sampling with the quantum denoiser}\label{alg:sample}
\begin{algorithmic}[1]
\REQUIRE trained angles $\theta$, class $k$, training lengths $\mathcal{R}_k$, bandwidth $\tau_k$
\STATE $\xi\gets\xi_r+i\,\xi_i$ with $\xi_r,\xi_i\sim\mathcal{N}(0,I_D)$, and $\ket{z_T}\gets\xi/\lVert\xi\rVert_2$
\FOR{$t=T$ \TO $1$}
\STATE Apply $U_1^{(k,t)}$, append the ancilla $\ket{0}_a$, and apply $U_2^{(t)}$
\STATE Measure the ancilla: keep $b_t=0$ (post-selection), or draw $b_t$ with probability $p_{b_t}$
\STATE $\ket{z_{t-1}}\gets U_3^{(t)}\ket{\phi_{b_t}}$
\ENDFOR
\STATE Read out $\lvert\hat\alpha\rvert$ from $\ket{z_0}$ with Eq.~\eqref{eq:readout}
\STATE Draw $\rho_j\sim\mathcal{U}(\mathcal{R}_k)$ and $\zeta\sim\mathcal{N}(0,1)$, and set $\hat z\gets l+(h-l)\odot\big(\rho_j e^{\tau_k\zeta}\lvert\hat\alpha\rvert\big)$
\ENSURE synthetic latent vector $\hat z$
\end{algorithmic}
\end{algorithm}

\subsection{Simulation of hardware noise}\label{sec:noise}

Algorithms~\ref{alg:train} and \ref{alg:sample} operate on exact state vectors. To estimate how the trained denoiser behaves on NISQ devices \citep{preskill2018quantum}, we additionally sample with density matrices. After every layer of every block, we apply a single-qubit depolarizing channel to each qubit \citep{nielsen2010quantum},
\begin{equation}
\mathcal{E}_p(\varrho)=(1-p)\,\varrho+\frac{p}{3}\big(X\varrho X+Y\varrho Y+Z\varrho Z\big).\label{eq:depolarizing}
\end{equation}
Here, $p$ is the error probability, and $X$, $Y$, and $Z$ are the Pauli operators acting on the qubit. We handle the ancilla on the density matrix as well. Post-selection keeps the block of the joint density matrix $\varrho_{\mathrm{joint}}$ that corresponds to the ancilla outcome $0$ and renormalizes it:
\begin{equation}
\varrho\;\mapsto\;\frac{(\bra{0}_a\otimes I)\,\varrho_{\mathrm{joint}}\,(\ket{0}_a\otimes I)}{\operatorname{Tr}\big[(\bra{0}_a\otimes I)\,\varrho_{\mathrm{joint}}\,(\ket{0}_a\otimes I)\big]}.\label{eq:dm_postselect}
\end{equation}
The trace in the denominator is the probability $p_0^{(t)}$ of the accepted outcome, which enters $P_{\mathrm{acc}}$ in Eq.~\eqref{eq:acceptance}. The variant without branch selection traces the ancilla out, $\varrho\mapsto\operatorname{Tr}_a\varrho_{\mathrm{joint}}$. We carry the density matrix through all $T$ steps, as on hardware, where the $T$ circuits run back to back with one ancilla measurement per step, and we read out the final state with $S$ shots as in Eq.~\eqref{eq:readout}, with counts drawn from the multinomial distribution defined by the diagonal of the final density matrix.

\subsection{Parameter-matched classical control}\label{sec:control}

To separate the contribution of the quantum circuit from that of its parameter budget, we compare the denoiser with a classical denoiser of the same size, trained on the same task. The control is a residual multilayer perceptron (MLP) that predicts a correction of the noisy state:
\begin{equation}
\hat x=\frac{x_t+g_\phi(f_t)}{\lVert x_t+g_\phi(f_t)\rVert_2},\qquad f_t=\big[\operatorname{Re}x_t\,;\;\operatorname{Im}x_t\,;\;e_t\,;\;e_k\big].\label{eq:control}
\end{equation}
In this control, $g_\phi$ represents an MLP with two hidden layers of width $h$ and sigmoid linear unit (SiLU) activations \citep{elfwing2018sigmoid}, whose $2D$ outputs form the real and the imaginary parts of the correction, and $e_t\in\{0,1\}^{T}$ and $e_k\in\{0,1\}^{K}$ denote the one-hot encodings of the step and the class. We choose the width $h$ as the positive root of
\begin{equation}
h^{2}+\big(d_{\mathrm{in}}+d_{\mathrm{out}}+2\big)\,h+d_{\mathrm{out}}=P,\label{eq:width}
\end{equation}
rounded to the nearest integer. The left-hand side counts the weights and biases of an MLP with the input and output sizes $d_{\mathrm{in}}=2D+T+K$ and $d_{\mathrm{out}}=2D$, and $P$ is the number of angles of the quantum denoiser in Eq.~\eqref{eq:count}. For $P=11{,}520$, the control has $h=54$ hidden units and $11{,}458$ parameters. We initialize the last layer of $g_\phi$ at zero, so that the control starts at the identity map like the quantum denoiser, and we train it with the loss of Eq.~\eqref{eq:loss_ps} and the same optimization procedure.

\subsection{Utility, fidelity, and privacy measures}\label{sec:measures}

We evaluate the generated latent vectors with a fixed protocol, which we apply to every generator compared in this work, on all $K$ classes. Let $\mathcal{X}_{\mathrm{tr}}=\mathcal{D}_{\mathrm{tr}}\cup\mathcal{D}_{\mathrm{va}}$ denote the real records available for training classifiers, let $\mathcal{X}_{\mathrm{te}}=\mathcal{D}_{\mathrm{te}}$ denote the real test records, which no generator observes during training, and let $\mathcal{S}$ denote the synthetic records of all classes. We train three classifiers, a random forest \citep{breiman2001random}, a support vector machine (SVM) with a radial basis function kernel \citep{cortes1995support}, and an MLP, with scikit-learn \citep{pedregosa2011scikit}, and we report their mean macro-F1 score in three settings:
\begin{itemize}
\item \emph{Augmentation (Aug)}: train on $\mathcal{X}_{\mathrm{tr}}\cup\mathcal{S}$ and test on $\mathcal{X}_{\mathrm{te}}$;
\item \emph{Train on synthetic, test on real (TSTR)} \citep{esteban2017real}: train on $\mathcal{S}$ and test on $\mathcal{X}_{\mathrm{te}}$;
\item \emph{Train on real, test on synthetic (TRTS)}: train on $\mathcal{X}_{\mathrm{tr}}$ and test on $\mathcal{S}$.
\end{itemize}
Aug measures whether the synthetic records add useful information to the real training data. TSTR and TRTS approach the same question from opposite directions, since both test whether the synthetic records carry the information that separates the classes in the real data. Every setting tests on records that no generator observed during training, so that a generator cannot raise its scores by reproducing its training records.

To measure the fidelity of the generated distribution, we compute for every class the squared maximum mean discrepancy (MMD) \citep{gretton2012kernel} between the real records $X=\{x_i\}_{i=1}^{n_X}$ and the synthetic records $Y=\{y_j\}_{j=1}^{n_Y}$ of the class:
\begin{equation}
\begin{aligned}
\mathrm{MMD}^{2}(X,Y)={}&\sum_{\sigma\in\Sigma}\Bigg[\frac{1}{n_X^{2}}\sum_{i,i'}\kappa_\sigma(x_i,x_{i'})\\
&+\frac{1}{n_Y^{2}}\sum_{j,j'}\kappa_\sigma(y_j,y_{j'})\\
&-\frac{2}{n_X n_Y}\sum_{i,j}\kappa_\sigma(x_i,y_j)\Bigg].
\end{aligned}\label{eq:mmd}
\end{equation}
In this estimator, $\kappa_\sigma(u,v)=\exp\!\big(-\lVert u-v\rVert_2^{2}/(2\sigma^{2})\big)$ represents a Gaussian kernel with bandwidth $\sigma$, and $\Sigma=\{0.25,0.5,1,2,4\}\cdot\sigma_{\mathrm{med}}$ denotes a set of bandwidths around the median pairwise distance $\sigma_{\mathrm{med}}$ of the pooled records. Lower values indicate closer distributions.

To quantify memorization, which is the main privacy risk of releasing synthetic records, we compare the distance from each synthetic record to its nearest training record with the same distance for unseen real records. For every class, we compute
\begin{equation}
\begin{aligned}
d_{\mathrm{gen}}(y)&=\min_{x\in\mathcal{X}_{\mathrm{tr}}}\lVert y-x\rVert_2,& y&\in\mathcal{S},\\
d_{\mathrm{ref}}(x')&=\min_{x\in\mathcal{X}_{\mathrm{tr}}}\lVert x'-x\rVert_2,& x'&\in\mathcal{X}_{\mathrm{te}}.
\end{aligned}\label{eq:dcr}
\end{equation}
In these distances, $d_{\mathrm{gen}}$ represents the distance to the closest training record of a synthetic record \citep{park2018data}, and $d_{\mathrm{ref}}$ denotes the same distance for a real test record, which sets the natural scale of the data. We then report
\begin{equation}
\begin{aligned}
r_{\mathrm{mem}}&=\frac{\operatorname{median}_{y\in\mathcal{S}}\,d_{\mathrm{gen}}(y)}{\operatorname{median}_{x'\in\mathcal{X}_{\mathrm{te}}}\,d_{\mathrm{ref}}(x')},\\
q_{10}&=\frac{1}{\lvert\mathcal{S}\rvert}\sum_{y\in\mathcal{S}}\mathbf{1}\Big[d_{\mathrm{gen}}(y)<Q_{0.1}\big(d_{\mathrm{ref}}\big)\Big].
\end{aligned}\label{eq:mem}
\end{equation}
Here, $Q_{0.1}(d_{\mathrm{ref}})$ is the $10$th percentile of the reference distances, and $\mathbf{1}[\cdot]$ is the indicator function. A ratio $r_{\mathrm{mem}}<1$ indicates that the synthetic records lie closer to the training records than unseen real records do, and a fraction $q_{10}$ well above $0.1$ suggests that the generator copies training records. As a test reference, we also compute the same fraction, denoted $q_{10}^{\mathrm{te}}$, with the distance of every synthetic record to its nearest real test record in place of $d_{\mathrm{gen}}$; since no generator observes the test records, a generator that copies training records yields a fraction $q_{10}$ well above $q_{10}^{\mathrm{te}}$. We compute these distances on features standardized with the statistics of $\mathcal{X}_{\mathrm{tr}}$.

Since the decoder of Eq.~\eqref{eq:decode} draws the length of every synthetic record around the lengths of the training records, we also test whether the released lengths reveal which records belong to $\mathcal{D}_{\mathrm{tr}}$. An adversary who knows the codec computes the length $\rho(x)$ of a candidate record $x$ of class $k$ and scores it with
\begin{equation}
\begin{aligned}
s(x)&=-\min_{y\in\mathcal{S}_k}\big\lvert\log\rho(x)-\log\rho(y)\big\rvert,\\
\mathrm{AUC}_{\mathrm{len}}&=\Pr\big(s(x)>s(x')\big)+\tfrac12\Pr\big(s(x)=s(x')\big).
\end{aligned}\label{eq:lenmia}
\end{equation}
In this test, $\mathcal{S}_k$ represents the synthetic records of class $k$, $x$ denotes a training record of class $k$, and $x'$ denotes a validation or test record of the same class, both drawn uniformly at random. The area under the receiver operating characteristic curve $\mathrm{AUC}_{\mathrm{len}}$ equals $0.5$ when the lengths carry no information about membership and $1$ when they separate the training records from the other records perfectly. The score of Eq.~\eqref{eq:lenmia} describes an adversary that looks for the nearest released length. A stronger adversary also knows the bandwidth $\tau_k$ and scores a candidate with the kernel density of the released log-lengths at its own log-length,
\begin{equation}
\begin{aligned}
s_{\mathrm{LR}}(x)&=\log\sum_{y\in\mathcal{S}_k}\exp\!\left(-\frac{\Delta(x,y)^{2}}{2\tau_k^{2}}\right),\\
\Delta(x,y)&=\log\rho(x)-\log\rho(y),
\end{aligned}\label{eq:lenlr}
\end{equation}
which approximates the likelihood-ratio test under the smoothed bootstrap, since the log-length of every training record is a center of the mixture from which the released log-lengths are drawn (Proposition~\ref{prop:length}). We report both adversaries, with bootstrap intervals of the AUC over the candidates.

For an adversary that also obtains the trained generator, we evaluate a loss-based membership test \citep{yeom2018privacy}, which exploits a lower training loss on members. The score of a record $x$ of class $k$ is the negative loss of the generator, averaged over the $T$ steps and $R$ fixed noise draws,
\begin{equation}
s_{\mathrm{loss}}(x)=-\frac{1}{TR}\sum_{t=1}^{T}\sum_{r=1}^{R}\ell\big(\theta;a(x),k,t,\epsilon_r\big).\label{eq:lossmia}
\end{equation}
The loss $\ell$ is $\ell_{\mathrm{ps}}$ or $\ell_{\mathrm{tr}}$ for the quantum denoisers and the corresponding loss for the classical generators of Sections~\ref{sec:control} and \ref{sec:setup_ddpm}. The members are a fixed random subset of $\mathcal{D}_{\mathrm{tr}}$ of the size of $\mathcal{D}_{\mathrm{te}}$, and the non-members are the records of $\mathcal{D}_{\mathrm{te}}$. Finally, we evaluate the classification measures on the round trips of real records as well, which gives the ceiling of each codec described in Section~\ref{sec:encoding}.

\subsection{Implementation}\label{sec:impl}

We implement the denoiser twice. The first implementation, written in PyTorch \citep{paszke2019pytorch}, builds the unitary matrix of every block explicitly and applies it to a whole batch of states; we employ it for training and for state-vector sampling. The second implementation applies the same gates one by one in PennyLane \citep{bergholm2018pennylane}; we use it to verify the first implementation and to run the density-matrix sampling of Section~\ref{sec:noise}. Both implementations share the same parameters, and their output states agree to within $5\times10^{-7}$. Both implementations support data-parallel execution with the message passing interface (MPI), in which every run is assigned to a group of processes that split each mini-batch and average their gradients after every step; Section~\ref{sec:setup_env} describes how we executed the reported runs.
\subsection{Theoretical analysis}\label{sec:theory}

In this subsection, we establish a theorem and two propositions that support the design choices of Sections~\ref{sec:encoding} to \ref{sec:sampling}: the training target, the initialization together with the acceptance probability that it implies, and the lengths of the decoded records, which underlie the privacy analysis of Section~\ref{sec:res_privacy}. The results hold exactly for the circuits of Section~\ref{sec:denoiser}, since every angle $\theta$ enters a single rotation $e^{-i\theta Q/2}$ generated by a Pauli operator $Q$. Two further results, on the smoothness of the training objective and on the number of measurements that the read-out and the loss require on hardware, follow from standard arguments; we state them as Lemmas~\ref{lem:smooth} and \ref{lem:shots} in Appendix~\ref{app:proofs}. Appendix~\ref{app:proofs} contains all proofs, and Appendix~\ref{app:numerics} verifies every statement numerically with our implementation (Table~\ref{tab:numerics}).

\paragraph{Optimal training target}
The first result compares our posterior-mean target with a single draw of $x_{t-1}$, the target of the training algorithm of \citet{cacioppo2025hybrid}.

\begin{theorem}[Optimal target]\label{thm:target}
Let $x_{t-1}$ be a draw from the posterior $q(x_{t-1}\mid x_t,x_0)$ of the forward process of Eq.~\eqref{eq:forward}, and let $\ket{\hat x_{t-1}}=x_{t-1}/\lVert x_{t-1}\rVert_2$. There exist constants $\lambda_1>\lambda_2\geq0$ with $\lambda_1+(D-1)\lambda_2=1$, which depend on $x_t$ and $x_0$ only, such that for every density matrix $\varrho$ of the data register,
\begin{equation}
\begin{aligned}
&\mathbb{E}\Big[1-\bra{\hat x_{t-1}}\varrho\ket{\hat x_{t-1}}\;\Big|\;x_t,x_0\Big]\\
&\quad=(1-\lambda_1)+(\lambda_1-\lambda_2)\Big(1-\bra{\tilde\mu_t}\varrho\ket{\tilde\mu_t}\Big).
\end{aligned}\label{eq:thm_target}
\end{equation}
\end{theorem}
In this identity, the left-hand side represents the expected infidelity of an output $\varrho$ when the target is a single draw, and the last factor is our loss of Eq.~\eqref{eq:loss_ps} or Eq.~\eqref{eq:loss_tr} with the posterior-mean target. Since $\lambda_1-\lambda_2>0$, the expected single-draw loss is an increasing affine function of our loss, so both losses are minimized by the same output, $\varrho=\ket{\tilde\mu_t}\bra{\tilde\mu_t}$, for every pair $(x_t,x_0)$. Over the training set, the two objectives differ by the positive weight $\lambda_1-\lambda_2$ of each pair, and they share their minimizers whenever the circuit can represent the optimal output for every input. In other words, the posterior-mean target leaves the optimal denoiser unchanged and removes the variance of the draw, which is substantial: for our schedule, a single draw has an expected fidelity of only $\lambda_1\approx0.7$ with $\ket{\tilde\mu_t}$ (Appendix~\ref{app:numerics}).

\paragraph{Identity initialization}
The second result quantifies how close the denoiser is to the identity map at initialization.

\begin{proposition}[Identity initialization]\label{prop:init}
Let the $N=3mL$ angles of an SEC block with an even number $L$ of layers on $m$ qubits be independent and distributed as $\mathcal{N}(0,\sigma^2)$. Then $U(0)=I$, and for every state $\ket{v}$,
\begin{equation}
\begin{aligned}
\mathbb{E}\big[U(\Theta)\big]&=e^{-N\sigma^2/8}\,I,\\
\mathbb{E}\big\lVert U(\Theta)\ket{v}-\ket{v}\big\rVert_2^{2}&=2\big(1-e^{-N\sigma^2/8}\big).
\end{aligned}\label{eq:prop_init}
\end{equation}
With $N_2=3(n+1)L_2$ angles in $U_2^{(t)}$, the ancilla outcome $1$ also satisfies $\mathbb{E}[p_1\mid\ket{x_t}]\leq2(1-e^{-N_2\sigma^2/8})$ at every step, and the acceptance probability of Eq.~\eqref{eq:acceptance} satisfies $\mathbb{E}[P_{\mathrm{acc}}]\geq(2e^{-N_2\sigma^2/8}-1)^{T}$ whenever the base is non-negative.
\end{proposition}
In this proposition, $N\sigma^2$ represents the total variance of the angles of a block, which controls its distance from the identity. With $\sigma=0.1$, the expected unitary of our blocks ($N=60$ and $N=72$) is $0.93I$ and $0.91I$, and the acceptance probability at initialization is at least $0.15$ for $T=10$. This bound is loose, since it holds for every input state and every step: at initialization, the measured acceptance probability equals $0.911$ (Appendix~\ref{app:numerics}). With the standard deviation of $0.7$ of the public implementation, the expected unitaries drop to $0.03I$ and $0.01I$, so that the blocks start far from the identity and the bound on $P_{\mathrm{acc}}$ no longer applies.

\paragraph{Lengths of decoded records}
The third result characterizes the lengths that the decoder of Eq.~\eqref{eq:decode} gives to the synthetic records.

\begin{proposition}[Lengths of decoded records]\label{prop:length}
With the min--max codec, every decoded record $\hat z$ of class $k$ satisfies $\rho(\hat z)=\rho$, where $\rho=\rho_j e^{\tau_k\zeta}$ is the length drawn in Eq.~\eqref{eq:decode}, whatever the generated state. Consequently, the following statements hold for the lengths of the synthetic records of class $k$.
(i) If $\tau_k>0$, then $\log\rho(\hat z)$ has the density
\begin{equation}
\hat f_k(v)=\frac{1}{\lvert\mathcal{R}_k\rvert}\sum_{r\in\mathcal{R}_k}\frac{1}{\sqrt{2\pi}\,\tau_k}\exp\!\left(-\frac{(v-\log r)^{2}}{2\tau_k^{2}}\right),\label{eq:kde}
\end{equation}
whose mean equals the mean $\bar v_k$ of $\{\log r:r\in\mathcal{R}_k\}$ and whose variance equals $(1-1/\lvert\mathcal{R}_k\rvert)\,s_k^{2}+\tau_k^{2}$.
(ii) If $\tau_k>0$, then $\Pr\big(\rho(\hat z)=r\big)=0$ for every $r>0$, so that no synthetic record reproduces the exact length of a training record.
(iii) If $\tau_k=0$, the lengths follow the empirical distribution of $\mathcal{R}_k$, and a training length that occurs once in $\mathcal{R}_k$ appears among the lengths of $n_s$ synthetic records of class $k$ with probability $1-(1-1/\lvert\mathcal{R}_k\rvert)^{n_s}$.
\end{proposition}
In this proposition, $\hat f_k$ represents the Gaussian kernel density estimate of the logarithms of the training lengths with bandwidth $\tau_k$, in which every element of $\mathcal{R}_k$ counts with its multiplicity, and $s_k$ denotes the sample standard deviation of Eq.~\eqref{eq:bandwidth}. Part (i) indicates that the synthetic lengths keep the distribution of the training lengths up to a smoothing of width $\tau_k$ on the logarithmic scale, so that the length of a synthetic record cannot serve as a class label beyond what the real training data already contain. With the bandwidth of Eq.~\eqref{eq:bandwidth} and the min--max codec, this smoothing increases the standard deviation of the logarithms by $0.9\%$ to $4.9\%$ for the classes of this study. Parts (ii) and (iii) explain why the decoder perturbs the drawn length. Part (ii) is a property of every continuous distribution and gives no guarantee by itself, since a released length close to a training length can still reveal membership; Section~\ref{sec:res_privacy} therefore measures this risk with two adversaries, also for the lengths drawn from the differentially private histogram of Eq.~\eqref{eq:dplengths}. With $\tau_k=0$, the release reproduces the exact lengths of training records: for the Startpage family, whose $97$ training records have distinct lengths, $512$ synthetic records contain the length of $99.5\%$ of them. An adversary who knows the codec can then compare these values with the length of a candidate record, and the membership test of Eq.~\eqref{eq:lenmia} reaches an AUC of $0.998$ for Startpage. With the bandwidth of Eq.~\eqref{eq:bandwidth}, no length is reproduced, and the AUC drops to $0.55$, close to the value $0.5$ of a random guess (Appendix~\ref{app:numerics}). Since the decoder is applied after sampling, the choice of $\tau_k$ requires no retraining of the denoiser. With the sign-split codec, the decoded length is at most $\rho$, and it equals $\rho$ when the two halves of $\lvert\hat\alpha\rvert$ have disjoint supports.

\section{Experimental setup}\label{sec:setup}

In this section, we describe the data, the configurations of the trained models, the reference points of the evaluation, the classical diffusion baseline, the evaluation protocol, and the computing environment of our experiments.

\subsection{Dataset}\label{sec:setup_data}

We employ a corpus of malware executables of $20$ families from the public repository VirusShare \citep{virusshare}, which was collected and preprocessed in the same way as the corpus of \citet{bao2025generating}. The corpus contains the $17$ most common families of the collection together with three of its rarest families, so that it reflects the class imbalance of real collections and contains families with few samples. The samples were collected in August 2023, and every sample keeps the malware-family label provided with it by VirusShare, without modification. Unlike the corpus of \citet{bao2025generating}, which combines samples of $25$ families from the Malicia dataset \citep{nappa2015malicia} and from VirusShare, our corpus draws every family from VirusShare. We disassembled every executable with objdump from the GNU Binutils \citep{binutils} and the pipeline of \citet{bao2025generating}, which keeps the lines of the disassembly that contain an instruction, extracts the instruction field, and removes its operands, so that the output lists the opcode mnemonics of the program in their order, one per line. Executables whose disassembly does not produce a usable instruction stream, for instance because of packing or anti-disassembly protections, are excluded without any attempt at repair, and exact duplicates are removed, which leaves $21{,}522$ programs. Since the CAE of Section~\ref{sec:latent} reads the first $7{,}000$ opcodes of a program, we further keep one program of every group whose first $7{,}000$ opcodes coincide, which leaves $19{,}843$ programs; $83\%$ of them contain more than $7{,}000$ opcodes, with a median of $28{,}711$. We discard the three families with fewer than $30$ programs (Diplugem, Lamechi.B, and Systex.A), which leaves $K=17$ families with $148$ to $4{,}158$ programs each. For every family, we draw $20\%$ of the programs at random for the test set and $15\%$ for the validation set and keep the rest for training, which yields $12{,}892$, $2{,}973$, and $3{,}969$ records; the smallest family, Startpage, has $97$ training, $22$ validation, and $29$ test records. The split depends on a single fixed seed, so that every training run observes the same records.

We fix this split before training any model, the CAE included. The vocabulary of the CAE consists of the $1{,}185$ opcodes that occur in the truncated training records, together with a padding token and an unknown token, which replaces the $138$ opcode occurrences of validation and test records outside the vocabulary. We train the CAE on the training records for $120$ epochs with Adam at a learning rate of $10^{-5}$ and mini-batches of $64$ records, and we replace every opcode of its input by the unknown token with a probability of $0.15$, so that the CAE learns to reconstruct corrupted sequences. The validation records only monitor this training, and the test records reach the CAE only after training, when its encoder maps them onto latent vectors. The trained CAE recovers $32.4\%$, $31.6\%$, and $31.7\%$ of the opcodes of the training, validation, and test records, with cross-entropies of $3.05$, $3.11$, and $3.11$. Hence, the CAE reconstructs unseen records nearly as well as its training records, and since all generative models receive the same latent vectors, the representation affects every compared model in the same way.

\subsection{Configurations of the quantum denoiser}\label{sec:setup_config}

We train the quantum denoiser with the linear schedule of Section~\ref{sec:forward} over $T=10$ steps, with $\beta$ increasing from $0.02$ to $0.10$, and with $L_1=L_2=L_3=4$ layers per block, which gives the $11{,}520$ angles of Table~\ref{tab:params} for the min--max codec and $13{,}800$ angles for the sign-split codec. We draw the initial angles from $\mathcal{N}(0,0.1^2)$, so that every block starts close to the identity (Proposition~\ref{prop:init}). We optimize Eq.~\eqref{eq:objective} with Adam at a learning rate of $0.01$ and a weight decay of $10^{-5}$, and we multiply the learning rate by $0.97$ every $10$ epochs. Every mini-batch contains $128$ training records, each combined with all $T$ steps, so that one update uses $1{,}280$ pairs $(x_t,\tilde\mu_t)$. Every five epochs, we evaluate the infidelity on fixed validation pairs, which combine every validation record with every step and keep the noise fixed across evaluations. We stop training after eight evaluations without a relative improvement of $0.1\%$, with at least $40$ and at most $400$ epochs, and we keep the parameters of the best evaluation, which we then evaluate on fixed test pairs.

We train six configurations: the denoiser with post-selection, the denoiser without branch selection, and the classical control of Section~\ref{sec:control}, each with the min--max and the sign-split codec. Every configuration runs with the ten seeds $0$ to $7$, $13$, and $42$, which set the initial parameters, the order of the mini-batches, the noise of the forward process, and the initial states of sampling. After training, every run generates $512$ records per family with Algorithm~\ref{alg:sample}, that is, $8{,}704$ records per run. We train the four configurations of the circuit with both implementations of Section~\ref{sec:impl}, the PennyLane implementation with the seeds $42$, $7$, and $13$ for verification, and report the PyTorch runs, which gives $60$ PyTorch runs and $12$ PennyLane runs in total. For the simulation of hardware noise (Section~\ref{sec:noise}), we sample the $20$ post-selected circuits for all families with $p\in\{0,0.001,0.01\}$ and with an exact read-out or $S=1{,}000$ shots, which gives $120$ sampling configurations of $512$ records per family. We also sample every post-selected circuit and every classical control with the two samplers of Eq.~\eqref{eq:samplers}, and we train the post-selected circuit with the penalized objective of Eq.~\eqref{eq:pacc_loss} for $\lambda\in\{0.001,0.01,0.1\}$, with both codecs and the seeds $42$, $7$, and $13$. Finally, to test whether early stopping ends the training of the circuit too early, we repeat the min--max run with the seed $13$, which stops at the highest test infidelity, with a patience of $16$ instead of eight evaluations.

\subsection{Reference points}\label{sec:setup_refs}

We interpret the scores of the quantum denoiser against reference points that share its latent vectors, its amplitude encoding, its decoder, and its evaluation protocol, so that a difference in scores can be attributed to the denoiser. Generators of other designs, such as classical or hybrid diffusion models with $1{,}000$ stochastic steps on unnormalized latent vectors, differ from our model in several of these components at once, and a comparison with them would therefore confound the contribution of the circuit with that of the representation and of the sampling procedure. We therefore employ such a generator only to place the complete framework (Section~\ref{sec:setup_ddpm}). We do not include hybrid quantum diffusion models, such as the quantum latent diffusion model of \citet{defalco2024latent}, in which a classical network carries the state between the steps, since their scores would not isolate the contribution of a circuit (Section~\ref{sec:rel_quantum}). We isolate the denoiser with five reference points:
\begin{itemize}
\item \emph{Identity map}: the map that returns the noisy input state unchanged, which is the starting point of every denoiser (Proposition~\ref{prop:init}) and the reference of the test infidelity.
\item \emph{Random states}: the magnitudes of the initial states $\ket{z_T}$ of Algorithm~\ref{alg:sample}, decoded with Eq.~\eqref{eq:decode} without any denoising step. These are the records that the identity map would release, and they mark the scores of a generator that has learned nothing. For every seed, we decode the $512$ initial states per family from which the trained denoisers of that seed start, so that every comparison between a denoiser and the random states is paired.
\item \emph{Round-trip ceiling}: the scores of the encoded and decoded real records of Section~\ref{sec:encoding}, which a perfect generator would reach with the same codec and decoder.
\item \emph{Real records}: the classifiers trained on $\mathcal{X}_{\mathrm{tr}}$ alone for the utility, the real test records, which set the scale of the memorization measures of Eq.~\eqref{eq:mem}, and the encoded test records for the class information of Section~\ref{sec:res_states}.
\item \emph{Parameter-matched classical control}: the classical denoiser of Section~\ref{sec:control}, which replaces the circuit with an MLP of the same size and keeps every other component, so that it separates the contribution of the circuit from that of the formulation of the generative task.
\end{itemize}
The random states and the round-trip ceiling bound the scores that a denoiser of this framework can attain, so that the position of a configuration between them indicates how much of the attainable class information its denoiser has learned. The classical control then indicates how much of the remaining gap a classical map of the same size closes under the same formulation.

\subsection{Classical diffusion baseline}\label{sec:setup_ddpm}

To place the fully quantum framework as a whole, including its encoding, its ten deterministic steps, and its decoder, we compare it with a classical latent DDPM in the design that \citet{bao2025generating} employed for embeddings of malware opcode sequences. The DDPM follows \citet{ho2020denoising} with $T=1{,}000$ steps and a linear schedule of $\beta_t$ from $10^{-4}$ to $0.02$, and it operates on the latent vectors scaled to $[-1,1]$ with the minimum and the maximum of every feature over the training records. A one-dimensional U-Net predicts the noise from the scaled vector, which it reads as a sequence of $32$ values. The U-Net consists of residual blocks with group normalization, SiLU activations, and convolutions of kernel size three, with stride-two downsampling, nearest-neighbor upsampling, and skip connections, and every block receives a sinusoidal embedding of the step added to a learned embedding of the family. Like the quantum denoiser, one model covers all families. We train two sizes: a full U-Net with $32$, $64$, and $128$ channels, two residual blocks per level, a dropout rate of $0.1$, and $1{,}185{,}089$ parameters, and a small U-Net with $7$ and $14$ channels, one residual block per level, and $11{,}404$ parameters, comparable to the $11{,}520$ angles of the quantum denoiser. We train both with Adam at a learning rate of $10^{-3}$ on mini-batches of $256$ records, keep an exponential moving average of the weights with a decay of $0.999$, evaluate the loss of this average on fixed validation pairs every $500$ steps, and stop after ten evaluations without a relative improvement of $0.1\%$, with at least $5{,}000$ and at most $200{,}000$ steps. Every size runs with the ten seeds of the quantum denoiser and generates $512$ records per family with the ancestral sampler of \citet{ho2020denoising}, which starts from Gaussian noise. \citet{bao2025generating} instead started the reverse process from training records noised to step $T/2$. We report this sampler as well, since every record it generates derives from an individual training record, which tests whether the memorization measures of Section~\ref{sec:measures} detect such a dependence.

\subsection{Evaluation protocol}\label{sec:setup_eval}

We evaluate every run with the measures of Section~\ref{sec:measures} on its $512$ records per family. The random forest uses $100$ trees, the SVM uses a radial basis function kernel with $C=1$ and a kernel width set from the variance of the features, and the MLP has two hidden layers of $128$ and $64$ units trained for at most $500$ iterations. The SVM and the MLP receive features standardized with the statistics of their training data. For every run, we average the macro-F1 scores of the three classifiers, and we then report the mean and the standard deviation over the ten seeds. We compute the round-trip ceiling of each codec with the same classifiers and average it over eight draws of the decoded lengths, since the decoder of Eq.~\eqref{eq:decode} is random. We compare paired configurations, such as the circuit and its classical control or the circuit and the random states of its codec, seed by seed with the two-sided Wilcoxon signed-rank test \citep{wilcoxon1945individual}; with ten seeds, its smallest attainable $p$-value equals $0.002$, which requires the difference to have the same sign in every seed. For the loss-based membership test, we employ $R=4$ noise draws per record and step for the quantum denoisers and the classical control, and $32$ draws of the step and the noise for the DDPM. We compute the length-based AUCs per family and average them over the $17$ families and the ten seeds, and we compute the loss-based AUC over the pooled records of all families, with a bootstrap interval over the records. For the private length histogram of Eq.~\eqref{eq:dplengths}, we employ $B=64$ bins and $\varepsilon\in\{0.1,0.3,1,3,10\}$, and we average its round-trip ceiling over three draws of the lengths.

\subsection{Computing environment}\label{sec:setup_env}

We carried out all experiments on classical hardware running quantum simulators, since fault-tolerant quantum processors that could train the denoiser are not yet available in the present NISQ era \citep{preskill2018quantum}. We performed all runs on the REmotely-managed Power Aware Computing Systems and Services (REPACSS) supercomputer at Texas Tech University \citep{repacss}, whose CPU nodes each provide two AMD EPYC 9754 processors with $256$ cores and whose GPU nodes each provide two Intel Xeon Gold 6448Y processors and four NVIDIA H100 NVL accelerators. We employed the CPU nodes for every training and sampling run of the diffusion models, because the largest matrices of the denoiser have $128\times128$ entries, for which a GPU brings no speed-up at our mini-batch size. We trained the CAE on one H100 accelerator in $42$~min, and the numerical verification of Appendix~\ref{app:numerics} ran on the processors of a GPU node.

We submitted the runs with the SLURM workload manager and executed them with MPI (Open MPI~4.1.8 with mpi4py~4.1.2). Every training run of the quantum denoiser and of the classical control was assigned to a group of four processes with two CPU threads each, which split every mini-batch and averaged their gradients, and the runs of one job trained concurrently. The $18$ PyTorch runs of the seeds $42$, $7$, and $13$ completed within $1.0$~h on one CPU node, the $42$ PyTorch runs of the seeds $0$ to $6$ within $2.3$~h on two nodes, and the $12$ PennyLane runs within $1.1$~h on one node. A post-selected run took $41$~min on average with the min--max codec and $79$~min with the sign-split codec, against about $2$~min for the classical control. The $120$ sampling configurations with hardware noise ran on $512$ processes of four CPU nodes, with a group of four or five processes per configuration that split its samples, and took $1.1$~h. The $20$ DDPM runs ran concurrently within one job on one CPU node, with $16$ threads for every full U-Net and four threads for every small one. Training took $101$ to $215$~min for the full U-Net and $40$ to $65$~min for the small U-Net, and the $1{,}000$ steps of the ancestral sampler generated the $8{,}704$ records of a run in $6.6$ to $8.3$~min and in $1.6$ to $2.0$~min, respectively. We implemented every model in Python~3.12 with PyTorch~2.5.1 \citep{paszke2019pytorch}, PennyLane~0.42.3 \citep{bergholm2018pennylane}, and scikit-learn~1.7.2 \citep{pedregosa2011scikit}.

\section{Results}\label{sec:results}

In this section, we evaluate the quantum denoiser with both codecs and both measurement rules against the reference points of Section~\ref{sec:setup_refs} and the classical diffusion baseline of Section~\ref{sec:setup_ddpm}. Unless stated otherwise, every value is the mean and the standard deviation over the ten seeds, the scores cover all $17$ families, the records come from the PyTorch implementation, and the decoder draws the smoothed lengths of Eq.~\eqref{eq:decode}. Section~\ref{sec:res_training} describes the training of the denoisers, Sections~\ref{sec:res_utility} and \ref{sec:res_privacy} evaluate the generated records, Section~\ref{sec:res_states} examines the generated states, Section~\ref{sec:res_sampling} varies the sampler, the weighting of the released records, and the training objective, and Section~\ref{sec:res_noise} reports the behavior of post-selected sampling under hardware noise.

\subsection{Training of the denoisers}\label{sec:res_training}

\begin{figure*}[t]
\centering
\includegraphics{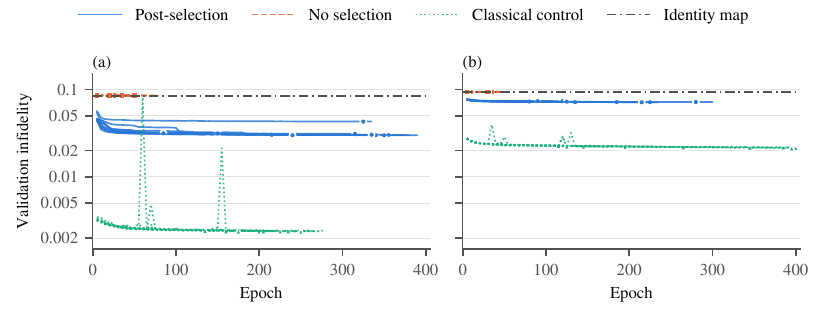}
\caption{Validation infidelity during training for the min--max codec (a) and the sign-split codec (b), with one line for each of the ten seeds. Blue solid lines: circuit with post-selection; orange dashed lines: circuit without branch selection; green dotted lines: parameter-matched classical control; gray dash-dotted line: identity map, which returns the noisy input state unchanged. Markers indicate the best epoch of each run. The models are evaluated every five epochs on fixed validation pairs of all steps $t=1,\dots,10$, with $T=10$, the linear schedule from $0.02$ to $0.10$, Adam with a learning rate of $0.01$, and early stopping after eight evaluations without improvement. The vertical axis is logarithmic.}
\label{fig:curves}
\end{figure*}

\begin{table*}[t]
\small
\caption{Trained denoisers with the two codecs, mean and standard deviation over the ten seeds. Best epoch: epoch of the lowest validation infidelity; reduction: decrease of the test infidelity relative to the identity map; $\bar p_0$: mean probability of the ancilla outcome $0$ on the test pairs; $P_{\mathrm{acc}}$: mean acceptance probability of Eq.~\eqref{eq:acceptance} along the sampling trajectories, averaged over the $512$ initial states of every family and over the $17$ families and the seeds. Without branch selection, no sample is rejected, so that $P_{\mathrm{acc}}=1$; the classical control has no ancilla.}\label{tab:training}
\begin{tabular*}{\tblwidth}{@{}LLRRRRRR@{}}
\toprule
Codec & Denoiser & Parameters & Best epoch & Test infidelity & Reduction & $\bar p_0$ & $P_{\mathrm{acc}}$\\
\midrule
Min--max & Identity map & -- & -- & $0.0842$ & -- & -- & --\\
 & Post-selection & $11{,}520$ & $252\pm95$ & $0.0315\pm0.0040$ & $62.6\%$ & $0.585\pm0.039$ & $3.0\times10^{-7}$\\
 & No selection & $11{,}520$ & $10\pm12$ & $0.0848\pm0.0000$ & $-0.8\%$ & $1.000$ & $1$\\
 & Classical control & $11{,}458$ & $198\pm36$ & $0.00225\pm0.00002$ & $97.3\%$ & -- & --\\
\midrule
Sign-split & Identity map & -- & -- & $0.0929$ & -- & -- & --\\
 & Post-selection & $13{,}800$ & $130\pm66$ & $0.0721\pm0.0009$ & $22.4\%$ & $0.712\pm0.021$ & $9.8\times10^{-5}$\\
 & No selection & $13{,}800$ & $5\pm0$ & $0.0935\pm0.0000$ & $-0.7\%$ & $1.000$ & $1$\\
 & Classical control & $13{,}862$ & $270\pm120$ & $0.0219\pm0.0004$ & $76.5\%$ & -- & --\\
\bottomrule
\end{tabular*}
\end{table*}

Figure~\ref{fig:curves} illustrates the validation infidelity during training, and Table~\ref{tab:training} summarizes the trained denoisers. With post-selection, the validation infidelity of the circuit exhibits a steep decline during the first $20$ epochs, followed by a gradual decline until early stopping ends the run, after $125$ to $390$ epochs with the min--max codec and after $105$ to $300$ epochs with the sign-split codec. On the test pairs, the circuit reaches an infidelity of $0.0315\pm0.0040$ with the min--max codec and $0.0721\pm0.0009$ with the sign-split codec, which lie $62.6\%$ and $22.4\%$ below the infidelity of the identity map, the map that returns the noisy input state unchanged. The spread with the min--max codec stems from the seed $13$, whose run stops at $0.0426$, whereas the other nine seeds reach $0.0295$ to $0.0312$, that is, $63\%$ to $65\%$ below the identity map. With a patience of $16$ instead of eight evaluations, the run of the seed $13$ continues to the limit of $400$ epochs and reaches $0.0424$, so that this seed converges to a worse solution and early stopping does not explain its higher infidelity. By contrast, the circuit without branch selection remains at the level of the identity map: its best validation infidelity occurs after $5$ to $40$ epochs, early stopping ends every run after $45$ to $80$ epochs, and the ancilla returns the outcome $0$ with probability $1.000$. The circuit therefore stays at the initialization of Proposition~\ref{prop:init}, which suggests that, without selection, the measurement only mixes the two output branches and that no mixture of them brings the output closer to the target than the unchanged input. The parameter-matched classical control, with $11{,}458$ parameters against the $11{,}520$ angles of the circuit, reaches $0.00225\pm0.00002$ and $0.0219\pm0.0004$, which are lower than the infidelities of the post-selected circuit by factors of $14$ and $3.3$, and it trains in about $2$~min, against $41$ and $79$~min for the circuit (Section~\ref{sec:setup_env}).

The PennyLane implementation reproduces these runs. Over the $12$ runs that both implementations trained, the test infidelities differ by at most $8.1\times10^{-8}$, every pair of runs selects the same best epoch, and the utility scores of Section~\ref{sec:res_utility} for the PennyLane runs agree with those of the PyTorch runs of the same seeds, with a mean TSTR difference of $-0.002\pm0.009$ for the post-selected circuit with the min--max codec. Since both implementations ran with the same groups of processes and the same seeds, they also drew the same noise during training.

The mean probability $\bar p_0$ of the ancilla outcome $0$ on the test pairs equals $0.585$ with the min--max codec and $0.712$ with the sign-split codec, but it hides a strong dependence on the step. Along the sampling trajectories of Algorithm~\ref{alg:sample}, $p_0^{(t)}$ stays above $0.95$ for $t=10$ to $5$ and falls to $0.05$, $0.03$, $0.01$, and $0.006$ for $t=4$ to $1$ (min--max codec, seed $42$, Startpage). As a result, the acceptance probability of Eq.~\eqref{eq:acceptance} drops from $0.91$ at initialization (Appendix~\ref{app:numerics}) to $3.0\times10^{-7}$ after training with the min--max codec and to $9.8\times10^{-5}$ with the sign-split codec. In other words, the trained circuit achieves its denoising in the last steps through ancilla outcomes of low probability, and a quantum device would need about $3.4\times10^{6}$ and $1.0\times10^{4}$ runs per accepted record.

\subsection{Utility of the synthetic records}\label{sec:res_utility}

\begin{table*}[t]
\small
\caption{Classification utility (macro-F1, mean of the random forest, the SVM, and the MLP) of the synthetic records of all $17$ families, mean and standard deviation over the ten seeds. The classifiers trained on the real records alone reach $0.635$. The round-trip ceilings of Section~\ref{sec:encoding} are averaged over eight length draws. Random states: the initial states of sampling, decoded without any denoising step (Section~\ref{sec:setup_refs}). Classical latent DDPM: the baseline of Section~\ref{sec:setup_ddpm}, which generates the latent vectors directly and samples from Gaussian noise or, with the sampler of \citet{bao2025generating}, from training records noised to step $T/2$.}\label{tab:utility}
\begin{tabular*}{\tblwidth}{@{}LLRRRR@{}}
\toprule
Model & Variant & Parameters & Aug & TSTR & TRTS\\
\midrule
Round-trip ceiling & Min--max & -- & $0.634\pm0.002$ & $0.605\pm0.002$ & $0.537\pm0.003$\\
 & Sign-split & -- & $0.627\pm0.003$ & $0.593\pm0.003$ & $0.521\pm0.004$\\
Random states & Min--max & -- & $0.581\pm0.002$ & $0.076\pm0.014$ & $0.020\pm0.001$\\
 & Sign-split & -- & $0.626\pm0.006$ & $0.039\pm0.007$ & $0.038\pm0.001$\\
\midrule
Quantum denoiser, min--max & Post-selection & $11{,}520$ & $0.593\pm0.005$ & $0.114\pm0.015$ & $0.028\pm0.002$\\
 & No selection & $11{,}520$ & $0.581\pm0.003$ & $0.070\pm0.007$ & $0.020\pm0.001$\\
 & Classical control & $11{,}458$ & $0.624\pm0.005$ & $0.218\pm0.025$ & $0.167\pm0.016$\\
Quantum denoiser, sign-split & Post-selection & $13{,}800$ & $0.625\pm0.002$ & $0.036\pm0.006$ & $0.041\pm0.001$\\
 & No selection & $13{,}800$ & $0.626\pm0.004$ & $0.041\pm0.007$ & $0.038\pm0.001$\\
 & Classical control & $13{,}862$ & $0.617\pm0.004$ & $0.367\pm0.008$ & $0.317\pm0.017$\\
\midrule
Classical latent DDPM & Full U-Net & $1{,}185{,}089$ & $0.631\pm0.003$ & $0.531\pm0.006$ & $0.529\pm0.009$\\
 & Small U-Net & $11{,}404$ & $0.620\pm0.003$ & $0.442\pm0.007$ & $0.381\pm0.010$\\
 & Full U-Net, start at $T/2$ & $1{,}185{,}089$ & $0.632\pm0.003$ & $0.531\pm0.006$ & $0.528\pm0.009$\\
\bottomrule
\end{tabular*}
\end{table*}

Table~\ref{tab:utility} places the classification utility of the generated records between the two bounds of Section~\ref{sec:setup_refs}, the decoded random states and the round-trip ceiling of each codec. The random states reach TSTR scores of $0.076\pm0.014$ and $0.039\pm0.007$, close to the macro-F1 score of $0.051$ that uniformly random labels obtain on the imbalanced real test records of the $17$ families. We attribute the small excess of the min--max codec to the decoded lengths, which follow the training lengths of each family (Proposition~\ref{prop:length}) and therefore carry a small amount of class information. Their TRTS scores of $0.020\pm0.001$ and $0.038\pm0.001$ lie below the value of $0.059$ of uniformly random labels on the balanced synthetic records, which indicates that the classifiers trained on real records assign most random states to a few families. The ceilings lie at $0.605$ and $0.593$ in TSTR and at $0.537$ and $0.521$ in TRTS, which leaves a wide range in which a denoiser can add class information to the records.

Within this range, the records of the post-selected circuit remain close to the random states. With the min--max codec, TSTR reaches $0.114\pm0.015$, which exceeds the random states in every seed, by $0.039\pm0.023$ (Wilcoxon $p=0.002$), and TRTS reaches $0.028\pm0.002$ against $0.020\pm0.001$, again in every seed. The circuit therefore covers $7\%$ of the TSTR range between the random states and the ceiling. With the sign-split codec, the Aug and TSTR scores of the circuit do not differ from those of the random states, with differences of $-0.002\pm0.005$ ($p=0.38$) and $-0.004\pm0.008$ ($p=0.28$), and TRTS exceeds them by $0.003\pm0.001$. The circuit without branch selection scores like the random states with both codecs, with differences of at most $0.006$ and $p\geq0.10$ in every setting, in agreement with its training at the identity map.

No generator of this study raises the Aug score above the real-only score of $0.635$. With the min--max codec, the records of the post-selected circuit lower it to $0.593$, which is less than the drop to $0.581$ caused by the random states, by $0.012\pm0.007$ in every seed, and with the sign-split codec, they lower it to $0.625$, as the random states do ($0.626$). Synthetic records that carry little class information thus act as label noise for the classifiers, and the size of this cost depends more on the codec than on the denoiser.

The classical control indicates how much of this range a classical map of the same size covers under the same formulation, since it shares the loss, the ten steps, the initial states, and the decoder of the circuit. With the min--max codec, the control reaches TSTR and TRTS scores of $0.218\pm0.025$ and $0.167\pm0.016$, and with the sign-split codec, it reaches $0.367\pm0.008$ and $0.317\pm0.017$, so that it covers $27\%$ and $59\%$ of the range between the random states and the ceiling in TSTR. The control exceeds the post-selected circuit in every seed, in TSTR by $0.104\pm0.027$ with the min--max codec and $0.332\pm0.009$ with the sign-split codec, and in TRTS by $0.138\pm0.015$ and $0.276\pm0.016$. Hence, the min--max codec restricts both denoisers, whereas with the sign-split codec a classical map of the same size recovers a large part of the class information that remains out of reach of the circuit.

The classical latent DDPM indicates the utility that a generator of the same latent vectors attains without the amplitude encoding and the ten deterministic steps. The full U-Net reaches TSTR and TRTS scores of $0.531\pm0.006$ and $0.529\pm0.009$, and the small U-Net, whose $11{,}404$ parameters are comparable to the $11{,}520$ angles of the circuit, reaches $0.442\pm0.007$ and $0.381\pm0.010$. The small U-Net exceeds the post-selected circuit with the min--max codec in every seed, by $0.328\pm0.017$ in TSTR and $0.353\pm0.011$ in TRTS, and it also exceeds the classical control with the sign-split codec in every seed, by $0.075\pm0.011$ and $0.064\pm0.019$. With the sign-split codec, the classical control therefore comes within $0.08$ of the small U-Net in TSTR, whereas the post-selected circuit remains $0.41$ below it, so that most of the gap between the circuit and the classical diffusion model lies between the circuit and a classical map under the same formulation. Adding the DDPM records to the real training data lowers the macro-F1 score slightly, to $0.631$ and $0.620$. The sampler of \citet{bao2025generating}, which starts from training records noised to step $T/2$, scores like the sampler that starts from Gaussian noise, with a TSTR difference of $0.000\pm0.006$ ($p=0.85$) for the full U-Net. Every DDPM run ended through early stopping, after $50{,}500$ to $93{,}500$ steps for the full U-Net and $100{,}000$ to $159{,}500$ steps for the small U-Net, well before the limit of $200{,}000$ steps, so that these scores describe converged models.

The ceilings themselves depend on how much family information the latent vectors retain. To place them, we trained the classifiers of Section~\ref{sec:measures} on opcode frequencies of the same programs, with the same split and the $1{,}217$ opcodes that occur in at least ten training programs. The relative opcode frequencies of the first $7{,}000$ opcodes, the part of a program that the CAE reads, reach a macro-F1 score of $0.682$, those of the whole program $0.696$, and these frequencies together with those of the $1{,}000$ most frequent opcode bigrams $0.749$, against $0.635$ for the $32$ latent features. The CAE therefore retains most of the family information of the opcode frequencies that it reads, and the remaining gap to richer representations of the programs bounds every generator of these latent vectors from above, the classical DDPM included.

\subsection{Distribution fidelity, memorization, and membership}\label{sec:res_privacy}

\begin{table*}[t]
\small
\caption{Distribution fidelity and memorization of the synthetic records of all $17$ families, mean and standard deviation over the ten seeds. MMD$^2$ and $r_{\mathrm{mem}}$ are averaged over the families; $q_{10}$ is the fraction of synthetic records that lie closer to the training data than $90\%$ of the real test records (Eq.~\eqref{eq:mem}), and $q_{10}^{\mathrm{te}}$ is the same fraction computed with the distances to the nearest real test record. Random states: the initial states of sampling, decoded without any denoising step.}\label{tab:privacy}
\begin{tabular*}{\tblwidth}{@{}LLRRRR@{}}
\toprule
Model & Variant & MMD$^2$ & $r_{\mathrm{mem}}$ & $q_{10}$ & $q_{10}^{\mathrm{te}}$\\
\midrule
Random states & Min--max & $2.237\pm0.004$ & $46.7\pm0.1$ & $0.0\%$ & $0.0\%$\\
 & Sign-split & $1.214\pm0.002$ & $4.90\pm0.02$ & $0.0\%$ & $0.0\%$\\
\midrule
Quantum denoiser, min--max & Post-selection & $1.787\pm0.075$ & $34.3\pm1.8$ & $0.0\%$ & $0.0\%$\\
 & No selection & $2.237\pm0.004$ & $46.7\pm0.1$ & $0.0\%$ & $0.0\%$\\
 & Classical control & $0.904\pm0.030$ & $3.11\pm0.06$ & $0.0\%$ & $0.0\%$\\
Quantum denoiser, sign-split & Post-selection & $1.171\pm0.010$ & $4.77\pm0.03$ & $0.0\%$ & $0.0\%$\\
 & No selection & $1.214\pm0.002$ & $4.91\pm0.02$ & $0.0\%$ & $0.0\%$\\
 & Classical control & $0.250\pm0.020$ & $2.67\pm0.06$ & $1.4\%$ & $1.1\%$\\
\midrule
Classical latent DDPM & Full U-Net & $0.032\pm0.003$ & $2.25\pm0.06$ & $0.05\%$ & $0.01\%$\\
 & Small U-Net & $0.089\pm0.006$ & $3.33\pm0.06$ & $0.0\%$ & $0.0\%$\\
 & Full U-Net, start at $T/2$ & $0.031\pm0.003$ & $2.25\pm0.06$ & $0.05\%$ & $0.01\%$\\
\bottomrule
\end{tabular*}
\end{table*}

Table~\ref{tab:privacy} reports the fidelity and memorization measures of Section~\ref{sec:measures}. The post-selected circuit moves its records toward the real distribution: with the min--max codec, it lowers the squared MMD from $2.237$ for the random states to $1.787$, that is, by $20\%$ and in every seed ($p=0.002$), and the memorization ratio from $46.7$ to $34.3$, and with the sign-split codec, it lowers the squared MMD from $1.214$ to $1.171$, again in every seed. A synthetic record of the circuit nevertheless lies $4.8$ to $34$ times farther from its nearest training record than an unseen real record does, and no synthetic record falls below the $10$th percentile of the reference distances ($q_{10}=0$). These measures exclude the copying of training records. However, the random states, which contain no learned information, also reach $q_{10}=0$, so that the distance of the records from the training data reflects their distance to the real distribution as a whole, and these measures alone do not establish a favorable trade-off between privacy and utility. The classical control lies closer to the real distribution (squared MMD of $0.90$ and $0.25$), at memorization ratios of $3.1$ and $2.7$ and with $q_{10}\leq1.4\%$ against $q_{10}^{\mathrm{te}}\leq1.1\%$. For the family Allaple.A, the classical control with the sign-split codec reaches $q_{10}=10.5\%$ on average over the seeds, but $q_{10}^{\mathrm{te}}=14.1\%$, so that these records lie close to the real records of the family in general rather than to individual training records.

The classical DDPM generates records that lie much closer to the real distribution, with a squared MMD of $0.032\pm0.003$ for the full U-Net and $0.089\pm0.006$ for the small U-Net, and these records still lie $2.3$ and $3.3$ times farther from their nearest training record than unseen real records do, with $q_{10}$ of $0.05\%$ and $0\%$. The sampler that starts from noised training records gives the same values. At step $T/2$, the linear schedule of the DDPM keeps a fraction $\sqrt{\bar\alpha_{T/2}}=0.28$ of the scaled record against Gaussian noise with a standard deviation of $0.96$, so that the reverse process retains little of the individual record from which it starts. At the scale of our data, the distance-based measures therefore detect no copying even for a generator whose records approach the real distribution.

\begin{table*}[t]
\small
\caption{Membership tests (AUC; $0.5$ corresponds to a random guess), mean over the ten seeds. Loss: the loss-based test of Eq.~\eqref{eq:lossmia} against the trained generator, with members drawn from $\mathcal{D}_{\mathrm{tr}}$ and the test records as non-members; its $95\%$ bootstrap interval spans about $\pm0.013$ in every row. Length: the tests of Eqs.~\eqref{eq:lenmia} (nearest) and \eqref{eq:lenlr} (likelihood) on the released lengths, averaged over the $17$ families, with the lengths drawn without the perturbation ($\tau_k=0$) or with the bandwidth of Eq.~\eqref{eq:bandwidth} (smoothed). The DDPM has no decoder that draws lengths, so that its length test uses the lengths of its records under the min--max codec.}\label{tab:mia}
\begin{tabular*}{\tblwidth}{@{}LLRRRR@{}}
\toprule
 & & & \multicolumn{3}{c}{Length}\\
Model & Variant & Loss & $\tau_k=0$ & Smoothed, nearest & Smoothed, likelihood\\
\midrule
Quantum denoiser, min--max & Post-selection & $0.492$ & $0.80$ & $0.51$ & $0.52$\\
 & No selection & $0.501$ & $0.80$ & $0.51$ & $0.52$\\
 & Classical control & $0.477$ & $0.80$ & $0.51$ & $0.52$\\
Quantum denoiser, sign-split & Post-selection & $0.477$ & $0.48$ & $0.48$ & $0.48$\\
 & No selection & $0.492$ & $0.48$ & $0.48$ & $0.48$\\
 & Classical control & $0.478$ & $0.50$ & $0.50$ & $0.51$\\
\midrule
Classical latent DDPM & Full U-Net & $0.484$ & \multicolumn{3}{c}{$0.51$ (lengths of the generated records)}\\
 & Small U-Net & $0.478$ & \multicolumn{3}{c}{$0.51$ (lengths of the generated records)}\\
\bottomrule
\end{tabular*}
\end{table*}

Table~\ref{tab:mia} reports the membership tests. The loss-based test, which requires the trained generator, reaches AUCs between $0.48$ and $0.50$ for every generator, so that no generator gives this adversary an advantage, whether it learns little, as the post-selected circuit, or approaches the real distribution, as the full DDPM. The values below $0.5$, whose intervals exclude $0.5$ for several generators, indicate that the training records have a slightly higher loss than the test records under every generator, which suggests a property of the latent vectors rather than of a generator; an adversary that reverses the score would reach at most $0.52$.

The smoothed lengths remove the leak that Proposition~\ref{prop:length} identifies. Since the decoded lengths of the min--max codec do not depend on the generated states, the three min--max denoisers and the random states release lengths with the same distribution. With $\tau_k=0$, these lengths repeat the length of $59.5\%$ of the training records of a family, on average over the $17$ families, against $1.8\%$ of its held-out records, and the nearest adversary reaches an AUC of $0.80$, with values up to $1.00$ for single families and seeds. With the bandwidth of Eq.~\eqref{eq:bandwidth}, the two rates become $2.3\%$ and $2.3\%$, and the AUC drops to $0.51$ for the nearest adversary and to $0.52$ for the adversary that knows $\tau_k$. The most exposed family is Startpage, the smallest one, with an AUC of $0.61$ and a $95\%$ interval of $[0.52,0.71]$ for the second adversary. This change requires no retraining and leaves the round-trip ceilings nearly unchanged: with the min--max codec, TSTR equals $0.605\pm0.002$ with smoothing and $0.602\pm0.002$ without it, and TRTS equals $0.537\pm0.003$ and $0.543\pm0.003$. The lengths of the sign-split records reveal little about membership even with $\tau_k=0$ (AUC of $0.48$ to $0.50$), since their decoded length is at most $\rho$ (Proposition~\ref{prop:length}). The records of the classical DDPM, which generates every length together with the rest of the latent vector, reach an AUC of $0.51$ under the min--max codec, so that the smoothed bootstrap brings the quantum denoiser to the level of a generator without this channel.

\begin{table*}[t]
\small
\caption{Lengths drawn from the $\varepsilon$-differentially private histogram of Eq.~\eqref{eq:dplengths} (min--max codec, $B=64$ bins), against the two rules of Eq.~\eqref{eq:decode}. AUC: length membership tests of Eqs.~\eqref{eq:lenmia} and \eqref{eq:lenlr}, averaged over the $17$ families ($512$ released lengths per family; for the histogram, one draw with the seed $42$). Aug, TSTR, and TRTS: round-trip ceiling with the lengths of each rule, averaged over three draws for the histogram and eight draws for the other rules.}\label{tab:dp}
\begin{tabular*}{\tblwidth}{@{}LRRRRR@{}}
\toprule
Length rule & AUC, nearest & AUC, likelihood & Aug & TSTR & TRTS\\
\midrule
Training lengths ($\tau_k=0$) & $0.80$ & $0.80$ & $0.635\pm0.002$ & $0.602\pm0.002$ & $0.543\pm0.003$\\
Smoothed bootstrap & $0.51$ & $0.52$ & $0.634\pm0.002$ & $0.605\pm0.002$ & $0.537\pm0.003$\\
\midrule
Histogram, $\varepsilon=10$ & $0.52$ & $0.51$ & $0.633\pm0.003$ & $0.604\pm0.003$ & $0.517\pm0.001$\\
Histogram, $\varepsilon=3$ & $0.52$ & $0.51$ & $0.632\pm0.002$ & $0.601\pm0.004$ & $0.514\pm0.001$\\
Histogram, $\varepsilon=1$ & $0.51$ & $0.52$ & $0.630\pm0.002$ & $0.597\pm0.004$ & $0.498\pm0.003$\\
Histogram, $\varepsilon=0.3$ & $0.51$ & $0.51$ & $0.626\pm0.003$ & $0.576\pm0.004$ & $0.460\pm0.002$\\
Histogram, $\varepsilon=0.1$ & $0.51$ & $0.51$ & $0.623\pm0.002$ & $0.550\pm0.001$ & $0.386\pm0.003$\\
\bottomrule
\end{tabular*}
\end{table*}

Table~\ref{tab:dp} reports the lengths drawn from the differentially private histogram of Eq.~\eqref{eq:dplengths}. For every privacy budget from $\varepsilon=10$ to $\varepsilon=0.1$, the mean AUC of both adversaries lies between $0.51$ and $0.52$, the level of the smoothed bootstrap. The largest values of single families occur for Startpage ($0.58$ to $0.64$), whose $51$ held-out records give its AUC a standard error of about $0.05$. The formal guarantee therefore keeps the empirical protection of the smoothed bootstrap, and its cost in utility depends on the budget. With $\varepsilon\geq1$, the round-trip ceiling stays within $0.008$ of the smoothed bootstrap in TSTR and within $0.039$ in TRTS, whereas $\varepsilon=0.1$ lowers it to $0.550$ and $0.386$. The noise of the counts probably places released lengths in bins that few or no training records of the family occupy, which weakens the class information that the lengths carry. Since the histogram replaces only the length rule of the decoder, the records of every denoiser inherit its guarantee without retraining.

\subsection{Class information in the generated states}\label{sec:res_states}

To locate the source of the limited utility, we compare the generated magnitudes with the mean encoded state of each class. For a vector $y\in\mathbb{R}^{D}_{\geq0}$ of unit norm generated for class $k$, we compute
\begin{equation}
\begin{aligned}
F_{\mathrm{own}}(y)&=\big(m_k^{\top}y\big)^{2},\qquad \hat k(y)=\arg\max_{c}\,\big(m_c^{\top}y\big)^{2},\\
m_c&=\frac{\sum_{j:\,k_j=c}a(z_j)}{\big\lVert\sum_{j:\,k_j=c}a(z_j)\big\rVert_2}.
\end{aligned}\label{eq:classfid}
\end{equation}
In these quantities, $m_c$ represents the normalized mean of the encoded training states of class $c$, $F_{\mathrm{own}}$ measures how close $y$ lies to the mean state of its own class, and $\hat k(y)$ denotes the class of the closest mean state. We report $F_{\mathrm{own}}$ and the nearest-mean accuracy, which is the fraction of vectors with $\hat k(y)=k$ over the $K=17$ classes.

\begin{table}[t]
\small
\caption{Class information in the generated magnitudes, pooled over the ten seeds and the $17$ classes ($512$ records per class and seed), measured with Eq.~\eqref{eq:classfid} against the mean encoded states of the training records. The magnitudes of random states are those of the initial states $\ket{z_T}$ before denoising. The chance level of the nearest-mean accuracy equals $1/17=0.059$.}\label{tab:states}
\begin{tabular*}{\tblwidth}{@{}LRRRR@{}}
\toprule
 & \multicolumn{2}{c}{Min--max} & \multicolumn{2}{c}{Sign-split}\\
Records & $F_{\mathrm{own}}$ & Accuracy & $F_{\mathrm{own}}$ & Accuracy\\
\midrule
Real test records & $0.982$ & $0.406$ & $0.675$ & $0.477$\\
Magnitudes of random states & $0.681$ & $0.059$ & $0.455$ & $0.061$\\
Post-selection & $0.816$ & $0.065$ & $0.475$ & $0.059$\\
No selection & $0.681$ & $0.058$ & $0.454$ & $0.059$\\
Classical control & $0.996$ & $0.173$ & $0.805$ & $0.392$\\
\bottomrule
\end{tabular*}
\end{table}

Table~\ref{tab:states} indicates that the post-selected circuit moves the random initial states toward the data without giving them the identity of their class. With the min--max codec, the magnitudes of random states have $F_{\mathrm{own}}=0.681$, the circuit raises this value to $0.816$, and real test records reach $0.982$. However, the nearest-mean accuracy of the generated states remains at $0.065$, close to the chance level of $0.059$, whereas real test records reach $0.406$. With the sign-split codec, the circuit leaves the states almost where they start ($0.475$ against $0.455$), and their accuracy equals the chance level. The circuit without branch selection leaves the random states unchanged with both codecs, in agreement with its training. By contrast, the classical control produces states closer to the mean states than real records ($0.996$ and $0.805$) and recovers part of the class identity ($0.173$ and $0.392$).

The geometry of the encoded states offers one explanation. With the min--max codec, two training records of different classes have a mean fidelity of $0.944$, even slightly higher than two records of the same class ($0.938$), so that the infidelity loss rewards a class-specific output only through small differences between nearly parallel states. With the sign-split codec, the classes separate, with a mean fidelity of $0.239$ between classes against $0.449$ within a class, but the post-selected circuit removes only $22\%$ of the infidelity of the identity map, which leaves the generated states close to random states.

Two further diagnostics exclude other explanations of the limited utility. First, sampling starts from Haar-random states, whereas training presents the circuit at step $T$ with records noised by the forward process. For $4{,}000$ training records of each codec, the magnitudes of the noised states at $t=T$ and those of as many Haar-random states differ by a squared MMD of $0.0013$ and $0.0012$, close to the value of $0.0012$ between two independent sets of Haar-random states of the same size, and the nearest-mean accuracy of the noised states ($0.167$ and $0.083$) lies within $0.006$ of that of the Haar-random states, which we evaluate with the labels of the same records. Since these records are drawn at random from the training set, their labels follow its imbalanced class frequencies, and the accuracy of states without class information can therefore differ from the chance level of $1/17$ that applies to the balanced records of Table~\ref{tab:states}. The initial states of sampling therefore match the distribution on which the circuit is trained. Second, the gradients of the post-selected loss could vanish with the number of qubits \citep{mcclean2018barren}. Over $20$ random initializations with our standard deviation of $0.1$, the mean variance of the partial derivatives of the shared blocks $U_2^{(t)}$ and $U_3^{(t)}$ equals $6.8\times10^{-5}$ and $8.2\times10^{-5}$ with five data qubits and $6.0\times10^{-5}$ and $7.0\times10^{-5}$ with six, so that the additional qubit of the sign-split codec lowers it by $12\%$ to $14\%$, and by $20\%$ for the class-specific blocks. With angles drawn with a standard deviation of $\pi$, the variance of the shared blocks drops by two orders of magnitude, and the additional qubit lowers the variance of every block by a factor of $3.6$ to $3.8$, as expected for random circuits. The narrow initialization of Proposition~\ref{prop:init} therefore keeps the circuit away from this regime, and vanishing gradients do not explain why the circuit learns less with the sign-split codec.

\subsection{Samplers, weighting of the released records, and acceptance penalty}\label{sec:res_sampling}

\begin{table*}[t]
\small
\caption{Samplers of Eq.~\eqref{eq:samplers} and weighting of the released records, mean and standard deviation over the ten seeds (all $17$ families). Deterministic: the chain of Algorithm~\ref{alg:sample}. Weighted: the records of the deterministic chain drawn with replacement with probabilities proportional to $P_{\mathrm{acc}}$, which approximate the release of a quantum device. $P_{\mathrm{acc}}$ and the effective sample size (ESS, out of $512$ initial states per family) are averaged over the families and the seeds.}\label{tab:sampling}
\footnotesize
\begin{tabular*}{\tblwidth}{@{}LLLRRRRRR@{}}
\toprule
Codec & Denoiser & Sampler & $P_{\mathrm{acc}}$ & ESS & Aug & TSTR & TRTS & MMD$^2$\\
\midrule
Min--max & Post-selection & Deterministic & $3.0\times10^{-7}$ & $361$ & $0.593\pm0.005$ & $0.114\pm0.015$ & $0.028\pm0.002$ & $1.79$\\
 & & Ancestral & $1.9\times10^{-7}$ & $405$ & $0.588\pm0.004$ & $0.098\pm0.014$ & $0.025\pm0.001$ & $1.96$\\
 & & DDIM & $1.7\times10^{-7}$ & $396$ & $0.588\pm0.004$ & $0.093\pm0.014$ & $0.023\pm0.002$ & $2.00$\\
 & & Weighted & -- & -- & $0.595\pm0.006$ & $0.112\pm0.009$ & $0.032\pm0.003$ & $1.66$\\
 & Classical control & Deterministic & -- & -- & $0.624\pm0.005$ & $0.218\pm0.025$ & $0.167\pm0.016$ & $0.90$\\
 & & Ancestral & -- & -- & $0.621\pm0.003$ & $0.234\pm0.020$ & $0.186\pm0.020$ & $0.56$\\
 & & DDIM & -- & -- & $0.621\pm0.003$ & $0.243\pm0.010$ & $0.186\pm0.022$ & $0.52$\\
\midrule
Sign-split & Post-selection & Deterministic & $9.8\times10^{-5}$ & $487$ & $0.625\pm0.002$ & $0.036\pm0.006$ & $0.041\pm0.001$ & $1.17$\\
 & & Ancestral & $8.5\times10^{-5}$ & $494$ & $0.627\pm0.002$ & $0.036\pm0.007$ & $0.040\pm0.001$ & $1.18$\\
 & & DDIM & $8.3\times10^{-5}$ & $492$ & $0.625\pm0.002$ & $0.036\pm0.005$ & $0.040\pm0.001$ & $1.18$\\
 & & Weighted & -- & -- & $0.620\pm0.003$ & $0.039\pm0.006$ & $0.041\pm0.002$ & $1.17$\\
 & Classical control & Deterministic & -- & -- & $0.617\pm0.004$ & $0.367\pm0.008$ & $0.317\pm0.017$ & $0.25$\\
 & & Ancestral & -- & -- & $0.619\pm0.003$ & $0.384\pm0.007$ & $0.336\pm0.012$ & $0.20$\\
 & & DDIM & -- & -- & $0.616\pm0.003$ & $0.382\pm0.008$ & $0.318\pm0.008$ & $0.18$\\
\bottomrule
\end{tabular*}
\end{table*}

Table~\ref{tab:sampling} tests whether the deterministic chain of Algorithm~\ref{alg:sample} limits the utility of the records. With the classical control, both samplers of Eq.~\eqref{eq:samplers} raise the TSTR score in every seed ($p=0.002$), by $0.016\pm0.010$ with the ancestral sampler and $0.026\pm0.019$ with the DDIM sampler for the min--max codec and by $0.017\pm0.007$ and $0.014\pm0.006$ for the sign-split codec, and they lower the squared MMD from $0.90$ to $0.56$ and $0.52$ and from $0.25$ to $0.20$ and $0.18$. The deterministic chain therefore costs the control part of its utility, but with the best sampler the control still reaches only $0.243$ and $0.384$ in TSTR, far below the ceilings of $0.605$ and $0.593$. For the post-selected circuit, both samplers lower the TSTR score with the min--max codec, by $0.017\pm0.019$ and $0.021\pm0.019$ ($p=0.010$ for each), and raise the squared MMD from $1.79$ to $1.96$ and $2.00$, whereas they leave the scores of the sign-split codec unchanged, with TSTR differences of at most $0.001$. The noise that both samplers add between the steps likely moves the states back toward random states, which the circuit cannot denoise again in the remaining steps. The sampler therefore does not explain the gap between the circuit and its classical control.

Post-selection on hardware releases a record only for the accepted runs, which weights the initial states with their acceptance probability (Section~\ref{sec:sampling}). Over the $512$ initial states of a family, the logarithm $\log_{10}P_{\mathrm{acc}}$ of the min--max circuit has a standard deviation of $0.26$, so that the effective sample size of the weights equals $361$ of $512$ on average, with a minimum of $326$, and the ten records of highest weight carry $6.7\%$ of the total weight, against $2.0\%$ for equal weights. With the sign-split codec, the effective sample size reaches $487$. The weighted records score like the unweighted ones, with TSTR differences of $-0.002\pm0.015$ ($p=0.63$) and $0.003\pm0.006$ ($p=0.28$), and the weighting lowers the squared MMD of the min--max records slightly, from $1.79$ to $1.66$. For the trained circuits, post-selection therefore changes the distribution of the released records little, so that the scores of the simulation, which returns every initial state once, also describe the records of a quantum device.

\begin{table*}[t]
\small
\caption{Post-selected circuits trained with the penalized objective of Eq.~\eqref{eq:pacc_loss}, mean and standard deviation over the seeds $42$, $7$, and $13$ (all $17$ families); $\lambda=0$ denotes the runs of Table~\ref{tab:training} with the same seeds. Runs per record: the expected number of runs per accepted record, $1/\mathbb{E}[P_{\mathrm{acc}}]$, with $P_{\mathrm{acc}}$ averaged over the initial states, the families, and the seeds.}\label{tab:penalty}
\footnotesize
\begin{tabular*}{\tblwidth}{@{}LLRRRRRRR@{}}
\toprule
Codec & $\lambda$ & Test infidelity & Reduction & $P_{\mathrm{acc}}$ & Runs per record & ESS & TSTR & TRTS\\
\midrule
Min--max & $0$ & $0.0341\pm0.0073$ & $59.4\%$ & $1.6\times10^{-7}$ & $6.3\times10^{6}$ & $366$ & $0.110\pm0.012$ & $0.028\pm0.003$\\
 & $0.001$ & $0.0366\pm0.0059$ & $56.5\%$ & $2.9\times10^{-4}$ & $3.5\times10^{3}$ & $405$ & $0.108\pm0.010$ & $0.024\pm0.001$\\
 & $0.01$ & $0.0370\pm0.0018$ & $56.0\%$ & $1.7\times10^{-2}$ & $58$ & $435$ & $0.091\pm0.012$ & $0.024\pm0.001$\\
 & $0.1$ & $0.0461\pm0.0020$ & $45.3\%$ & $0.14$ & $7.4$ & $487$ & $0.077\pm0.004$ & $0.020\pm0.001$\\
\midrule
Sign-split & $0$ & $0.0723\pm0.0008$ & $22.2\%$ & $1.3\times10^{-4}$ & $7.8\times10^{3}$ & $487$ & $0.039\pm0.010$ & $0.041\pm0.002$\\
 & $0.001$ & $0.0722\pm0.0006$ & $22.3\%$ & $1.8\times10^{-3}$ & $5.5\times10^{2}$ & $490$ & $0.041\pm0.007$ & $0.040\pm0.003$\\
 & $0.01$ & $0.0738\pm0.0008$ & $20.5\%$ & $4.4\times10^{-2}$ & $23$ & $497$ & $0.037\pm0.008$ & $0.040\pm0.002$\\
 & $0.1$ & $0.0935\pm0.0000$ & $-0.7\%$ & $1.00$ & $1.0$ & $512$ & $0.047\pm0.004$ & $0.038\pm0.001$\\
\bottomrule
\end{tabular*}
\end{table*}

Table~\ref{tab:penalty} reports the penalized objective of Eq.~\eqref{eq:pacc_loss}. With the min--max codec, the weights $\lambda=0.001$, $0.01$, and $0.1$ raise the mean acceptance probability from $1.6\times10^{-7}$ to $2.9\times10^{-4}$, $1.7\times10^{-2}$, and $0.14$, so that the expected number of runs per accepted record falls from $6.3\times10^{6}$ to $3.5\times10^{3}$, $58$, and $7.4$. The same weights lower the reduction of the test infidelity from $59.4\%$ to $56.5\%$, $56.0\%$, and $45.3\%$, and the TSTR score from $0.110$ to $0.108$, $0.091$, and $0.077$, the last of which equals the score of the random states ($0.076$). With the sign-split codec, $\lambda=0.001$ and $0.01$ raise the acceptance probability from $1.3\times10^{-4}$ to $1.8\times10^{-3}$ and $4.4\times10^{-2}$ at reductions of $22.3\%$ and $20.5\%$, and $\lambda=0.1$ returns the circuit to the identity map: its best epoch is the first evaluation, the ancilla returns the outcome $0$ with probability $1.000$, and its test infidelity equals that of the circuit without branch selection. In other words, the denoising of the trained circuit relies on the ancilla outcomes of low probability that make post-selection expensive, and the penalty exchanges denoising for acceptance without raising the utility of the records above that of the random states.

\subsection{Sampling under hardware noise}\label{sec:res_noise}

\begin{table*}[t]
\small
\caption{Density-matrix sampling of the post-selected circuits for all $17$ families, with depolarizing noise of probability $p$ after every layer and a read-out with exact probabilities ($S=\infty$) or with $S=1{,}000$ shots. $P_{\mathrm{acc}}$: mean acceptance probability of a sample. The scores are means over the ten seeds; their standard deviation is at most $0.015$ for TSTR, $0.005$ for Aug, and $0.003$ for TRTS. At $p=0$, the values differ slightly from those of Tables~\ref{tab:training} and \ref{tab:utility}, because the initial states and the decoded lengths are drawn separately.}\label{tab:noise}
\begin{tabular*}{\tblwidth}{@{}LRRRRRRRR@{}}
\toprule
Codec & $p$ & $S$ & $P_{\mathrm{acc}}$ & Aug & TSTR & TRTS & MMD$^2$ & $r_{\mathrm{mem}}$\\
\midrule
Min--max & $0$ & $\infty$ & $3.0\times10^{-7}$ & $0.592$ & $0.110$ & $0.028$ & $1.78$ & $34.2$\\
 & $0$ & $1{,}000$ & $3.0\times10^{-7}$ & $0.589$ & $0.113$ & $0.028$ & $1.79$ & $35.3$\\
 & $0.001$ & $\infty$ & $6.5\times10^{-7}$ & $0.624$ & $0.055$ & $0.021$ & $3.23$ & $22.4$\\
 & $0.001$ & $1{,}000$ & $6.5\times10^{-7}$ & $0.619$ & $0.063$ & $0.020$ & $2.88$ & $23.5$\\
 & $0.01$ & $\infty$ & $2.6\times10^{-5}$ & $0.625$ & $0.021$ & $0.019$ & $5.42$ & $26.2$\\
 & $0.01$ & $1{,}000$ & $2.6\times10^{-5}$ & $0.626$ & $0.029$ & $0.015$ & $3.98$ & $27.2$\\
\midrule
Sign-split & $0$ & $\infty$ & $9.9\times10^{-5}$ & $0.625$ & $0.038$ & $0.040$ & $1.17$ & $4.76$\\
 & $0$ & $1{,}000$ & $9.9\times10^{-5}$ & $0.624$ & $0.041$ & $0.041$ & $1.12$ & $4.94$\\
 & $0.001$ & $\infty$ & $1.7\times10^{-4}$ & $0.630$ & $0.018$ & $0.027$ & $2.14$ & $3.50$\\
 & $0.001$ & $1{,}000$ & $1.7\times10^{-4}$ & $0.629$ & $0.023$ & $0.031$ & $1.91$ & $3.64$\\
 & $0.01$ & $\infty$ & $1.3\times10^{-3}$ & $0.635$ & $0.017$ & $0.007$ & $3.04$ & $3.32$\\
 & $0.01$ & $1{,}000$ & $1.3\times10^{-3}$ & $0.632$ & $0.019$ & $0.023$ & $2.39$ & $3.47$\\
\bottomrule
\end{tabular*}
\end{table*}

Table~\ref{tab:noise} reports the density-matrix sampling of Section~\ref{sec:noise} for the post-selected circuits. At $p=0$ and with the exact read-out, this sampler applies the same circuit as the state-vector sampler of Algorithm~\ref{alg:sample}, but it draws its initial states and decoded lengths separately, so that its acceptance probability and scores differ slightly from those of Tables~\ref{tab:training} and \ref{tab:utility}. Depolarizing noise increases the acceptance probability, at $p=0.01$ from $3.0\times10^{-7}$ to $2.6\times10^{-5}$ with the min--max codec and from $9.9\times10^{-5}$ to $1.3\times10^{-3}$ with the sign-split codec, since the noise drives the ancilla toward the maximally mixed state, in which the outcome $0$ has probability $1/2$. The same noise degrades the generated records: at $p=0.01$ and with the exact read-out, TSTR decreases from $0.110$ and $0.038$ to $0.021$ and $0.017$, below the scores of the random states ($0.076$ and $0.039$), TRTS decreases to $0.019$ and $0.007$, the latter close to the value of $0.0065$ of a classifier that assigns every synthetic record to the same family, and the squared MMD increases from $1.78$ to $5.42$ and from $1.17$ to $3.04$. The noise also drives the data register toward the maximally mixed state, whose magnitudes are uniform, so that the records of different families become similar to each other and dissimilar to the real records. This similarity may explain the rise of the Aug score toward the real-only score of $0.635$, since classifiers can separate such records from the real ones without changing their decisions on real records. The memorization ratio decreases under noise, while $q_{10}$ remains zero in every setting. A read-out of $1{,}000$ shots changes the scores at $p=0$ by less than their spread between seeds, in agreement with Lemma~\ref{lem:shots}, which bounds the mean squared error of the magnitudes by $0.031$. Even at $p=0.01$, one accepted record requires about $3.9\times10^{4}$ runs with the min--max codec and $7.9\times10^{2}$ runs with the sign-split codec on average, so that post-selected sampling of the trained circuits would still require a large number of circuit executions on current hardware.

\section{Discussion}\label{sec:discussion}

In this section, we interpret the results of Section~\ref{sec:results}, discuss what they imply for the release of synthetic security data and for the execution of the denoiser on quantum hardware, and state the limitations of our study.

\subsection{From the training objective to class information}\label{sec:disc_learning}

The central observation of our experiments is that the post-selected circuit learns its training objective, while the class information of its records remains close to that of the random states from which sampling starts. The circuit lowers the test infidelity by $62.6\%$ and $22.4\%$ relative to the identity map (Table~\ref{tab:training}), the PennyLane runs reproduce this result to within $8.1\times10^{-8}$, and the circuit moves its records toward the real distribution, with a squared MMD $20\%$ below that of the random states with the min--max codec (Table~\ref{tab:privacy}). Nevertheless, its records cover $7\%$ of the TSTR range between the random states and the round-trip ceiling with the min--max codec and none of it with the sign-split codec, whereas a classical latent DDPM of comparable size reaches a TSTR score of $0.442$ on the same latent vectors (Table~\ref{tab:utility}). These observations are compatible, since the infidelity measures the distance between whole states, while the classification depends on the small part of each state that differs between families.

We attribute this gap to three properties of the model. First, the min--max encoding hides the family in small differences between nearly parallel states: two records of different families have a slightly higher mean fidelity than two records of the same family ($0.944$ against $0.938$). A denoiser can therefore reach a low infidelity by moving every state toward the common region of the data, which the post-selected circuit does ($F_{\mathrm{own}}$ rises from $0.681$ to $0.816$) without giving the states a family identity (nearest-mean accuracy of $0.065$ against a chance level of $0.059$). The sign-split codec separates the families, and the classical control exploits this separation (TSTR of $0.367$), but on this codec the circuit removes only $22\%$ of the infidelity of the identity map and leaves the states close to random ones. Neither a mismatch between the initial states and the training distribution nor vanishing gradients explains this behavior (Section~\ref{sec:res_states}).

Second, the formulation of sampling constrains every denoiser of this framework, including the classical control. The control reaches an infidelity $14$ and $3.3$ times lower than the circuit, yet its TSTR scores of $0.218$ and $0.367$ remain below the ceilings of $0.605$ and $0.593$. Its states lie closer to the mean states of their families than real records do ($F_{\mathrm{own}}$ of $0.996$ against $0.982$ with the min--max codec and $0.805$ against $0.675$ with the sign-split codec), which is the behavior that a deterministic chain of posterior means predicts, since the fidelity-optimal output of every step is a mean direction over the records compatible with its input. The samplers of Table~\ref{tab:sampling} confirm part of this explanation: the noise that they add between the steps raises the TSTR score of the control in every seed and lowers its squared MMD, but with the best sampler the control still remains at $0.243$ and $0.384$, and neither sampler raises the scores of the circuit. The remaining gap between the control and the small U-Net, $0.075$ in TSTR with the sign-split codec, therefore stems from the other differences between them, such as the normalization of the states, the ten steps instead of $1{,}000$, and the infidelity loss, which our experiments do not separate.

Third, the non-unitary step of the circuit relies on post-selection. Without branch selection, the circuit remains at the level of the identity map, which agrees with the argument of \citet{zhang2024generative} that a denoising map must be contractive and that a unitary circuit cannot provide this contraction. With post-selection, the circuit learns, but it concentrates its denoising in the last steps, through ancilla outcomes whose probability falls to $0.006$ (Section~\ref{sec:res_training}). The penalized objective indicates that this concentration is not an accident of the optimization (Table~\ref{tab:penalty}): when training rewards acceptance, the acceptance probability rises by up to six orders of magnitude, but the reduction of the infidelity and the TSTR score fall toward those of the random states, and the strongest penalty returns the sign-split circuit to the identity map. In this design, the denoising and the acceptance of the circuit therefore compete for the same ancilla outcomes.

\subsection{Privacy of the released records}\label{sec:disc_privacy}

The privacy measures of the quantum denoiser appear favorable at first sight: no synthetic record lies closer to the training data than $90\%$ of the real test records ($q_{10}=0$), the loss-based membership test reaches an AUC of at most $0.50$, and the length test reaches $0.51$ once the decoder smooths the lengths (Tables~\ref{tab:privacy} and \ref{tab:mia}). However, the decoded random states, which contain no information learned from the data, reach the same distance values, so that the distance to the training records reflects the distance of the records to the real distribution as a whole. The classical control and the classical DDPM, whose records carry far more class information, also avoid copying, even with the sampler that starts from noised training records, and their loss-based AUCs equal about $0.48$. At the scale of our data, neither the distance measures nor the loss-based test separates the generators. The favorable privacy measures of the quantum denoiser should therefore be read together with its utility, as \citet{stadler2022synthetic} recommend for synthetic data in general, and the trade-off between privacy and utility of the fully quantum design can be assessed once its records approach the round-trip ceiling. Moreover, the loss-based test is a weak adversary, and attacks that calibrate the loss of every record with shadow models trained on similar data \citep{shokri2017membership} may detect membership that it misses.

The analysis of the decoder carries a lesson that extends beyond our model. Every amplitude encoding removes the length of the encoded vector, so that any generator that operates on encoded states requires a classical component that restores it. This component derives from the training data, whether it stores the training lengths as our decoder does or learns a separate model of the lengths, and in both cases it forms a channel that bypasses the quantum circuit. Storing the length in an additional amplitude would move this channel into the circuit, where it would face the same memorization risks as the rest of the record. With exact training lengths, our decoder reveals $59.5\%$ of the training lengths of a family and supports a membership test with an AUC of $0.80$ on average over the families, whatever the quality of the generated states (Proposition~\ref{prop:length}). The smoothed bootstrap of Eq.~\eqref{eq:decode} closes this channel without retraining and without a measurable cost in utility, and it lowers the AUC to $0.51$ for the nearest adversary and to $0.52$ for the adversary that knows the bandwidth, close to the $0.51$ of the classical DDPM, which has no such decoder. The smallest family, Startpage, remains the most exposed, with an AUC of $0.61$ for the second adversary, so that the empirical protection weakens where the training records of a family are few. The differentially private histogram of Eq.~\eqref{eq:dplengths} turns this protection into a formal guarantee: for every budget from $\varepsilon=10$ to $\varepsilon=0.1$, the two length tests remain at $0.51$ to $0.52$, and with $\varepsilon\geq1$, the round-trip ceiling stays within $0.008$ of the smoothed bootstrap in TSTR (Table~\ref{tab:dp}). This guarantee covers the lengths only. A release with decoder also exposes the scaling parameters of the codec and the CAE, which derive from $\mathcal{D}_{\mathrm{tr}}$ as well, so that the guarantee of the lengths holds conditionally on these parameters, and a complete release would require an encoder and scaling parameters estimated from public data or with a privacy guarantee of their own. A privacy analysis of a quantum generator must therefore cover the classical pre- and post-processing around the circuit.

\subsection{Feasibility on quantum hardware}\label{sec:disc_hardware}

The two expectations behind fully quantum denoisers (Section~\ref{sec:intro}) can now be examined directly. The first concerns the parameter budget: the denoiser employs $11{,}520$ angles and acts on six qubits with the min--max codec, yet the classical control, with $11{,}458$ parameters and the same formulation, reaches a test infidelity $14$ times lower and trains in about $2$ minutes against $41$ minutes for the circuit (Table~\ref{tab:training}). A classical latent DDPM with $11{,}404$ parameters reaches a TSTR score of $0.442$ against $0.114$ for the circuit (Table~\ref{tab:utility}). At this scale, classical models of the same size therefore make better use of the budget than the circuit. The second expectation concerns the execution on a quantum device, where the cost of post-selection dominates. After training, the mean acceptance probability equals $3.0\times10^{-7}$ with the min--max codec, so that one accepted record requires about $3.4\times10^{6}$ executions of the ten-step circuit when every attempt starts from a fresh initial state. The acceptance probability varies little between the initial states (an effective sample size of $361$ of $512$), so that post-selection on hardware would release nearly the same distribution of records as the simulation (Table~\ref{tab:sampling}). A penalty on low acceptance probabilities lowers the cost to $7.4$ runs per record, but only by removing most of the denoising (Table~\ref{tab:penalty}), and depolarizing noise raises the acceptance probability by mixing the ancilla while it degrades the records (Table~\ref{tab:noise}). The read-out with $1{,}000$ shots costs little in comparison, in agreement with Lemma~\ref{lem:shots}. Training on a quantum device would add the cost of estimating every term of the loss with swap tests (Lemma~\ref{lem:shots}), on top of the state preparation of Eq.~\eqref{eq:amp}, whose gate count grows exponentially with the number of qubits. At the scale of this study, our design therefore does not indicate a quantum advantage for the generation of malware records, and the competition between denoising and acceptance identifies post-selection as the first challenge for an execution on quantum hardware.

\subsection{Ethical and dual-use considerations}\label{sec:disc_ethics}

The generation of malware data raises the question of dual use. Our generators produce latent vectors of $32$ features, which contain no executable code, and the decoder of the CAE would map them onto sequences of opcode mnemonics without operands, addresses, or data, from which no executable program can be assembled directly. In this study, the executables were processed only statically, by disassembly, and were not executed. The released artifacts therefore do not extend the capabilities of an attacker beyond those of the public repositories from which the samples come, whereas the privacy analysis addresses the main risk of a release, the disclosure of information about the training records of the data holder.

\subsection{Limitations and future work}\label{sec:disc_limitations}

One of the primary limitations of this study is the scope of its evidence. We evaluate a single corpus of $17$ families from one repository, with ten seeds, so that the results describe this corpus and its latent representation. The differences that support our conclusions, such as the TSTR differences of $0.49$ and $0.56$ between the post-selected circuit and the round-trip ceilings, of $0.33$ between the circuit and its classical control with the sign-split codec, or of $0.33$ between the circuit and the small DDPM, exceed the spread between seeds many times, and the paired differences among them hold in every seed, whereas an evaluation on further corpora would test whether they generalize. The comparison is also limited in its generators: the only generator of another design is the classical latent DDPM. Other classical generators, such as tabular diffusion models \citep{kotelnikov2023tabddpm} or conditional GANs \citep{xu2019modeling}, and hybrid quantum diffusion models would widen the range against which the fully quantum design is placed, although the scores of the latter would not isolate the contribution of a circuit (Section~\ref{sec:setup_refs}).

A second limitation concerns the design space of the denoiser. We keep the ten steps, the depth of the blocks, and the family-specific first blocks fixed, and we do not separate the design choices in which our model differs from that of \citet{cacioppo2023quantum}, such as the posterior-mean target, the narrow initialization, and the conditioning on the family. In particular, $88.5\%$ of the angles lie in the family-specific blocks, whereas every parameter of the classical control serves all families, so that the two denoisers match in their number of parameters but not in their allocation. A third limitation concerns the execution of the circuits, which we simulate with exact state vectors and with a depolarizing channel of uniform strength. Real devices add readout errors, crosstalk, and drift, which our noise model does not capture.

A fourth limitation lies in the representation and in the privacy evaluation. We evaluate the generated records in the latent space of the CAE and do not decode them into opcode sequences, so that we cannot judge whether the decoded sequences resemble executable programs. The CAE reads only the first $7{,}000$ opcodes, which cover the whole program for $17\%$ of the records, and its latent vectors reach a macro-F1 score of $0.635$ against $0.749$ for opcode unigram and bigram frequencies of the whole program (Section~\ref{sec:res_utility}), so that they summarize a program coarsely. On the privacy side, our membership tests employ the generator and the released lengths, but not shadow models, and the differentially private lengths leave the CAE and the scaling parameters of the codec unprotected (Section~\ref{sec:disc_privacy}).

Based on the results, we propose several directions to improve the fully quantum denoiser. First, since denoising and acceptance compete for the same ancilla outcomes in our design, amplitude amplification of the accepted branch, or circuits whose non-unitary step does not rely on rare outcomes, may raise the acceptance probability without removing the denoising. Second, encodings that separate the families, together with a conditioning of every block on the family, may give the loss a stronger signal about the family, and a larger number of steps may reduce the contraction of the samples toward the family means. Third, an encoder trained on public programs and scaling parameters estimated with a privacy guarantee would extend the guarantee of the lengths to the complete release with decoder. Since the decoder runs after sampling, each of these changes keeps the length protection of Proposition~\ref{prop:length} intact.

\section{Conclusion}\label{sec:conclusion}

In this research, we presented a fully quantum latent diffusion model for malware records, in which a reverse-bottleneck circuit with one ancilla measurement denoised amplitude-encoded latent vectors of opcode sequences without any classical neural layer, and we examined whether such a model can generate useful and private security records at a sampling cost that quantum hardware can afford.

The circuit learned its objective: with post-selection, it lowered the infidelity of its denoising steps by $63\%$ relative to the identity map and moved its samples toward the data distribution, the theorem, the propositions, and the lemmas that we established on its target, initialization, training, read-out, and decoder held in every numerical test, and an independent PennyLane implementation reproduced every run to within $8.1\times10^{-8}$. Its records, however, carried little more class information than the random states from which sampling starts. Over $17$ families and ten seeds, classifiers trained on them reached a TSTR score of $0.114$, against $0.076$ for the random states and a round-trip ceiling of $0.605$, whereas a parameter-matched classical control reached $0.367$ with the sign-split codec and a classical latent DDPM with $11{,}404$ parameters reached $0.442$. Stochastic samplers did not raise the scores of the circuit, and on three seeds, a penalty on rejected runs raised the acceptance probability from $1.6\times10^{-7}$ to $0.14$ only by removing most of the denoising. On the privacy side, no synthetic record copied a training record, a loss-based membership test against every generator remained at the level of a random guess, and the smoothed bootstrap of the decoded lengths reduced the AUC of a membership test on the released lengths from $0.80$ to $0.51$ without retraining and without a measurable loss of utility. A differentially private histogram of the lengths gave the same empirical protection with a formal guarantee, at a cost of at most $0.008$ in the TSTR ceiling for $\varepsilon\geq1$.

These results suggest that a low training loss and a small parameter budget form a necessary first step toward a quantum generator of sensitive data, and that the utility and the privacy of its samples, measured together with the cost of sampling them, determine whether such a generator can serve as a source of shareable data. Moreover, the privacy analysis must include the classical encoding and decoding around the circuit, where our decoder initially exposed membership information that the circuit itself never processed.

Two limitations of our study deserve attention. First, our evidence comes from a single corpus of $17$ families, evaluated in the latent space of an encoder that reads only the first $7{,}000$ opcodes of a program and retains less family information than opcode frequencies. Second, we simulated the circuits with exact state vectors and a uniform depolarizing channel, so that the behavior on real devices with readout errors and drift remains untested.

We plan to improve the utility first, since it precedes any hardware study. Since denoising and acceptance compete for the same ancilla outcomes in the current design, we will examine non-unitary steps that do not rely on rare outcomes, such as amplitude amplification of the accepted branch, together with encodings that separate the families and a conditioning of every block on the family. We will also train the encoder and the scaling parameters of the codec on public programs, so that the guarantee of the lengths extends to a release with decoder, and assess the decoded opcode sequences of further malware corpora.

As quantum processors gain qubits and reliable mid-circuit measurements, a joint evaluation of utility and privacy that accounts for the cost of sampling can determine when a quantum generator becomes a credible source of shareable security data.

\section*{Declaration of generative AI use}
During the preparation of this work, the authors used Claude (Anthropic) to assist with language editing, code documentation, and the drafting of parts of the manuscript. After using this tool, the authors reviewed and edited the content as needed and take full responsibility for the content of the published article.

\section*{Author contributions}
\textbf{Ban Q. Tran}: Conceptualization, Methodology, Software, Validation, Formal analysis, Investigation, Data curation, Visualization, Writing -- original draft. \textbf{Quang Duy Tran}: Data curation, Validation, Writing -- review \& editing. \textbf{Anh H. Nguyen}: Software, Validation, Writing -- review \& editing. \textbf{Ziwen Pan}: Methodology, Writing -- review \& editing. \textbf{Tommy Dang}: Resources, Writing -- review \& editing. \textbf{Susan Mengel}: Conceptualization, Supervision, Resources, Writing -- review \& editing.

\section*{Declaration of competing interest}
The authors declare that they have no known competing financial interests or personal relationships that could have appeared to influence the work reported in this paper.

\section*{Acknowledgements}
We thank Texas Tech University and its REmotely-managed Power Aware Computing Systems and Services (REPACSS) for providing the high-performance computing servers on which we ran the experiments of this study. This research was partially funded by FPT University.

\section*{Data availability}
The malware samples are available from the public repository VirusShare. The hashes and family labels of the samples and the source code of this study are available from the authors upon request.

\appendix
\counterwithin{equation}{section}
\counterwithin{table}{section}

\section{Proofs and supplementary results}\label{app:proofs}

\subsection{Proof of Theorem~\ref{thm:target}}

\begin{proof}
\emph{Posterior.} With $\epsilon\sim\mathcal{CN}(0,2I_D)$, the real and the imaginary parts of the forward chain evolve as two independent real Gaussian chains with variances $\beta_t$, so the posterior of each part follows from the real case \citep{ho2020denoising}:
\begin{equation}
q(x_{t-1}\mid x_t,x_0)=\mathcal{CN}\big(\tilde\mu_t,\,2\tilde\beta_t I_D\big),\qquad \tilde\beta_t=\frac{1-\bar\alpha_{t-1}}{1-\bar\alpha_t}\,\beta_t.\label{eq:app_posterior}
\end{equation}
In this posterior, $\tilde\beta_t$ represents the variance of each real and each imaginary component. For $t=1$, $\tilde\beta_1=0$ and $x_0=\tilde\mu_1$, so Eq.~\eqref{eq:thm_target} holds with $\lambda_1=1$ and $\lambda_2=0$. For $t\geq2$, $\tilde\beta_t>0$, and we write $x_{t-1}=\tilde\mu_t+\sqrt{\tilde\beta_t}\,\xi$ with $\xi\sim\mathcal{CN}(0,2I_D)$. Since $\tilde\mu_t\neq0$ almost surely, we let $\hat\mu=\tilde\mu_t/\lVert\tilde\mu_t\rVert_2$, so that $\ket{\hat\mu}=\ket{\tilde\mu_t}$.

\emph{Second-moment operator.} Conditional on $(x_t,x_0)$, the left-hand side of Eq.~\eqref{eq:thm_target} equals $1-\operatorname{Tr}(M\varrho)$, where $M=\mathbb{E}\big[\ket{\hat x_{t-1}}\bra{\hat x_{t-1}}\big]$ is a density matrix. Let $V$ be a unitary with $V\hat\mu=\hat\mu$. The distribution of $\xi$ is invariant under unitaries, so $Vx_{t-1}=\tilde\mu_t+\sqrt{\tilde\beta_t}\,V\xi$ has the same distribution as $x_{t-1}$, and $VMV^{\dagger}=M$. The choice $V=\ket{\hat\mu}\bra{\hat\mu}-(I-\ket{\hat\mu}\bra{\hat\mu})$ implies that $M$ has no block coupling $\operatorname{span}\{\hat\mu\}$ with its orthogonal complement $\hat\mu^{\perp}$. The choice $V=\ket{\hat\mu}\bra{\hat\mu}+W$, with an arbitrary unitary $W$ on $\hat\mu^{\perp}$, implies that the block of $M$ on $\hat\mu^{\perp}$ commutes with every unitary of $\hat\mu^{\perp}$, so it is a multiple of the identity. Hence
\begin{equation}
\begin{aligned}
M&=\lambda_1\ket{\hat\mu}\bra{\hat\mu}+\lambda_2\big(I-\ket{\hat\mu}\bra{\hat\mu}\big),\\
\lambda_1+(D-1)\lambda_2&=\operatorname{Tr}M=1.
\end{aligned}\label{eq:app_M}
\end{equation}
In this decomposition, $\lambda_1=\mathbb{E}\lvert\braket{\hat\mu|\hat x_{t-1}}\rvert^{2}$ and $\lambda_2=\mathbb{E}\lvert\braket{e|\hat x_{t-1}}\rvert^{2}$ for any unit vector $e\perp\hat\mu$. Substituting $M$ gives $\operatorname{Tr}(M\varrho)=\lambda_2+(\lambda_1-\lambda_2)\bra{\hat\mu}\varrho\ket{\hat\mu}$, which is Eq.~\eqref{eq:thm_target}.

\emph{Ordering of the constants.} Let $e_1=\hat\mu,e_2,\dots,e_D$ be an orthonormal basis and $y_j=\braket{e_j|x_{t-1}}$. The $y_j$ are independent, with $y_1=m+\sqrt{\tilde\beta_t}\,\xi_1$, $m=\lVert\tilde\mu_t\rVert_2>0$, and $y_j=\sqrt{\tilde\beta_t}\,\xi_j$ for $j\geq2$, where $\xi_j=\braket{e_j|\xi}$. With $A=\lvert y_1\rvert^{2}$, $B=\lvert y_2\rvert^{2}$, and $W=\sum_{j\geq3}\lvert y_j\rvert^{2}$,
\begin{equation}
\lambda_1-\lambda_2=\mathbb{E}\big[g(A)\big],\qquad g(a)=\frac{a-B}{a+B+W}.\label{eq:app_gap}
\end{equation}
In this expression, $g$ is bounded and strictly increasing in $a$ whenever $B+W>0$, which holds almost surely. Let $A'=\tilde\beta_t\lvert\xi_1\rvert^{2}$ be the value of $A$ for $m=0$. Since $A'$ and $B$ are independent, identically distributed, and independent of $W$, exchangeability gives $\mathbb{E}[g(A')]=0$. We write $\xi_1=u+iv$ with independent $u,v\sim\mathcal{N}(0,1)$ and $\delta=m/\sqrt{\tilde\beta_t}>0$. For every $s>0$, the function $\delta\mapsto\Pr(\lvert u+\delta\rvert\leq s)=\Phi(s-\delta)-\Phi(-s-\delta)$ has the derivative $\varphi(s+\delta)-\varphi(s-\delta)<0$ for $\delta>0$, where $\Phi$ and $\varphi$ denote the standard normal distribution and density. Conditionally on $v$, this gives $\Pr(A\leq c\mid v)<\Pr(A'\leq c\mid v)$ for every $c>\tilde\beta_t v^{2}$, so $\Pr(A\leq c)<\Pr(A'\leq c)$ for every $c>0$, and $A$ strictly dominates $A'$ stochastically. Since $(B,W)$ is independent of $A$ and $A'$, integration by parts gives $\mathbb{E}[g(A)]-\mathbb{E}[g(A')]=\mathbb{E}\int_0^{\infty}g'(a)\big(\Pr(A'\leq a)-\Pr(A\leq a)\big)\,da>0$, that is, $\lambda_1>\lambda_2$.
\end{proof}

\subsection{Proof of Proposition~\ref{prop:init}}

\begin{proof}
Every angle $\theta$ of a block enters a single rotation $e^{-i\theta Q/2}=\cos(\theta/2)\,I-i\sin(\theta/2)\,Q$ with a Pauli operator $Q$, and the angles are independent. For $\theta\sim\mathcal{N}(0,\sigma^2)$, $\mathbb{E}\cos(\theta/2)=e^{-\sigma^2/8}$ and $\mathbb{E}\sin(\theta/2)=0$, so $\mathbb{E}\big[e^{-i\theta Q/2}\big]=e^{-\sigma^2/8}I$. By independence, the expectation of each layer $R(\Theta_l)$, a product and tensor product of $3m$ such rotations, equals $e^{-3m\sigma^2/8}I$, and the expectation of the product in Eq.~\eqref{eq:sec} equals $e^{-N\sigma^2/8}E_{L-1}\cdots E_0=e^{-N\sigma^2/8}I$, since the chains cancel in pairs ($C^{\dagger}C=I$) for even $L$. The same cancellation with $\Theta=0$ gives $U(0)=I$. For a state $\ket{v}$, $\lVert U\ket{v}-\ket{v}\rVert_2^{2}=2-2\operatorname{Re}\bra{v}U\ket{v}$, whose expectation is $2(1-e^{-N\sigma^2/8})$.

For the ancilla, $(\bra{1}_a\otimes I)(\ket{0}_a\otimes\ket{y})=0$ for every $\ket{y}$. With $\ket{y}=U_1^{(k,t)}\ket{x_t}$, this gives
\begin{equation}
\begin{aligned}
p_1&=\Big\lVert(\bra{1}_a\otimes I)\big(U_2^{(t)}-I\big)\big(\ket{0}_a\otimes\ket{y}\big)\Big\rVert_2^{2}\\
&\leq\Big\lVert\big(U_2^{(t)}-I\big)\big(\ket{0}_a\otimes\ket{y}\big)\Big\rVert_2^{2}.
\end{aligned}\label{eq:app_p1}
\end{equation}
In this bound, the right-hand side has the expectation $2(1-e^{-N_2\sigma^2/8})$ over the angles of $U_2^{(t)}$, which are independent of $\ket{y}$. Hence $\mathbb{E}[p_0\mid\ket{x_t},\ket{y}]\geq c:=2e^{-N_2\sigma^2/8}-1$ at every step. During sampling, the input of step $t$ and the block $U_1^{(k,t)}$ do not depend on $\theta_2[t]$, so conditioning successively on the steps $t=1,2,\dots,T$ gives $\mathbb{E}[P_{\mathrm{acc}}]=\mathbb{E}\big[\prod_t p_0^{(t)}\big]\geq c^{T}$ when $c\geq0$.
\end{proof}

\subsection{Smoothness and convergence of training}\label{app:smooth}

The first supplementary result bounds the derivatives of the losses and the resulting convergence of gradient-based training. Convergence analyses of neural networks often assume a convex or strongly convex objective, which deep models do not satisfy; the bounds below follow from the circuit itself.

\begin{lemma}[Smoothness and convergence]\label{lem:smooth}
Let $\ell(\theta)$ denote the loss of one sample $(x_0,k,t,\epsilon)$ as a function of the $P$ angles, and let $\partial_i$ denote the derivative with respect to the $i$-th angle.
(i) Without branch selection, $\lvert\partial_i\ell_{\mathrm{tr}}\rvert\leq\frac12$ and $\lvert\partial_i\partial_j\ell_{\mathrm{tr}}\rvert\leq\frac12$ for all $i,j$, so that $\lVert\nabla\ell_{\mathrm{tr}}\rVert_2\leq G=\sqrt{P}/2$ and $\nabla\ell_{\mathrm{tr}}$ is $\Lambda$-Lipschitz with $\Lambda=P/2$.
(ii) With post-selection, wherever $p_0\geq p_{\min}>0$, $\lvert\partial_i\ell_{\mathrm{ps}}\rvert\leq1/p_{\min}$ and $\lvert\partial_i\partial_j\ell_{\mathrm{ps}}\rvert\leq2/p_{\min}^{2}$, so that $G=\sqrt{P}/p_{\min}$ and $\Lambda=2P/p_{\min}^{2}$.
(iii) Let stochastic gradient descent minimize $\mathcal{L}$ of Eq.~\eqref{eq:objective} with mini-batches of $B$ independent samples and a constant step $\eta\leq1/\Lambda$, and, in case (ii), let $p_0\geq p_{\min}$ hold for every sample along the iterates. Then, after $J$ steps,
\begin{equation}
\frac{1}{J}\sum_{j=0}^{J-1}\mathbb{E}\big\lVert\nabla\mathcal{L}(\theta_j)\big\rVert_2^{2}\leq\frac{2}{\eta J}+\frac{\eta\,\Lambda\,G^{2}}{B}.\label{eq:thm_conv}
\end{equation}
\end{lemma}
In this bound, the first term represents the decrease of the gradient norm with the number of steps, and the second term denotes the floor caused by the variance of the stochastic gradients. With $\eta=\min\{1/\Lambda,\sqrt{2B/(\Lambda G^{2}J)}\}$, the bound decreases as $O(1/\sqrt{J})$. The lemma requires neither convexity nor assumptions on the data, because the constants $\Lambda$ and $G$ follow from the circuit itself: the parameter-shift rule \citep{mitarai2018quantum,schuld2019evaluating} makes every loss a trigonometric polynomial of degree one in each angle, and the loss lies in $[0,1]$. Proposition~\ref{prop:init} indicates that $p_0$ is close to one at initialization, and we monitor its mean on the validation set during training; beyond initialization, condition (ii) remains an assumption. The bound covers stochastic gradient descent; for Adam, which we employ, comparable rates hold under the same smoothness and bounded-gradient conditions \citep{defossez2022simple}. Since the constants hold for the worst case, the step size that part (iii) requires for the post-selected circuit, $\eta\leq1/\Lambda\leq1.7\times10^{-4}$, lies far below our learning rate of $0.01$, so that the lemma bounds the smoothness of the objective without describing the course of our training runs.

\begin{proof}
\emph{Structure of the losses.} Without branch selection,
\begin{equation}
\begin{aligned}
1-\ell_{\mathrm{tr}}={}&\operatorname{Tr}\Big[\big(I_a\otimes U_3^{(t)\dagger}\ket{\tilde\mu_t}\bra{\tilde\mu_t}U_3^{(t)}\big)\\
&U_2^{(t)}\big(\ket{0}\bra{0}_a\otimes U_1^{(k,t)}\ket{x_t}\bra{x_t}U_1^{(k,t)\dagger}\big)U_2^{(t)\dagger}\Big].
\end{aligned}\label{eq:app_ltr}
\end{equation}
In this expression, the trace is the probability of an outcome of a measurement, so $\ell_{\mathrm{tr}}\in[0,1]$, and every angle appears once in a unitary and once in its adjoint. Writing each rotation as $\cos(\theta_i/2)I-i\sin(\theta_i/2)Q$ implies that $1-\ell_{\mathrm{tr}}$ is affine in $(\cos\theta_i,\sin\theta_i)$ for every angle $\theta_i$. With post-selection, $\ell_{\mathrm{ps}}=1-g/p_0$, where $g=\lvert\bra{\tilde\mu_t}U_3^{(t)}(\bra{0}_a\otimes I)U_2^{(t)}(\ket{0}_a\otimes U_1^{(k,t)}\ket{x_t})\rvert^{2}$ and $p_0$ are probabilities of the same form, and $0\leq g\leq p_0\leq1$ by the Cauchy--Schwarz inequality.

\emph{Derivative bounds.} Let $f\in[0,1]$ be a function of this form, let $e_i$ be the $i$-th unit vector, and let $s=\pi/2$. The parameter-shift rule \citep{mitarai2018quantum,schuld2019evaluating} gives, for $i\neq j$,
\begin{equation}
\begin{aligned}
\partial_if&=\tfrac12\big[f(\theta+se_i)-f(\theta-se_i)\big],\\
\partial_i\partial_jf&=\tfrac14\big[f_{++}-f_{+-}-f_{-+}+f_{--}\big],\\
\partial_i^{2}f&=-(f-a_i).
\end{aligned}\label{eq:app_shift}
\end{equation}
In these identities, $f_{\pm\pm}$ denotes $f$ evaluated at $\theta\pm se_i\pm se_j$, and $a_i$ is the mean of $f$ over $\theta_i$. Since $f\in[0,1]$, the first two expressions lie in $[-\frac12,\frac12]$, and $\lvert f-a_i\rvert\leq\frac12\big(\max_{\theta_i}f-\min_{\theta_i}f\big)\leq\frac12$ bounds the third. Applied to $f=1-\ell_{\mathrm{tr}}$, these bounds give part (i), with $\lVert\nabla\ell_{\mathrm{tr}}\rVert_2\leq\sqrt{P}/2$ and $\lVert\nabla^{2}\ell_{\mathrm{tr}}\rVert_2\leq\lVert\nabla^{2}\ell_{\mathrm{tr}}\rVert_F\leq P/2$. For part (ii), the same bounds for $g$ and $p_0$, together with $g\leq p_0$, give
\begin{equation}
\begin{aligned}
\Big\lvert\partial_i\frac{g}{p_0}\Big\rvert&\leq\frac{\lvert\partial_ig\rvert}{p_0}+\frac{g\,\lvert\partial_ip_0\rvert}{p_0^{2}}\leq\frac{1}{p_0},\\
\Big\lvert\partial_i\partial_j\frac{g}{p_0}\Big\rvert&\leq\frac{1}{p_0}+\frac{1}{p_0^{2}}\leq\frac{2}{p_0^{2}}.
\end{aligned}\label{eq:app_quot}
\end{equation}
In these bounds, the second inequality expands the second derivative of the quotient into four terms, which are bounded by $\frac{1}{2p_0}$, $\frac{1}{2p_0^{2}}$, $\frac{1}{2p_0}$, and $\frac{1}{2p_0^{2}}$, respectively.

\emph{Convergence.} The objective $\mathcal{L}$ averages the per-sample losses, so it inherits the bounds on the gradient and the Hessian and is $\Lambda$-smooth. For a mini-batch gradient $\hat g_j$ of $B$ independent samples, $\mathbb{E}[\hat g_j]=\nabla\mathcal{L}(\theta_j)$ and $\mathbb{E}\lVert\hat g_j-\nabla\mathcal{L}(\theta_j)\rVert_2^{2}\leq G^{2}/B$. The descent lemma for $\Lambda$-smooth functions and $\eta\leq1/\Lambda$ give
\begin{equation}
\mathbb{E}\,\mathcal{L}(\theta_{j+1})\leq\mathbb{E}\,\mathcal{L}(\theta_j)-\frac{\eta}{2}\,\mathbb{E}\big\lVert\nabla\mathcal{L}(\theta_j)\big\rVert_2^{2}+\frac{\Lambda\eta^{2}G^{2}}{2B}.\label{eq:app_descent}
\end{equation}
In this recursion, summing over $j=0,\dots,J-1$, using $\mathcal{L}(\theta_0)-\mathcal{L}(\theta_J)\leq1$ because $\mathcal{L}\in[0,1]$, and dividing by $\eta J/2$ yield Eq.~\eqref{eq:thm_conv} \citep{ghadimi2013stochastic}.
\end{proof}

\subsection{Shot complexity}\label{app:shots}

The second supplementary result bounds the number of measurements required on hardware for the read-out and for the training loss.

\begin{lemma}[Shot complexity]\label{lem:shots}
(i) Let $p_i=\lvert\braket{i|z_0}\rvert^{2}$, and let the counts of $S$ measurements of $\ket{z_0}$ in the computational basis follow the multinomial distribution with probabilities $p$. The read-out of Eq.~\eqref{eq:readout} satisfies
\begin{equation}
\mathbb{E}\sum_{i=0}^{D-1}\Big(\lvert\hat\alpha_i\rvert-\sqrt{p_i}\Big)^{2}\leq\frac{D-1}{S},\label{eq:prop_shots}
\end{equation}
and, for the same drawn length $\rho$, the decoded vectors of Eq.~\eqref{eq:decode} satisfy $\mathbb{E}\lVert\hat z_S-\hat z_\infty\rVert_2^{2}\leq\rho^{2}\Delta^{2}(D-1)/S$, with $\Delta=\max_i(h_i-l_i)$.
(ii) If the fidelity $F=\bra{\tilde\mu_t}\varrho_{\mathrm{out}}\ket{\tilde\mu_t}$ is estimated from $S$ swap tests as $\hat F=2\hat P_0-1$, where $\hat P_0$ is the fraction of zero outcomes, then $\Pr(\lvert\hat F-F\rvert\geq\varepsilon)\leq2e^{-S\varepsilon^{2}/2}$.
\end{lemma}
In this lemma, $\hat z_S$ and $\hat z_\infty$ represent the latent vectors decoded from $S$ shots and from the exact magnitudes. The read-out error decreases as $(D-1)/S$, which remains small for $D=32$: with $S=1{,}000$ shots, the mean squared error of the magnitudes is at most $0.031$. Part (ii) indicates that $S\geq2\ln(2/\delta)/\varepsilon^{2}$ swap tests estimate one term of the loss to accuracy $\varepsilon$ with probability at least $1-\delta$ \citep{hoeffding1963probability}, which quantifies the cost of evaluating the training objective on a quantum device.

\begin{proof}
(i) If $p_i=0$, then $c_i=0$ almost surely and the term vanishes. If $p_i>0$, then, with $\hat p_i=c_i/S$,
\begin{equation}
\begin{aligned}
\big(\sqrt{\hat p_i}-\sqrt{p_i}\big)^{2}&=\frac{(\hat p_i-p_i)^{2}}{\big(\sqrt{\hat p_i}+\sqrt{p_i}\big)^{2}}\leq\frac{(\hat p_i-p_i)^{2}}{p_i},\\
\mathbb{E}(\hat p_i-p_i)^{2}&=\frac{p_i(1-p_i)}{S}.
\end{aligned}\label{eq:app_shots}
\end{equation}
In these relations, the variance stems from the binomial marginal distribution of $c_i$. Summing gives $\sum_{i:\,p_i>0}(1-p_i)/S\leq(D-1)/S$, which is Eq.~\eqref{eq:prop_shots} since $\lvert\hat\alpha_i\rvert=\sqrt{\hat p_i}$. The decoder of Eq.~\eqref{eq:decode} is affine in the magnitudes, $\hat z_S-\hat z_\infty=\rho\,(h-l)\odot\big(\lvert\hat\alpha\rvert-\sqrt{p}\big)$, which gives the bound on the decoded vectors.
(ii) The swap test returns $0$ with probability $P_0=(1+F)/2$ \citep{buhrman2001quantum}, by linearity also for the mixed state $\varrho_{\mathrm{out}}$, so $\hat F$ is unbiased, and Hoeffding's inequality \citep{hoeffding1963probability} gives $\Pr(\lvert\hat P_0-P_0\rvert\geq\varepsilon/2)\leq2e^{-S\varepsilon^{2}/2}$.
\end{proof}

\subsection{Proof of Proposition~\ref{prop:length}}

\begin{proof}
By Eq.~\eqref{eq:decode}, $(\hat z_i-l_i)/(h_i-l_i)=\rho\,\lvert\hat\alpha_i\rvert\geq0$, so the maximum with zero in Eq.~\eqref{eq:minmax} is inactive and $u(\hat z)=\rho\,\lvert\hat\alpha\rvert$. Both the exact magnitudes and the estimate of Eq.~\eqref{eq:readout} have unit norm, the latter because $\sum_ic_i=S$; hence $\rho(\hat z)=\lVert u(\hat z)\rVert_2=\rho$. Since $\rho_j$ and $\zeta$ are drawn independently of the generated state, the length of $\hat z$ does not depend on this state.
(i) The logarithm $\log\rho=\log\rho_j+\tau_k\zeta$ is the sum of a variable distributed uniformly over the elements of $\{\log r:r\in\mathcal{R}_k\}$ and an independent $\mathcal{N}(0,\tau_k^{2})$ variable. Its density is therefore the convolution of the two distributions, which is Eq.~\eqref{eq:kde}. The law of total expectation gives $\mathbb{E}[\log\rho]=\mathbb{E}[\log\rho_j]=\bar v_k$, and the law of total variance gives $\operatorname{Var}[\log\rho]=\operatorname{Var}[\log\rho_j]+\tau_k^{2}$, in which $\operatorname{Var}[\log\rho_j]=(1-1/\lvert\mathcal{R}_k\rvert)\,s_k^{2}$ is the variance of the empirical distribution.
(ii) For $\tau_k>0$, the distribution of $\log\rho$ has the density of Eq.~\eqref{eq:kde} and therefore assigns probability zero to every single value, so that $\Pr(\rho=r)=\Pr(\log\rho=\log r)=0$ for every $r>0$.
(iii) For $\tau_k=0$, $\rho=\rho_j$ follows the empirical distribution of $\mathcal{R}_k$. The $n_s$ lengths are independent draws from $\mathcal{U}(\mathcal{R}_k)$, and a value that occurs once in $\mathcal{R}_k$ is missed by all of them with probability $(1-1/\lvert\mathcal{R}_k\rvert)^{n_s}$.
For the sign-split codec, the same computation gives $u(\hat z)=\rho\,\big[\max(0,a^{+}-a^{-})\,;\,\max(0,a^{-}-a^{+})\big]$ with $a^{\pm}$ the two halves of $\lvert\hat\alpha\rvert$, whose norm is at most $\rho$, with equality when the halves have disjoint supports.
\end{proof}

\section{Numerical verification}\label{app:numerics}

We verify every statement of Section~\ref{sec:theory} with the PyTorch implementation of Section~\ref{sec:impl} and its default configuration ($n=5$, $K=17$, $T=10$, $L_1=L_2=L_3=4$, $\sigma=0.1$, linear schedule from $0.02$ to $0.10$). Table~\ref{tab:numerics} summarizes the results. For Theorem~\ref{thm:target}, we estimate $M$ from $10^{5}$ posterior draws for $t\in\{2,5,10\}$ and three random pairs $(x_t,x_0)$ per step, and we evaluate Eq.~\eqref{eq:thm_target} for $20$ random pure outputs per pair; the leading eigenvalue $\lambda_1$ lies between $0.65$ and $0.76$ for these pairs. For Proposition~\ref{prop:init}, we average $U(\Theta)$ over $4{,}000$ random initializations and run the post-selected sampler of Algorithm~\ref{alg:sample} for $64$ initial states under $20$ random initializations. For Lemma~\ref{lem:smooth}, we compute the exact gradient and Hessian with automatic differentiation for the $192$ angles used by one pair $(k,t)$, for $12$ random combinations of input, target, and angle scale ($0.1$, $1$, and $\pi$); the observed derivatives remain well below the bounds, which hold for the worst case. For Lemma~\ref{lem:shots}, we draw $20{,}000$ multinomial read-outs of a random state with $D=32$ and $200{,}000$ swap-test estimates. For Proposition~\ref{prop:length}, we decode $200$ random states per family with the codec fitted on our training set, compare the mean and the standard deviation of $10^{5}$ logarithms of lengths drawn with Eq.~\eqref{eq:decode} with part (i) for every family, and draw $512$ lengths $200$ times for three families of different sizes (Startpage, Expiro.BK, and FakeRean, with $97$, $184$, and $691$ training records), with $\tau_k=0$ and with the bandwidth of Eq.~\eqref{eq:bandwidth}. On the same draws, the membership test of Eq.~\eqref{eq:lenmia}, with the validation and test records of each family as non-members, reaches a mean AUC of $0.998$, $0.968$, and $0.767$ for these families with $\tau_k=0$, and of $0.552$, $0.513$, and $0.511$ with the bandwidth of Eq.~\eqref{eq:bandwidth}.

\begin{table*}[htbp]
\small
\caption{Numerical verification of Theorem~\ref{thm:target}, Propositions~\ref{prop:init} and \ref{prop:length}, and Lemmas~\ref{lem:smooth} and \ref{lem:shots} with the implementation of Section~\ref{sec:impl} ($n=5$, $K=17$, $T=10$, $L_1=L_2=L_3=4$, $\sigma=0.1$). The quantities $\overline{U}$, $p_1$, and $P_{\mathrm{acc}}$ are averaged over random initializations.}\label{tab:numerics}
\begin{tabular*}{\tblwidth}{@{}lp{6.7cm}ll@{}}
\toprule
Result & Quantity & Prediction & Observed\\
\midrule
Theorem~\ref{thm:target} & Overlap of the leading eigenvector of $M$ with $\ket{\tilde\mu_t}$ & $1$ & $0.99999$\\
 & Largest residual of Eq.~\eqref{eq:thm_target} over random outputs & $0$ & $3\times10^{-4}$\\
Proposition~\ref{prop:init} & $\lVert\overline{U}-e^{-N\sigma^2/8}I\rVert$ for $N=60$ / $72$ & $0$ / $0$ & $0.007$ / $0.015$\\
 & $\mathbb{E}[p_1]$ at initialization & $\leq0.172$ & $0.009$\\
 & $\mathbb{E}[P_{\mathrm{acc}}]$ at initialization & $\geq0.151$ & $0.911$\\
Lemma~\ref{lem:smooth} & $\max\lvert\partial_i\ell_{\mathrm{tr}}\rvert$ / $\max\lvert\partial_i\partial_j\ell_{\mathrm{tr}}\rvert$ & $\leq0.5$ / $\leq0.5$ & $0.079$ / $0.072$\\
 & $\max p_0\lvert\partial_i\ell_{\mathrm{ps}}\rvert$ / $\max p_0^{2}\lvert\partial_i\partial_j\ell_{\mathrm{ps}}\rvert$ & $\leq1$ / $\leq2$ & $0.079$ / $0.072$\\
Lemma~\ref{lem:shots} & Read-out error for $S=100$ / $1{,}000$ / $10{,}000$ & $\leq0.31$ / $0.031$ / $0.0031$ & $0.097$ / $0.0090$ / $0.0008$\\
 & $\Pr(\lvert\hat F-F\rvert\geq0.1)$ for $S=1{,}000$, $F=0.8$ & $\leq0.013$ & $1\times10^{-5}$\\
Proposition~\ref{prop:length} & $\max\lvert\rho(\hat z)-\rho\rvert$ & $0$ & $9\times10^{-16}$\\
 & Deviation of the mean / standard deviation of $\log\rho(\hat z)$ from part (i), in predicted standard deviations (maximum over $17$ classes) & $0$ / $0$ & $0.004$ / $0.004$\\
 & Reproduced training lengths with $512$ samples and $\tau_k=0$: Startpage / Expiro.BK / FakeRean ($\lvert\mathcal{R}_k\rvert=97$ / $184$ / $691$) & $0.995$ / $0.939$ / $0.524$ & $0.995$ / $0.938$ / $0.524$\\
 & The same with $\tau_k$ of Eq.~\eqref{eq:bandwidth} & $0$ / $0$ / $0$ & $0$ / $0$ / $0$\\
\bottomrule
\end{tabular*}
\end{table*}
\FloatBarrier

\bibliographystyle{elsarticle-num-names}
\bibliography{qdm-refs}

@inproceedings{sohl2015deep,
  author    = {Sohl-Dickstein, Jascha and Weiss, Eric and Maheswaranathan, Niru and Ganguli, Surya},
  title     = {Deep Unsupervised Learning using Nonequilibrium Thermodynamics},
  booktitle = {Proceedings of the 32nd International Conference on Machine Learning},
  series    = {Proceedings of Machine Learning Research},
  volume    = {37},
  pages     = {2256--2265},
  year      = {2015}
}

@inproceedings{ho2020denoising,
  author    = {Ho, Jonathan and Jain, Ajay and Abbeel, Pieter},
  title     = {Denoising Diffusion Probabilistic Models},
  booktitle = {Advances in Neural Information Processing Systems},
  volume    = {33},
  pages     = {6840--6851},
  year      = {2020}
}

@inproceedings{song2021denoising,
  author    = {Song, Jiaming and Meng, Chenlin and Ermon, Stefano},
  title     = {Denoising Diffusion Implicit Models},
  booktitle = {International Conference on Learning Representations},
  year      = {2021}
}

@article{cacioppo2023quantum,
  author  = {Cacioppo, Andrea and Colantonio, Lorenzo and Bordoni, Simone and Giagu, Stefano},
  title   = {Quantum Diffusion Models},
  journal = {arXiv preprint arXiv:2311.15444},
  year    = {2023}
}

@inproceedings{cacioppo2025hybrid,
  author    = {Cacioppo, Andrea and Colantonio, Lorenzo and Bordoni, Simone and Giagu, Stefano},
  title     = {Hybrid and Hardware-Oriented Approaches for Quantum Diffusion Models},
  booktitle = {2025 International Joint Conference on Neural Networks (IJCNN)},
  publisher = {IEEE},
  year      = {2025},
  doi       = {10.1109/IJCNN64981.2025.11227283}
}

@article{zhang2024generative,
  author  = {Zhang, Bingzhi and Xu, Peng and Chen, Xiaohui and Zhuang, Quntao},
  title   = {Generative Quantum Machine Learning via Denoising Diffusion Probabilistic Models},
  journal = {Physical Review Letters},
  volume  = {132},
  number  = {10},
  pages   = {100602},
  year    = {2024}
}

@article{schuld2020circuit,
  author  = {Schuld, Maria and Bocharov, Alex and Svore, Krysta M. and Wiebe, Nathan},
  title   = {Circuit-centric quantum classifiers},
  journal = {Physical Review A},
  volume  = {101},
  number  = {3},
  pages   = {032308},
  year    = {2020}
}

@book{schuld2021machine,
  author    = {Schuld, Maria and Petruccione, Francesco},
  title     = {Machine Learning with Quantum Computers},
  publisher = {Springer},
  address   = {Cham},
  year      = {2021}
}

@inproceedings{tran2026systematic,
  author    = {Tran, Ban Q. and Chu, Duong M. and Pham, Hai T. D. and Mengel, Susan},
  title     = {A Systematic Literature Review of Quantum Deep Learning in the {NISQ} Era},
  booktitle = {Proceedings of the 2026 15th International Conference on Software and Computer Applications (ICSCA 2026)},
  publisher = {ACM},
  address   = {Langkawi, Malaysia},
  pages     = {93--103},
  year      = {2026}
}

@inproceedings{tran2026comprehensive,
  author    = {Tran, Ban Q. and Chu, Duong M. and Pham, Hai T. D. and Nguyen, Viet Q. and Pham, Quan A. and Mengel, Susan},
  title     = {A Comprehensive Analysis of Accuracy and Robustness in Quantum Neural Networks},
  booktitle = {International Conference on Quantum Computing Infrastructure \& Technologies (QUANCOM 2026)},
  address   = {Trento, Italy},
  note      = {To appear; preprint arXiv:2604.26110},
  year      = {2026}
}

@inproceedings{goodfellow2014generative,
  author    = {Goodfellow, Ian and Pouget-Abadie, Jean and Mirza, Mehdi and Xu, Bing and Warde-Farley, David and Ozair, Sherjil and Courville, Aaron and Bengio, Yoshua},
  title     = {Generative Adversarial Nets},
  booktitle = {Advances in Neural Information Processing Systems},
  volume    = {27},
  pages     = {2672--2680},
  year      = {2014}
}

@inproceedings{abadi2016deep,
  author    = {Abadi, Martin and Chu, Andy and Goodfellow, Ian and McMahan, H. Brendan and Mironov, Ilya and Talwar, Kunal and Zhang, Li},
  title     = {Deep Learning with Differential Privacy},
  booktitle = {Proceedings of the 2016 ACM SIGSAC Conference on Computer and Communications Security},
  pages     = {308--318},
  year      = {2016}
}

@article{jiang2024netdiffusion,
  author  = {Jiang, Xi and Liu, Shinan and Gember-Jacobson, Aaron and Bhagoji, Arjun Nitin and Schmitt, Paul and Bronzino, Francesco and Feamster, Nick},
  title   = {{NetDiffusion}: Network Data Augmentation Through Protocol-Constrained Traffic Generation},
  journal = {Proceedings of the ACM on Measurement and Analysis of Computing Systems},
  volume  = {8},
  number  = {1},
  pages   = {11},
  year    = {2024}
}

@article{kwun2025mixed,
  author  = {Kwun, Gino and Zhang, Bingzhi and Zhuang, Quntao},
  title   = {Mixed-state quantum denoising diffusion probabilistic model},
  journal = {Physical Review A},
  volume  = {111},
  pages   = {032610},
  year    = {2025}
}

@inproceedings{carlini2023extracting,
  author    = {Carlini, Nicholas and Hayes, Jamie and Nasr, Milad and Jagielski, Matthew and Sehwag, Vikash and Tram{\`e}r, Florian and Balle, Borja and Ippolito, Daphne and Wallace, Eric},
  title     = {Extracting Training Data from Diffusion Models},
  booktitle = {32nd USENIX Security Symposium},
  pages     = {5253--5270},
  year      = {2023}
}

@inproceedings{stadler2022synthetic,
  author    = {Stadler, Theresa and Oprisanu, Bristena and Troncoso, Carmela},
  title     = {Synthetic Data -- Anonymisation Groundhog Day},
  booktitle = {31st USENIX Security Symposium},
  pages     = {1451--1468},
  year      = {2022}
}

@inproceedings{shokri2017membership,
  author    = {Shokri, Reza and Stronati, Marco and Song, Congzheng and Shmatikov, Vitaly},
  title     = {Membership Inference Attacks Against Machine Learning Models},
  booktitle = {IEEE Symposium on Security and Privacy},
  year      = {2017}
}

@article{anderson2018ember,
  author  = {Anderson, Hyrum S. and Roth, Phil},
  title   = {{EMBER}: An Open Dataset for Training Static {PE} Malware Machine Learning Models},
  journal = {arXiv preprint arXiv:1804.04637},
  year    = {2018}
}

@article{santos2013opcode,
  author  = {Santos, Igor and Brezo, Felix and Ugarte-Pedrero, Xabier and Bringas, Pablo G.},
  title   = {Opcode sequences as representation of executables for data-mining-based unknown malware detection},
  journal = {Information Sciences},
  volume  = {231},
  pages   = {64--82},
  year    = {2013}
}

@article{hu2017generating,
  author  = {Hu, Weiwei and Tan, Ying},
  title   = {Generating Adversarial Malware Examples for Black-Box Attacks Based on {GAN}},
  journal = {arXiv preprint arXiv:1702.05983},
  year    = {2017}
}

@inproceedings{xu2019modeling,
  author    = {Xu, Lei and Skoularidou, Maria and Cuesta-Infante, Alfredo and Veeramachaneni, Kalyan},
  title     = {Modeling Tabular Data using Conditional {GAN}},
  booktitle = {Advances in Neural Information Processing Systems},
  volume    = {32},
  year      = {2019}
}

@inproceedings{kotelnikov2023tabddpm,
  author    = {Kotelnikov, Akim and Baranchuk, Dmitry and Rubachev, Ivan and Babenko, Artem},
  title     = {{TabDDPM}: Modelling Tabular Data with Diffusion Models},
  booktitle = {Proceedings of the 40th International Conference on Machine Learning},
  series    = {Proceedings of Machine Learning Research},
  volume    = {202},
  pages     = {17564--17579},
  year      = {2023}
}

@inproceedings{rombach2022high,
  author    = {Rombach, Robin and Blattmann, Andreas and Lorenz, Dominik and Esser, Patrick and Ommer, Bj{\"o}rn},
  title     = {High-Resolution Image Synthesis with Latent Diffusion Models},
  booktitle = {Proceedings of the IEEE/CVF Conference on Computer Vision and Pattern Recognition},
  pages     = {10684--10695},
  year      = {2022}
}

@article{lloyd2018quantum,
  author  = {Lloyd, Seth and Weedbrook, Christian},
  title   = {Quantum Generative Adversarial Learning},
  journal = {Physical Review Letters},
  volume  = {121},
  pages   = {040502},
  year    = {2018}
}

@article{dallairedemers2018quantum,
  author  = {Dallaire-Demers, Pierre-Luc and Killoran, Nathan},
  title   = {Quantum generative adversarial networks},
  journal = {Physical Review A},
  volume  = {98},
  pages   = {012324},
  year    = {2018}
}

@article{zoufal2019quantum,
  author  = {Zoufal, Christa and Lucchi, Aur{\'e}lien and Woerner, Stefan},
  title   = {Quantum Generative Adversarial Networks for learning and loading random distributions},
  journal = {npj Quantum Information},
  volume  = {5},
  pages   = {103},
  year    = {2019}
}

@article{benedetti2019generative,
  author  = {Benedetti, Marcello and Garcia-Pintos, Delfina and Perdomo, Oscar and Leyton-Ortega, Vicente and Nam, Yunseong and Perdomo-Ortiz, Alejandro},
  title   = {A generative modeling approach for benchmarking and training shallow quantum circuits},
  journal = {npj Quantum Information},
  volume  = {5},
  pages   = {45},
  year    = {2019}
}

@article{mcclean2018barren,
  author  = {McClean, Jarrod R. and Boixo, Sergio and Smelyanskiy, Vadim N. and Babbush, Ryan and Neven, Hartmut},
  title   = {Barren plateaus in quantum neural network training landscapes},
  journal = {Nature Communications},
  volume  = {9},
  pages   = {4812},
  year    = {2018}
}

@article{parigi2024quantum,
  author  = {Parigi, Marco and Martina, Stefano and Caruso, Filippo},
  title   = {Quantum-Noise-Driven Generative Diffusion Models},
  journal = {Advanced Quantum Technologies},
  pages   = {2300401},
  year    = {2024}
}

@article{kolle2024quantum,
  author  = {K{\"o}lle, Michael and Stenzel, Gerhard and Stein, Jonas and Zielinski, Sebastian and Ommer, Bj{\"o}rn and Linnhoff-Popien, Claudia},
  title   = {Quantum Denoising Diffusion Models},
  journal = {arXiv preprint arXiv:2401.07049},
  year    = {2024}
}

@article{defalco2024latent,
  author  = {De Falco, Francesca and Ceschini, Andrea and Sebastianelli, Alessandro and Le Saux, Bertrand and Panella, Massimo},
  title   = {Quantum latent diffusion models},
  journal = {Quantum Machine Intelligence},
  volume  = {6},
  pages   = {85},
  year    = {2024}
}

@article{li2026quantum,
  author  = {Li, Guangxi and Jing, Mingrui and Zhu, Chenghong and Wang, Xin},
  title   = {Quantum latent diffusion model for high-dimensional state generation},
  journal = {Quantum Machine Intelligence},
  volume  = {8},
  pages   = {10},
  year    = {2026}
}

@article{mercaldo2022towards,
  author  = {Mercaldo, Francesco and Ciaramella, Giovanni and Iadarola, Giacomo and Storto, Michele and Martinelli, Fabio and Santone, Antonella},
  title   = {Towards Explainable Quantum Machine Learning for Mobile Malware Detection and Classification},
  journal = {Applied Sciences},
  volume  = {12},
  number  = {23},
  pages   = {12025},
  year    = {2022}
}

@misc{repacss,
  author       = {{REPACSS Team}},
  title        = {{REPACSS} User Guide: Architecture},
  howpublished = {Texas Tech University},
  url          = {https://guide.repacss.org/understanding/repacss-system/architecture.html},
  note         = {Accessed 26 October 2026},
  year         = {2026}
}

@misc{binutils,
  author       = {{Free Software Foundation}},
  title        = {{GNU Binutils}: objdump},
  url          = {https://www.gnu.org/software/binutils/},
  note         = {Accessed 26 October 2026},
  year         = {2026}
}

@misc{virusshare,
  author       = {{VirusShare}},
  title        = {{VirusShare}: A repository of malware samples},
  url          = {https://virusshare.com},
  note         = {Accessed 26 October 2026},
  year         = {2026}
}

@article{defalco2024quantum,
  author  = {De Falco, Francesca and Ceschini, Andrea and Sebastianelli, Alessandro and Le Saux, Bertrand and Panella, Massimo},
  title   = {Quantum hybrid diffusion models for image synthesis},
  journal = {KI - K{\"u}nstliche Intelligenz},
  volume  = {38},
  pages   = {311--326},
  year    = {2024}
}

@book{silverman1986density,
  author    = {Silverman, Bernard W.},
  title     = {Density Estimation for Statistics and Data Analysis},
  series    = {Monographs on Statistics and Applied Probability},
  volume    = {26},
  publisher = {Chapman and Hall},
  address   = {London},
  year      = {1986}
}

@article{silverman1987bootstrap,
  author  = {Silverman, Bernard W. and Young, G. Alastair},
  title   = {The bootstrap: To smooth or not to smooth?},
  journal = {Biometrika},
  volume  = {74},
  number  = {3},
  pages   = {469--479},
  year    = {1987}
}

@article{benedetti2019parameterized,
  author  = {Benedetti, Marcello and Lloyd, Erika and Sack, Stefan and Fiorentini, Mattia},
  title   = {Parameterized quantum circuits as machine learning models},
  journal = {Quantum Science and Technology},
  volume  = {4},
  number  = {4},
  pages   = {043001},
  year    = {2019}
}

@article{preskill2018quantum,
  author  = {Preskill, John},
  title   = {Quantum Computing in the {NISQ} era and beyond},
  journal = {Quantum},
  volume  = {2},
  pages   = {79},
  year    = {2018}
}

@article{mottonen2005transformation,
  author  = {M{\"o}tt{\"o}nen, Mikko and Vartiainen, Juha J. and Bergholm, Ville and Salomaa, Martti M.},
  title   = {Transformation of quantum states using uniformly controlled rotations},
  journal = {Quantum Information and Computation},
  volume  = {5},
  number  = {6},
  pages   = {467--473},
  year    = {2005}
}

@article{plesch2011quantum,
  author  = {Plesch, Martin and Brukner, {\v{C}}aslav},
  title   = {Quantum-state preparation with universal gate decompositions},
  journal = {Physical Review A},
  volume  = {83},
  number  = {3},
  pages   = {032302},
  year    = {2011}
}

@article{barenco1997stabilization,
  author  = {Barenco, Adriano and Berthiaume, Andr{\'e} and Deutsch, David and Ekert, Artur and Jozsa, Richard and Macchiavello, Chiara},
  title   = {Stabilization of quantum computations by symmetrization},
  journal = {SIAM Journal on Computing},
  volume  = {26},
  number  = {5},
  pages   = {1541--1557},
  year    = {1997}
}

@article{buhrman2001quantum,
  author  = {Buhrman, Harry and Cleve, Richard and Watrous, John and de Wolf, Ronald},
  title   = {Quantum Fingerprinting},
  journal = {Physical Review Letters},
  volume  = {87},
  number  = {16},
  pages   = {167902},
  year    = {2001}
}

@article{mitarai2018quantum,
  author  = {Mitarai, Kosuke and Negoro, Makoto and Kitagawa, Masahiro and Fujii, Keisuke},
  title   = {Quantum circuit learning},
  journal = {Physical Review A},
  volume  = {98},
  number  = {3},
  pages   = {032309},
  year    = {2018}
}

@article{schuld2019evaluating,
  author  = {Schuld, Maria and Bergholm, Ville and Gogolin, Christian and Izaac, Josh and Killoran, Nathan},
  title   = {Evaluating analytic gradients on quantum hardware},
  journal = {Physical Review A},
  volume  = {99},
  number  = {3},
  pages   = {032331},
  year    = {2019}
}

@book{nielsen2010quantum,
  author    = {Nielsen, Michael A. and Chuang, Isaac L.},
  title     = {Quantum Computation and Quantum Information: 10th Anniversary Edition},
  publisher = {Cambridge University Press},
  address   = {Cambridge},
  year      = {2010}
}

@inproceedings{kingma2015adam,
  author    = {Kingma, Diederik P. and Ba, Jimmy},
  title     = {Adam: A Method for Stochastic Optimization},
  booktitle = {International Conference on Learning Representations},
  year      = {2015}
}

@article{elfwing2018sigmoid,
  author  = {Elfwing, Stefan and Uchibe, Eiji and Doya, Kenji},
  title   = {Sigmoid-weighted linear units for neural network function approximation in reinforcement learning},
  journal = {Neural Networks},
  volume  = {107},
  pages   = {3--11},
  year    = {2018}
}

@article{gretton2012kernel,
  author  = {Gretton, Arthur and Borgwardt, Karsten M. and Rasch, Malte J. and Sch{\"o}lkopf, Bernhard and Smola, Alexander},
  title   = {A Kernel Two-Sample Test},
  journal = {Journal of Machine Learning Research},
  volume  = {13},
  pages   = {723--773},
  year    = {2012}
}

@article{esteban2017real,
  author  = {Esteban, Crist{\'o}bal and Hyland, Stephanie L. and R{\"a}tsch, Gunnar},
  title   = {Real-valued (Medical) Time Series Generation with Recurrent Conditional {GANs}},
  journal = {arXiv preprint arXiv:1706.02633},
  year    = {2017}
}

@article{park2018data,
  author  = {Park, Noseong and Mohammadi, Mahmoud and Gorde, Kshitij and Jajodia, Sushil and Park, Hongkyu and Kim, Youngmin},
  title   = {Data Synthesis based on Generative Adversarial Networks},
  journal = {Proceedings of the VLDB Endowment},
  volume  = {11},
  number  = {10},
  pages   = {1071--1083},
  year    = {2018}
}

@article{breiman2001random,
  author  = {Breiman, Leo},
  title   = {Random Forests},
  journal = {Machine Learning},
  volume  = {45},
  number  = {1},
  pages   = {5--32},
  year    = {2001}
}

@article{cortes1995support,
  author  = {Cortes, Corinna and Vapnik, Vladimir},
  title   = {Support-vector networks},
  journal = {Machine Learning},
  volume  = {20},
  number  = {3},
  pages   = {273--297},
  year    = {1995}
}

@article{pedregosa2011scikit,
  author  = {Pedregosa, Fabian and Varoquaux, Ga{\"e}l and Gramfort, Alexandre and Michel, Vincent and Thirion, Bertrand and others},
  title   = {Scikit-learn: Machine Learning in {P}ython},
  journal = {Journal of Machine Learning Research},
  volume  = {12},
  pages   = {2825--2830},
  year    = {2011}
}

@inproceedings{paszke2019pytorch,
  author    = {Paszke, Adam and Gross, Sam and Massa, Francisco and Lerer, Adam and Bradbury, James and Chanan, Gregory and others},
  title     = {{PyTorch}: An Imperative Style, High-Performance Deep Learning Library},
  booktitle = {Advances in Neural Information Processing Systems},
  volume    = {32},
  pages     = {8024--8035},
  year      = {2019}
}

@article{bergholm2018pennylane,
  author  = {Bergholm, Ville and Izaac, Josh and Schuld, Maria and Gogolin, Christian and Killoran, Nathan and others},
  title   = {{PennyLane}: Automatic differentiation of hybrid quantum-classical computations},
  journal = {arXiv preprint arXiv:1811.04968},
  year    = {2018}
}

@article{ghadimi2013stochastic,
  author  = {Ghadimi, Saeed and Lan, Guanghui},
  title   = {Stochastic First- and Zeroth-Order Methods for Nonconvex Stochastic Programming},
  journal = {SIAM Journal on Optimization},
  volume  = {23},
  number  = {4},
  pages   = {2341--2368},
  year    = {2013}
}

@article{defossez2022simple,
  author  = {D{\'e}fossez, Alexandre and Bottou, L{\'e}on and Bach, Francis and Usunier, Nicolas},
  title   = {A Simple Convergence Proof of {Adam} and {Adagrad}},
  journal = {Transactions on Machine Learning Research},
  year    = {2022}
}

@article{hoeffding1963probability,
  author  = {Hoeffding, Wassily},
  title   = {Probability Inequalities for Sums of Bounded Random Variables},
  journal = {Journal of the American Statistical Association},
  volume  = {58},
  number  = {301},
  pages   = {13--30},
  year    = {1963}
}

@article{bao2025generating,
  author  = {Bao, Tiffany and Trousil, Kylie and Tran, Quang Duy and Di Troia, Fabio and Park, Younghee},
  title   = {Generating Synthetic Malware Samples Using Generative {AI}},
  journal = {IEEE Access},
  volume  = {13},
  pages   = {59725--59736},
  year    = {2025}
}

@inproceedings{camerota2024addressing,
  author    = {Camerota, Chiara and Pappone, Lorenzo and Pecorella, Tommaso and Esposito, Flavio},
  title     = {Addressing Data Security in {IoT}: Minimum Sample Size and Denoising Diffusion Models for Improved Malware Detection},
  booktitle = {Proceedings of the 20th International Conference on Network and Service Management (CNSM)},
  year      = {2024}
}

@inproceedings{paim2025maldatagen,
  author    = {Paim, Kayua Oleques and Nogueira, Angelo Gaspar Diniz and Kreutz, Diego and Cordeiro, Weverton and Mansilha, Rodrigo Brandao},
  title     = {{MalDataGen}: A Modular Framework for Synthetic Tabular Data Generation in Malware Detection},
  booktitle = {Proceedings of the Brazilian Symposium on Cybersecurity (SBSeg)},
  year      = {2025}
}

@inproceedings{tran2022word,
  author    = {Tran, Quang Duy and Di Troia, Fabio},
  title     = {Word Embeddings for Fake Malware Generation},
  booktitle = {Silicon Valley Cybersecurity Conference (SVCC 2022)},
  series    = {Communications in Computer and Information Science},
  volume    = {1683},
  pages     = {22--37},
  publisher = {Springer},
  address   = {Cham},
  year      = {2022},
  doi       = {10.1007/978-3-031-24049-2_2}
}

@article{dockhorn2023differentially,
  author  = {Dockhorn, Tim and Cao, Tianshi and Vahdat, Arash and Kreis, Karsten},
  title   = {Differentially Private Diffusion Models},
  journal = {Transactions on Machine Learning Research},
  year    = {2023}
}

@inproceedings{jordon2019pate,
  author    = {Jordon, James and Yoon, Jinsung and van der Schaar, Mihaela},
  title     = {{PATE-GAN}: Generating Synthetic Data with Differential Privacy Guarantees},
  booktitle = {International Conference on Learning Representations (ICLR)},
  year      = {2019}
}

@inproceedings{fang2022dpctgan,
  author    = {Fang, Mei Ling and Dhami, Devendra Singh and Kersting, Kristian},
  title     = {{DP-CTGAN}: Differentially Private Medical Data Generation Using {CTGANs}},
  booktitle = {Artificial Intelligence in Medicine (AIME 2022)},
  series    = {Lecture Notes in Computer Science},
  publisher = {Springer},
  address   = {Cham},
  year      = {2022},
  doi       = {10.1007/978-3-031-09342-5_17}
}

@article{herr2021anomaly,
  author  = {Herr, Daniel and Obert, Benjamin and Rosenkranz, Matthias},
  title   = {Anomaly Detection with Variational Quantum Generative Adversarial Networks},
  journal = {Quantum Science and Technology},
  volume  = {6},
  number  = {4},
  pages   = {045004},
  year    = {2021},
  doi     = {10.1088/2058-9565/ac0d4d}
}

@inproceedings{dwork2006calibrating,
  author    = {Dwork, Cynthia and McSherry, Frank and Nissim, Kobbi and Smith, Adam},
  title     = {Calibrating Noise to Sensitivity in Private Data Analysis},
  booktitle = {Theory of Cryptography Conference (TCC 2006)},
  series    = {Lecture Notes in Computer Science},
  volume    = {3876},
  pages     = {265--284},
  publisher = {Springer},
  address   = {Berlin, Heidelberg},
  year      = {2006},
  doi       = {10.1007/11681878_14}
}

@inproceedings{yeom2018privacy,
  author    = {Yeom, Samuel and Giacomelli, Irene and Fredrikson, Matt and Jha, Somesh},
  title     = {Privacy Risk in Machine Learning: Analyzing the Connection to Overfitting},
  booktitle = {IEEE 31st Computer Security Foundations Symposium (CSF)},
  pages     = {268--282},
  year      = {2018},
  doi       = {10.1109/CSF.2018.00027}
}

@article{wilcoxon1945individual,
  author  = {Wilcoxon, Frank},
  title   = {Individual Comparisons by Ranking Methods},
  journal = {Biometrics Bulletin},
  volume  = {1},
  number  = {6},
  pages   = {80--83},
  year    = {1945},
  doi     = {10.2307/3001968}
}

@article{nappa2015malicia,
  author  = {Nappa, Antonio and Rafique, M. Zubair and Caballero, Juan},
  title   = {The {MALICIA} dataset: identification and analysis of drive-by download operations},
  journal = {International Journal of Information Security},
  volume  = {14},
  number  = {1},
  pages   = {15--33},
  year    = {2015}
}

\end{document}